\documentclass[11pt]{article}
\usepackage{amsmath, amsthm, amssymb}
\usepackage{makecell}
\usepackage{multirow}
\usepackage{graphicx}
\usepackage{enumerate}
\usepackage{natbib}
\usepackage{url} 
\usepackage{caption}
\usepackage{subcaption}
\usepackage{fancyvrb}
\usepackage{enumerate}
\usepackage{relsize}
\usepackage{commath}
\usepackage{etoolbox}
\usepackage{editing-macros}
\usepackage{statistics-macros}
\usepackage{stefan_tex}

\usepackage{wrapfig}

\usepackage{authblk}
\usepackage{booktabs}

\newcommand{\TOC}{\operatorname{TOC}}
\newcommand{\hTOC}{\widehat{\TOC}}

\newcommand{\AUTOC}{\operatorname{AUTOC}}

\theoremstyle{plain}
\newtheorem{prop}{Proposition}

\newtheorem{coro}[prop]{Corollary}
\newtheorem{lemm}[prop]{Lemma}
\newtheorem{theo}[prop]{Theorem}

\theoremstyle{definition}
\newtheorem{exam}{Example}
\newtheorem{defi}{Definition}

\theoremstyle{remark}

\newtheorem{rema}[prop]{Remark}

\newcommand\blfootnote[1]{%
  \begingroup
  \renewcommand\thefootnote{}\footnote{#1}%
  \addtocounter{footnote}{-1}%
  \endgroup
}

\def\spacingset#1{\renewcommand{\baselinestretch}%
{#1}\small\normalsize} \spacingset{1}

\AtBeginEnvironment{tabular}{\spacingset{1.0}}

\begin{document}

\title{\bf Evaluating Treatment Prioritization Rules via Rank-Weighted Average Treatment Effects}

\author[1]{Steve Yadlowsky}
\affil[1]{Google DeepMind}

\author[2]{Scott Fleming}
\affil[2]{Department of Biomedical Data Science, Stanford University}

\author[3]{Nigam Shah}
\affil[3]{Center for Biomedical Informatics Research, Stanford University}

\author[4]{Emma Brunskill}
\affil[4]{Department of Computer Science, Stanford University}

\author[5]{Stefan Wager}
\affil[5]{Graduate School of Business, Stanford University}

\maketitle

\begin{abstract}
There are a number of available methods for selecting whom to prioritize for treatment, including ones based on treatment effect estimation, risk scoring, and hand-crafted rules. We propose rank-weighted average treatment effect (RATE) metrics as a simple and general family of metrics for comparing and testing the quality of treatment prioritization rules. RATE metrics are agnostic as to how the prioritization rules were derived, and only assess how well they identify individuals that benefit the most from treatment. We define a family of RATE estimators and prove a central limit theorem that enables asymptotically exact inference in a wide variety of randomized and observational study settings. RATE metrics subsume a number of existing metrics, including the Qini coefficient, and our analysis directly yields inference methods for these metrics. We showcase RATE in the context of a number of applications, including optimal targeting of aspirin to stroke patients.
\end{abstract}

\spacingset{1.2}
\section{Introduction}

From\blfootnote{\hspace{-7mm}
SY and SF contributed equally to this research.
SF, NS, EB and SW acknowledge support from NHLBI grant R01HL144555.
All data access was performed by Stanford-affiliated co-authors; Google did not have access to the data. The analyses conducted with de-identified patient data were approved by the institutional panel on human subjects research at Stanford School of Medicine under protocol 46829.}
medicine to marketing, algorithms are commonly used to guide decision making around personalized interventions. Some  personalized intervention methods output a discrete intervention decision (often called a policy). 
However, the final choice of intervention for an individual is often informed both by predicted outcomes and additional considerations, such as cost or broader resource constraints. This motivates interest in approaches that output a score that ranks individuals in terms of intervention benefit, providing a   ``prioritization rule'' which can be used by decision makers, together with other contextual information, to design a final policy.

A conceptually clear approach to developing a prioritization rule is to fit a  heterogeneous treatment effect model that
estimates the benefit of an intervention for each individual based on their baseline covariates,
and then prioritizes individuals with the most intervention/treatment benefit. Recent methodological advances for estimating the conditional average treatment effect (CATE)
that can use data from a randomized (or appropriate observational) study allow data scientists to directly identify individuals most likely to benefit from
an intervention
\citep[e.g.,][]{hill2011bayesian,wager2018estimation,kunzel2019metalearners}. In principle, accurately estimated
CATE-based rules would be the gold standard for intervention selection and prioritization \citep{manski2004statistical}. However,
learning good CATE-based rules in practice requires either running a large randomized experiment or designing a comparable
quasi-experimental study, both of which involve considerable expense and effort.

A popular
alternative is to first estimate the baseline probability of the outcome in absence of any intervention, and then use this as a non-causal heuristic to prioritize individuals with a high baseline risk \citep{Kent2016risk, KentEtAl20}. 
There is considerable interest in understanding 
the extent to which risk-based rules are sufficient for making  high-quality individualized intervention decisions. For example, in medical settings, this may work well for preventative treatments where the treatment reduces the risk of 
an event occurring, as patients with a higher baseline risk have the most potential for benefit.
In such cases, the heterogeneous treatment effect is correlated with baseline risk, but
the amount of data available to build such a risk score may be much larger than that available
for treatment effect estimation \citep{KentEtAl20}, so a risk based approach may work better.

In this paper, we study {\it rank-weighted average treatment effect} (RATE) metrics, a suite of metrics
that can be used to quantify, estimate and run hypothesis tests about the value of a treatment prioritization
rule in a variety of statistical settings---including
with survival endpoints subject to censoring, which are common in medical settings.
At a high level, the RATE captures the extent to which individuals who are highly ranked by the
prioritization rule are more responsive to treatment than randomly selected individuals.
This approach is agnostic as to whether the rules were derived via CATE estimation or non-causal heuristics,
and only considers the extent to which the rules succeed in ranking units according to how much they
benefit from the intervention; thus, it can be used to compare risk-based and CATE-based rules on
a level playing field.

To enhance the utility of our results and methodology, we make publicly available the code used in our experiments as well as a \texttt{rank\_average\_treatment\_effect} function provided as part of the open-source \texttt{grf} package in R. Additional materials and information are provided as part of the \texttt{grf} package documentation.

\subsection{Motivating Application}
\label{sec:motivation}

To illustrate the need for a simple metric that can be used for testing the benefit of generic
prioritization rules, consider the case of stroke treatment using aspirin.
Stroke---a medical condition in which disruption of blood flow to part of the brain leads to
insufficient perfusion and brain cell death---is the third leading cause of death and disability worldwide.
Stroke is becoming more common over time, with a 70\% increase in stroke incidence globally from
1990 to 2019 \citep{feigin2021global}. Stroke will often lead directly to death \citep{saposnik2008stroke}; 
however, non-fatal cases can also have long-lasting debilitating impacts. More than a quarter
of previously independent patients remain dependent on a caregiver to complete activities of daily living 
a year after their stroke incident \citep{ullberg2015changes}. The  global cost of stroke, including
acute and long-term care, is estimated to be more than 1.12\% of the global GDP \citep{owolabi2021primary}.
There is significant interest in both preventing stroke and improving the efficacy
of treatments that mitigate its long-term consequences.

One commonly employed treatment for mitigating the long-term adverse impacts of stroke is acetylsalicylic acid, or Aspirin. Supported by randomized trial evidence that Aspirin significantly reduces the risk of death or recurrent stroke \citep{chen2000indications}, current clinical guidelines recommend the use of Aspirin for all patients without serious bleeding complications who are not already taking an anticoagulant or antiplatelet at stroke onset \citep{powers2019guidelines}. This guidance influences treatment for the majority of stroke patients, but it is not clear that the benefits of Aspirin are equally distributed. As an anti-platelet agent, aspirin can reduce platelet activation and consequently reduce the risk of clotting in the brain; however, this mechanism can also increase the risk of uncontrolled bleeding. The same trials that found aspirin to reduce the risk of death or dependence overall reported an increased risk of severe bleeding and hemorrhagic stroke \citep{chen1997cast, group1997international}. Although aspirin reduces the risk of death or long-term dependence for stroke patients on average, are there some patients for whom the harms of aspirin outweigh the benefits? How can we run a hypothesis test to validate treatment heterogeneity in this setting?

One key source of evidence in understanding benefits of aspirin to stroke patients is the
International Stroke Trial (IST) \citep{group1997international}. Comprising 19,435 patients
across 36 countries, the IST was one of the largest randomized clinical trials ever conducted
with acute stroke patients \citep{sandercock2011international}; see Section \ref{sec:ist} for
further discussion of the study. One of the headline findings of \citet{group1997international} was
that administering aspirin provided a statistically significant reduction in deaths after 6 months.

The evidence regarding treatment heterogeneity, however, is more ambiguous. A number of
studies have found suggestive evidence that patients with poor initial prognosis (i.e.,
with high risk of death regardless of treatment) may benefit less from aspirin.
In the IST trial, \citet{group1997international} considered treatment effects (on death at 6 months)
separately on subgroups stratified by an ``overall prognostic index'', divided into three bins,
``Good'', ``Average'' and ``Poor'', and found a  suggestive but non-statistically significant
trend of poor-prognosis patients experiencing less benefit from Aspirin relative to good-prognosis patients.
\citet{chen2000indications} obtained similar results when merging data from the IST with a related trial,
the Chinese Acute Stroke Trial. 
In Figure \ref{fig:IST:CATE-vs-risk}, we show results from our replication of this type of analysis on the IST
dataset: We randomly split the data in two, trained a random forest to estimate prognosis on the first half,
and report average treatment effects by quintiles of the risk score on the second half. Good prognosis groups have
more beneficial (i.e., more negative) point estimates, but all confidence intervals overlap.

\begin{wrapfigure}{R}{0.45\textwidth}
\centering
\includegraphics[width=0.42\textwidth]{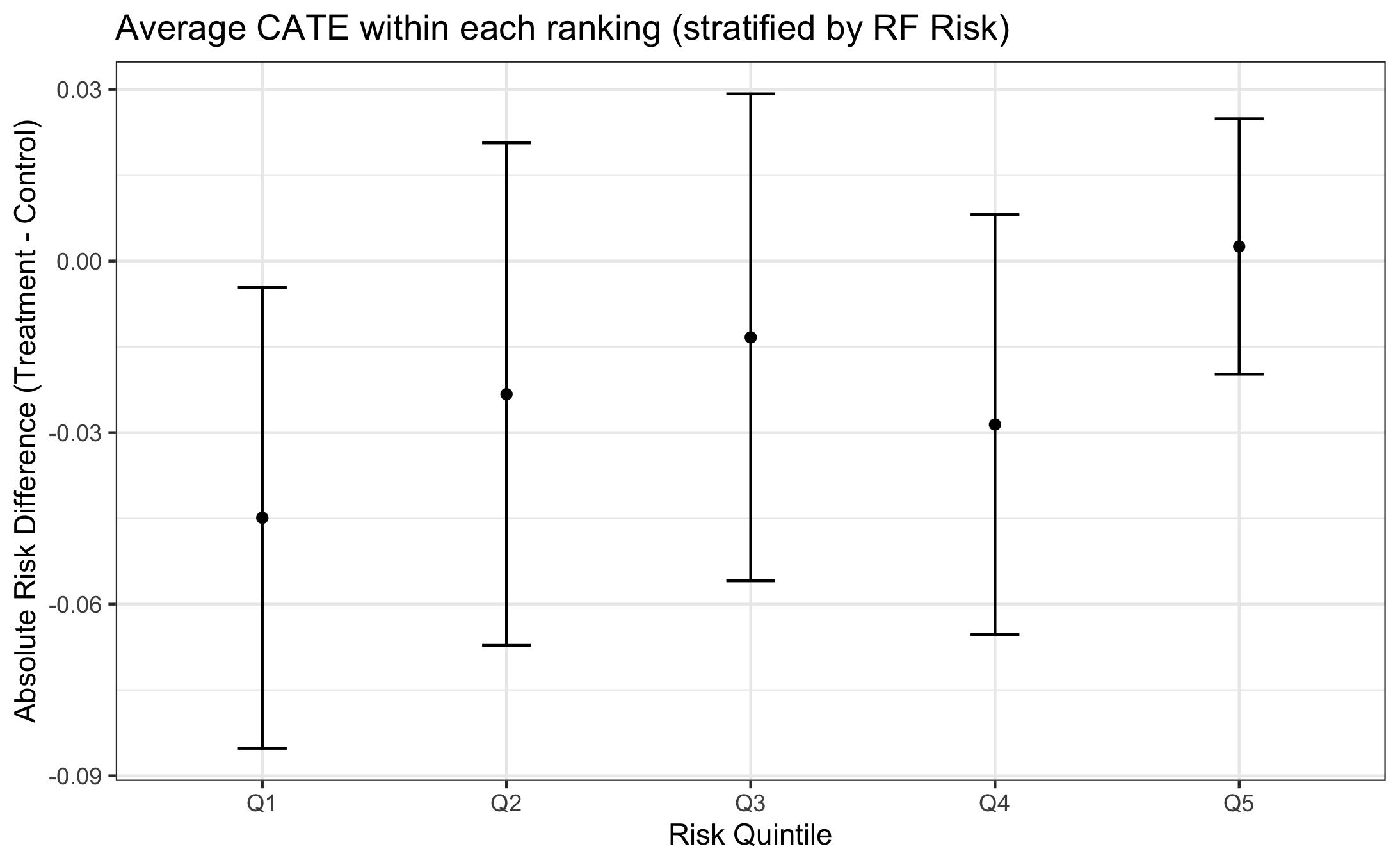}
    \caption{Treatment effect stratified by predicted risk within a 50\% held out test set on the International Stroke Trial data \citep{group1997international}.\label{fig:IST:CATE-vs-risk}}
\end{wrapfigure}

The findings of \citet{group1997international} and \citet{chen2000indications}, as illustrated
in Figure \ref{fig:IST:CATE-vs-risk}, leave us with many open questions. What is a high power way
to test for treatment heterogeneity here---ideally without dividing the data into an arbitrary
number of subgroups? Is a prognostic index the right way to express heterogeneity here, or could
we have done better by directly estimating the CATE? Does the IST show evidence that there are
some patients who benefit less from aspirin than others? In Section \ref{sec:ist}, we will revisit
this setting, and argue that analysis using the RATE can help shed light on these questions.
More broadly, we will show that RATE can be used as a tool in many other settings, and provide additional concrete examples in uplift modeling and antihypertensive treatment evaluation.

\subsection{Related Work}

Methodologically, our paper fits into a growing literature on evaluating estimators of heterogeneous treatment effects. One of the popular metrics for this purpose used in the marketing literature is the Qini curve,
which plots an estimate of the cumulative gain obtained by treating a growing fraction of units (as ranked by a prioritization
rule) against the fraction of units treated \citep{Radcliffe07}. The Qini curve is often summarized by the area under the curve, called the Qini coefficient. Recently, \citet{ImaiL19} studied statistical properties of
the Qini curve under randomization inference, while
\citet{sun2021treatment} considered analogous metrics in a setting where the cost of treating units is unknown and
may vary across units. Meanwhile, in the statistics literature, \citet{ZhaoTiCaClWe13} considered a different
area-under-the-curve for evaluating treatment rules. More broadly, \citet{chernozhukov2018generic} advocate assessing
treatment rules by evaluating the  average treatment effects across  different quantiles of estimated treatment effects.

In the context of this line of work, the main methodological contributions of our paper are as follows.
First, we unify a variety of evaluation metrics under the umbrella of RATE metrics, including the Qini coefficient from \cite{Radcliffe07} and the one proposed by \cite{ZhaoTiCaClWe13}. Within this unified context it is easier to compare and contrast the metrics in terms of their power and other statistical properties. Second,
we develop estimators for RATE metrics that can be used in a number of statistical
settings, including randomized studies and observational studies under an unconfoundedness
assumption \citep{rosenbaum1983central}, and that can accommodate potentially
censored survival endpoints. Third, we prove a central limit theorem that applies to a
wide class of RATE metrics and estimators, and that can be used to justify resampling-based
confidence intervals. Our formal results blend recent advances in doubly-robust estimation of
heterogeneous treatment effects 
\citep{athey2021policy,ChernozhukovChDeDuHaNeRo18,kennedy2023towards,semenova2017debiased}
with classical results on $L$-statistics \citep{ShorackWe09}.
We also note that, even in the simplest setting of Qini curves in randomized trials, our results yield the
first central limit theorem and associated exact Gaussian confidence intervals
(\citet{ImaiL19} provide bias and variance bounds under the Neyman model, but do not prove a
central limit theorem).

\section{The Rank-Weighted Average Treatment Effect}
\label{sec:RATE}

We assume access to $i = 1,\ldots,n$ independent and identically
distributed samples $(X_i, \, Y_i, \, W_i) \in \xx \times \RR \times \cb{0, \, 1}$, where $W_i$ denotes treatment
assignments, $Y_i$ is the outcome of interest, and the $X_i$ are auxiliary covariates. In some settings, we
may need to consider a larger number of observed variables, e.g., censoring times for survival analysis;
see Section \ref{sec:scoring-rules} for further discussion. Following the Neyman--Rubin model  \citep{imbens2015causal}, we posit potential outcomes
$\cb{Y_i(0), \, Y_i(1)}$ such that we observe $Y_i = Y_i(W_i)$, and we interpret $Y_i(1) - Y_i(0)$ as the
effect of the treatment on the $i$-th unit. We will frequently write the conditional average treatment effect (CATE) as
\begin{equation}
\label{eq:CATE}
\tau(x) = \EE{Y_i(1) - Y_i(0) \cond X_i = x}.
\end{equation}
where $x$ could be a scalar or vector of covariates, consistent with the dimensions of $X_i$. Given this notation, we now define
a prioritization rule and its associated targeting operator characteristic curve, characterizing how different the average treatment effect (ATE) is in units above each quantile of the prioritization rule from the ATE in the entire population.

\begin{defi}
A prioritization rule is defined in terms of a priority scoring function $S : \xx \rightarrow \RR$, such
that samples $i = 1, \, \ldots, \, n$ are prioritized in order $j = 1, \, \ldots, \, n$ for treatment in decreasing order of $S(X_i)$ (that is, a larger value of $S(X_i)$ implies that the sample should be treated first).
We let $i(j)$ denote the mapping from the rank $j$ to the sample index $i$.
\end{defi}

The priority scoring function $S$ is a user-provided mapping, and our methods are agnostic to the way that it is selected: $S$ could be a learned priority rule, a learned CATE estimate, a hard-coded heuristic, or something else.
Throughout our analysis, we treat the prioritization rule $S(\cdot)$ as fixed and deterministic; our $n$ samples will only be used to evaluate $S(\cdot)$ (i.e., they act as a test set). If $S(\cdot)$ is learned from data, our results should be understood as conditional on $S(\cdot)$.

Let $F_S(\cdot)$ be the cumulative distribution function (CDF) of $S(X_i)$. For now, we will assume that there are no ties in the prioritization rule $S$, so that $F_S$ is continuous; ties can be broken by randomization, as discussed in the Supplementary Materials Section A.2.

\begin{defi}
For any rule with priority score $S(\cdot)$ and any threshold $0 < u \leq 1$ such that $F_S^{-1}(u)$ exists, the targeting operator characteristic (TOC) is
\begin{equation}
\label{eq:toc}
\TOC(u; \, S) \defeq 
\EE{Y_i(1) - Y_i(0) \cond F_S(S(X_i)) \ge 1-u} -  \EE{Y_i(1) - Y_i(0)}.
\end{equation}
\end{defi}
\noindent Note if $u=1$, the first term is the average treatment effect, and so $TOC(u; \,S) = 0$.

Our main proposal involves evaluating prioritization rules in terms of weighted averages of the
TOC, which we refer to as rank-weighted average  treatment effects (RATEs).
RATEs only depend on the priority score $S(\cdot)$ via its induced priority ranking, and
not via the values of $S(X_i)$ for individual units.
After defining the RATE, we show below that a number of evaluation metrics considered in the literature are in
fact special cases of a RATE.

\begin{defi}
\label{def:rate}
For any weight function $\alpha:(0, \, 1] \rightarrow \RR$, the induced rank-weighted average treatment effect (RATE)
of a priority score $S(\cdot)$ is
$\theta_\alpha(S) = \int_0^1 \alpha(u) \TOC(u; \, S) \ \dif{u}$.
\end{defi}

\begin{rema}
Note that if $S(X_i) \indep Y_i(1) - Y_i(0)$, then both the TOC and any RATE will be identically 0.
Furthermore, if a priority score is monotone predictive of treatment effects in the sense that
\smash{$f(s) = \mathbb{E}[Y_i(1) - Y_i(0) \cond S(X_i) = s]$} is non-decreasing in $s$, then
$\TOC(u; \, S)$ is non-negative and non-increasing for $u \in (0,\, 1)$, and any RATE
with a non-negative weighting function will be non-negative.
\end{rema}

\begin{exam}[high-vs-others]
A very simple way to assess the effectiveness of a prioritization rule is by comparing the ATE for the top $u$-th of
units prioritized by the rule to the overall ATE. This is a RATE where $\alpha(u)$ is a point mass at $u$ and $\theta(S) = \TOC(u; \, S)$.
\end{exam}

\begin{exam}[AUTOC]
One limitation of the high-vs-other metric is that it is focused on comparisons at a specific quantile $F_S^{-1}(u)$ that needs
to be selected. To get a quantile-agnostic performance
measure, consider the area under the TOC curve, 
\begin{equation*}
    \operatorname{AUTOC}(S) = \int_0^1 \TOC(u; \, S) \ \dif{u}.
\end{equation*}
This metric was also considered in \citet{ZhaoTiCaClWe13}.
\end{exam}

\begin{exam}[Qini]
The Qini coefficient is another possible ``under the curve'' metric.
The Qini curve \citep{Radcliffe07} is defined by evaluating cumulative benefits as we increase the treatment fraction
according to a prioritization rule, and the Qini coefficient measures the area under the Qini curve:
\begin{equation*}
\label{eq:Qini}
\operatorname{QINI}(S) = \int_{0}^1 \EE{1\p{\cb{F_S(S(X_i)) \geq 1-u}} \p{Y_i(1) - Y_i(0)}} - u\EE{Y_i(1) - Y_i(0)} \ \dif{u}.
\end{equation*}
This is a RATE with linear weight function $\alpha(u) = u$, i.e., 
\smash{$\operatorname{QINI}(S) = \int_0^1 u \TOC(u; \, S) \ \dif{u}$}.
\end{exam}

\begin{exam}[AUPEC]
\citet{ImaiL19} proposed a modified Qini-coefficient, where no units are assigned to treatment if the priority score falls below a specified threshold $s^*$:
\begin{align*}
&\operatorname{AUPEC}(S; \, s^*) = \int_{0}^1 \EE{1\p{\cb{F_S(S(X_i)) \geq \max\cb{1 - u, \, F_S(s^*)}}} \p{Y_i(1) - Y_i(0)}} \nonumber \\
& \hspace{5cm} - u\EE{Y_i(1) - Y_i(0)} \ \dif{u}.
\label{eq:AUPEC}
\end{align*}
The AUPEC with score threshold $s^* \not= -\infty$ is not a RATE, because it also depends on the value of $S(X_i)$ as opposed to
the induced ranking only.
\end{exam}

\subsection{Weighted ATE Representation}
\label{sec:representation}
For analytic purposes, it is helpful to represent RATEs from Definition \ref{def:rate} as weighted averages of the individual treatment effects $Y_i(1) - Y_i(0)$, with weights depending on the quantile of $S(X_i)$.
For any weight function $\alpha(u)$ for a RATE metric as in Definition~\ref{def:rate}, let 
\begin{equation}
    w_a(t) = \int_{t}^1 \frac{\alpha(u)}{u} \ \dif{u} - \int_{0}^1 \alpha(u) \ \dif{u},
\end{equation}
assuming that these integrals exist and are finite. Proposition~\ref{prop:to-Lstat} shows that $\theta_\alpha(S)$ has a natural representation as an average treatment effect weighted by $w_\alpha$.

\begin{prop}
\label{prop:to-Lstat}
Let $\alpha(u)$ be the weight function for a RATE metric as in Definition~\ref{def:rate}. Assume that $\EE{Y_i(1) - Y_i(0) \cond X_i = x}$ is uniformly bounded. If $\alpha(u)$ is absolutely integrable on $[0, 1]$ and $\alpha(u)/u$ is absolutely integrable on $[t, 0]$ for any $t \in (0, 1)$, then $\theta_\alpha(S)$ can be written equivalently as
\begin{equation}
    \theta_{\alpha}(S) = \EE{ w_{\alpha}\p{1-F_S(S(X_i))}\p{Y_i(1) - Y_i(0)}}.
    \label{eq:as-Lstat}
\end{equation}
\end{prop}

This representation is valuable for a number of reasons. First, from Proposition \ref{prop:to-Lstat},
we see that the RATE metrics considered above have the following representations:
\begin{equation}
\label{eq:wrep}
\begin{split}
&\AUTOC\p{S} = \EE{\p{-\log\p{1 - F_S(S(X_i))} - 1} \p{Y_i(1) - Y_i(0)}}, \\
&\operatorname{QINI}\p{S} = \EE{\p{F_S(S(X_i)) - \frac12} \p{Y_i(1) - Y_i(0)}}.
\end{split}
\end{equation}
This gives further insight into the qualitative behavior of different RATE metrics. For example,
we see that the AUTOC measure strongly upweights treatment effects for the very first units prioritized by
$S(\cdot)$ (i.e., with $F_S(S(X_i)) \approx 1$), whereas the Qini coefficient considers the beginning and end of the ranking
given by $S(X_i)$ symmetrically.

Second, the representation \eqref{eq:as-Lstat} provides a natural starting point for a unified asymptotic analysis
of RATE metrics, allowing
us to leverage a large and well understood set of asymptotics results for $L$-statistics \citep{ShorackWe09}.\footnote{
\spacingset{1}\footnotesize
Strictly speaking, $\eta_w$ as defined in \eqref{eq:etaw}
is not an $L$-statistic because of the random multiplicative factor $Y_i(1) - Y_i(0)$. However, this will not
impede our application of standard results on $L$-statistics in proving a central limit theorem.}  We can then apply these results
to RATE metrics via Proposition \ref{prop:to-Lstat}.

Third, this representation gives us an alternative way to define a RATE. For any $w:(0,1) \to \R$ that satisfies appropriate regularity conditions on $w$ (a formal discussion is presented in the Supplementary Materials Section B.2, Assumptions D and E), any centered weighted ATE of the form
\begin{equation}
\label{eq:etaw}
\eta_w(S) = \EE{w\p{1 - F_S(S(X_i))}\p{Y_i(1) - Y_i(0)}}
\end{equation}
is also a RATE. By centered, we mean that $\int_0^1 w(u) \ \dif{u} = 0$. Then, we have that the weighted ATE in \eqref{eq:etaw} is equivalent to a RATE metric as defined in Definition~\ref{def:rate} with weights $\alpha_w(\cdot)$ defined as
$\eta_w(S) = \theta_{\alpha_w}(S)$ and $\alpha_w(t) = -t w'(1-t)$.

\subsection{Estimating the RATE}
\label{sec:estimation}

Our goal is to provide a general framework for designing RATE estimators that can be
applied in a wide variety of statistical settings, including observational studies and studies
with survival endpoints. To this end, we follow the approach used by \citet{semenova2017debiased} and \citet{athey2021policy}
for other methodological applications to provide similar generality. The main idea is to assume
the existence of ``scores'' \smash{$\hGamma_i$}  with the property that they act as nearly unbiased (but noisy)
proxies for the CATE \eqref{eq:CATE},
\begin{equation}
\label{eq:gamma_approx}
\EE{\hGamma_i \cond X_i} \approx \tau(X_i) = \EE{Y_i(1) - Y_i(0) \cond X_i}.
\end{equation}
The law of iterated expectation means that we can replace $Y_i(1) - Y_i(0)$ with such a score in the TOC definition \eqref{eq:toc} and weighted ATE representation of RATEs \eqref{eq:etaw}, as long as the approximation in \eqref{eq:gamma_approx} is good.
In the case of evaluating the effect of a treatment on non-survival outcomes in
a randomized controlled trial with randomization probability $\pi$, one simple choice
of scoring rule we could use is inverse-propensity weighting, i.e.,
$\hGamma_i = \frac{W_iY_i}{\pi} - \frac{(1 - W_i)Y_i}{1 - \pi}$.
In this case \eqref{eq:gamma_approx} holds exactly. In more complicated settings, however, constructing scores
\smash{$\hGamma_i$}  requires increased care and
involves estimation of nuisance components \citep{chernozhukov2016locally}. We defer a discussion
of how to construct scores \smash{$\hGamma_i$} and precise conditions of the type \eqref{eq:gamma_approx} to
Section \ref{sec:scoring-rules}, and for now take the availability of such scores as given.

Given this setting, we estimate the TOC and RATE by sample-averaging estimators. First,
for any $1/n \leq q \leq 1$, we estimate the TOC as
\begin{equation}
\label{eq:TOC_est}
\hTOC(u; \, s) = \frac{1}{\lfloor un \rfloor} \sum_{j = 1}^{\lfloor un \rfloor} \hGamma_{i(j)} - \frac{1}{n} \sum_{i = 1}^{n} \hGamma_{i}.
\end{equation}
This TOC estimator implies a natural RATE estimator for smooth
weight functions $\alpha(\cdot)$:
\begin{equation}
\label{eq:RATE_est}
\htheta_\alpha(S) = \frac{1}{n} \sum_{j = 1}^n \alpha\p{\frac{j}{n}} \hTOC\p{\frac{j}{n}; \, S}.
\end{equation}
The weighted ATE representation in \eqref{eq:as-Lstat} induces a natural estimator that does not require $\alpha(\cdot)$ to be smooth.
In this form, we can use \smash{$\hGamma_i$} directly in the empirical estimate
\begin{equation}
\label{eq:RATE_L_est}
\hat{\eta}_{w_\alpha}(S) = \frac{1}{n} \sum_{j = 1}^n w\p{\frac{j}{n}} \hGamma_{i(j)}.
\end{equation}

\subsection{Score construction}
\label{sec:scoring-rules}
One advantage of RATEs, and their estimators discussed in Section~\ref{sec:estimation}, is their generality in terms of causal estimation strategies. Deriving a score $\hGamma_i$ that approximately satisfies Eq.~\eqref{eq:gamma_approx} is nontrivial, because individual treatment effects depend on unobserved counterfactuals. Here, we summarize some useful scores from various causal estimation strategies.

Many of the scores derived below depend on quantities that must be estimated---we refer to them as nuisance parameters---thus motivating the use of the hat~\smash{$\,\what{\cdot}\,$}~in \smash{$\hGamma_i$}. It is useful to compare these scores to an \emph{oracle} score where the nuisance parameters are known a priori, which we denote $\Gamma_i^*$. We assume that the oracle score satisfies the condition~\eqref{eq:gamma_approx} exactly,
$ \EE{\Gamma_i^* \cond X_i} = \tau(X_i).$
Then, we can quantify the score approximation error as
    $\delta_i \defeq \hGamma_i - \Gamma_i^*$.
The properties of $\delta_i$ are useful for studying the estimation error and inference in Section~\ref{sec:asymptotics}.

\paragraph{Randomized Trials}
As mentioned in Section~\ref{sec:estimation}, randomized trials admit a simple IPW score that satisfies the condition~\eqref{eq:gamma_approx} exactly; however, these scores can have larger-than-necessary variance.
Using an Augmented IPW (AIPW) estimator  can help reduce the variance by adjusting for the baseline covariates $X_i$ measured at the beginning of the trial,
\begin{gather}
    \hGamma_i = \hat{m}(X_i, 1) - \hat{m}(X_i, 0) + \frac{W_i - \pi}{\pi\p{1 - \pi}} \p{Y_i - \hat{m}(X_i, W_i)}, \label{eq:rct-aipw}
\end{gather}
where $\hat{m}(x, w) \approx \E[Y_i(w) | X_i = x]$ is an estimate of the nuisance parameter $m(x, w)$ representing the expected outcome given a subject's covariates and treatment assignment $w \in \{0, 1\}$. The nuisance parameter must be estimated---using an appropriate parametric model, a nonparametric / machine learning estimator, or a simple surrogate such as the lagged outcome at the time of randomization. In a randomized trial, the treated proportion $\pi$ is known, so that $\E[W_i \cond X_i] = \pi$. Therefore, the estimated score $\hGamma_i$ from \eqref{eq:rct-aipw} satisfies the condition~\eqref{eq:gamma_approx} exactly whenever cross-fitting (discussed below) is used.

\paragraph{Observational Study with Unconfoundedness}

In this context, we estimate doubly robust scores for each participant using AIPW scores \citep{robins1994estimation}:
\begin{align}
    \hGamma_i = \hat{m}(X_i, 1) - \hat{m}(X_i, 0) + \frac{W_i - \hat{e}(X_i)}{\hat{e}(X_i)\p{1 - \hat{e}(X_i)}} \p{Y_i - \hat{m}(X_i, W_i)}, \label{eqn:dr-score-unconfoundedness}\\
    e(x) = P[W_i = 1 | X_i = x], \hspace{1cm}
    m(x, w) = \E[Y_i(w) | X_i = x] 
    \nonumber
\end{align}
where $e(x)$ represents the probability of an individual being assigned to treatment conditioned on observables; $m(x, w)$ represents the expected outcome given a subject's covariates and treatment assignment $w \in \{0, 1\}$; and $\hat{e}(x)$, $\hat{m}(x, w)$ represent nonparametric estimates of $e(x)$ and $m(x, w)$, respectively.
The oracle score would be
\begin{equation}
\label{eqn:oracle-aipw-score}
    \Gamma_i^* = m(X_i, 1) - m(X_i, 0) + \frac{W_i - e(X_i)}{e(X_i)\p{1 - e(X_i)}} \p{Y_i - m(X_i, W_i)}.
\end{equation}
This satisfies the condition~\eqref{eq:gamma_approx} exactly, because \smash{$\E[Y_i - m(X_i, W_i) \cond X_i, W_i] = 0$}, and $\tau(X_i) = m(X_i, 1) - m(X_i, 0).$

\paragraph{Time-to-event Outcomes}
We discuss an AIPW score for time-to-event outcomes with right-censoring  \citep{robins1994estimation,tsiatis2007semiparametric} in the Supplementary Materials A.3,
building on the presentation in \citet{cui2023estimating}. Discussing the functional form
of the scores is not possible here due to space constraints. More broadly, however,
the fact that we are able to immediately generalize our results to time-to-event
outcomes by drawing from existing results on doubly robust estimation in this
setting highlights the power and generality of our approach.

\paragraph{Cross-fitting}
All of the scores defined in this section require estimation of an unknown nuisance parameter function, which we will generically denote as $\eta(v)$, where $v$ is some generic set of arguments to the nuisance parameter. For example, for observational studies, $m(x, w)$ and $e(x)$ are unknown, and must be estimated.
To avoid ``own observation'' bias when fitting these nuisance parameters, we can split the sample into two components, one for estimating the nuisance parameters, and the other for applying the score functions. To improve efficiency, \citet{zheng2011cross} suggested splitting the data into $k$ folds, and repeating the sample splitting estimates of the nuisance parameters on each fold. This approach is discussed extensively in \citet{ChernozhukovChDeDuHaNeRo18} under the name ``cross-fitting.'' Using cross-fitting with the scores discussed in this section and appropriate machine learning estimators of the nuisance parameters often achieves the needed assumptions on $\delta_i$ in Assumption~\ref{assume:nuisance} to be introduced in  Section~\ref{sec:asymptotics}.

\section{Asymptotics and Inference}
\label{sec:asymptotics}

In order for the RATE to be a useful tool for assessing priority scoring rules, we need to be
able to use RATE estimates as the basis for hypothesis tests for comparing different scoring
rules. In this section, we will provide a central limit theorem for a large family of RATE estimators
that will enable confidence intervals and hypothesis tests via resampling-based methods.
Throughout this section, we will assume for simplicity that that $S(X_i)$ has no ties; however, all our
results extend immediately to the case with ties via the randomized tiebreaking procedure discussed in
Supplementary Materials Section~A.2.

In Sections~\ref{sec:representation} and~\ref{sec:estimation}, we discussed two representations for the RATE and their respective estimators. Except for technical issues related to regularity conditions, these representations are equivalent (see Prop.~\ref{prop:to-Lstat}). It is most convenient to state regularity
assumptions and prove results for statistics of the weighted ATE form \eqref{eq:etaw}. On the other hand, the estimators~\eqref{eq:RATE_est} and~\eqref{eq:RATE_L_est} are similar in spirit, albeit not precisely equivalent. However, both can be represented in the following generalized form of \eqref{eq:RATE_L_est},
\begin{equation}
    \what{\theta} = \frac{1}{n}\sum_{j=1}^n w_n\p{\frac{j}{n}} \what{\Gamma}_{i(j)},
    \label{eq:general_est}
\end{equation}
where $w_n$ is now a sequence of empirical weight functions that depends on the sample size $n$. As long as $w_n$ converges to $w$, in an appropriate sense discussed below, any estimator of the form \eqref{eq:general_est} will have similar statistical behavior to the exact plug-in estimator \eqref{eq:RATE_L_est}.

We show that the estimator $\what{\theta}$ is asymptotically linear under certain conditions on the data generating distribution, the weight function of the metric, and the nuisance parameter estimates used to construct $\what{\Gamma}_i$. It follows that the asymptotic distribution can be approximated via the half-sample bootstrap (subsampling without replacement) which, as discussed further in Appendix A.4, enables construction of confidence intervals and hypothesis tests. For notational convenience, we define $Q = S(X)$, $Q_i = S(X_i)$, and $\overline{\tau}(q) = \E[\tau(X_i) \mid Q_i = q]$.

We now define 3 assumptions that are important to our results: note that the constants must simultaneously satisfy the constraints of all the needed assumptions, which are separated for clarity of defining different components that contribute to the estimator $\what{\theta}$.
\begin{assumption}[Weight Regularity]
\label{assume:weight-fn}
There are almost surely no ties in the scores $S(X)$, so that $F_S(s)$ is continuous. The weights $w_n$ and $w$ are asymptotically similar, and each is sufficiently diffuse: 
$w$ and $w_n$ are both squared integrable on $(0,1)$, and there exists $M<\infty$ and $b<1$ satisfying the conditions of Assumption~\ref{assume:regularity} and Assumption~\ref{assume:nuisance}, so that $|w(t)| \le B(t) \defeq M(t (1-t))^{-b}$ for $b < 1$, $
\E\left[\frac{1}{n} \sum_{j=1}^n \p{w_n\p{j/n} - w(1- F_S(Q_{i(j)}))}^2\right] \to 0,$
and
$
\limsup_{n \to \infty} \sum_{j=1}^n \p{w_n\p{j/n} - w\p{j/n}}^p < \infty~\text{for some}~p > 2.
$
\end{assumption}

\begin{assumption}[Data Regularity]
\label{assume:regularity}
There exists $C_g < \infty$ such that $\var(\Gamma_i^\ast \mid X_i=x) \le C_g$.
$\overline{\tau}(q)$ is of bounded variation. $\E[ |\overline{\tau}(Q)|^{r}] < \infty$ for some $r > 2$ satisfying $1/r + b < 1/2$ and $1/r \le 1 - 1/p$ with $b$ and $p$ defined in Assumption~\ref{assume:weight-fn}.
\end{assumption}

To satisfy the variance condition for the IPW or AIPW score, a sufficient condition is that there exists $C_v < \infty$ such that $\sigma_w^2(x) \defeq \var(Y(w) \mid X=x) \le C_v$ and $\var(Y(w)) \le C_v$, and
there exists $0 < C_e < \infty$ such that $C_e \le P(W=1 \mid X=x) \le 1 - C_e$.

\begin{assumption}[Nuisance Parameter Convergence]
\label{assume:nuisance}
(a) The $\delta_i$ are (uniformly) randomly partitioned into $K \in \N$ (independently of $n$) sets $\{I_k\}_{k=1}^K$, each containing $n/K$ items, with $\{\delta_i\}_{i\in S_k}$ independent, conditionally on $B_k = \{(Y_i, W_i, X_i)\}_{i \not \in I_k}$ and an event $G_k$ with $P(G_k) \to 1$. (b) The bias $ \E[\sqrt{n} w(1 - F_S(Q_i)) \delta_i \mid G_k, B_k] \to 0$, and $q \mapsto \E[w(1 - F_S(Q_i)) \delta_i \mid Q_i = q, G_k, B_k]$ is of bounded variation. (c) The higher moments of $\delta_i$ satisfy $\E[(1+B^2(F(Q_i)) (t(1-t))^{-2\epsilon}) \delta_{i}^2 \mid Q_i=q, G_k, B_k] \le \min\{\zeta_n v(q), C\}$ with $\epsilon >0$ for some $v$ satisfying $\E[v(Q_i)] < \infty$, $\zeta_n \to 0$ and $C < \infty$, and finally $\E[\delta_i^r] < \infty$ for $r$ satisfying the same conditions as in Assumption~\ref{assume:regularity}.
\end{assumption}
This assumption is satisfied if (i) the IPW score is used with a known propensity score (such as in a randomized experiment), or (ii) the AIPW score is used with nuisance parameters that are known or estimated with cross-fitting \citep{ChernozhukovChDeDuHaNeRo18} satisfying $\|\what{\pi}(\cdot) - \pi(\cdot)\|_{m, P}\max\{\|\what{\mu}_1(\cdot) - \mu_1(\cdot)\|_{m, P}, \|\what{\mu}_1(\cdot) - \mu_1(\cdot)\|_{m, P}\} = o_P(n^{-1/2})$, and $\|\what{\pi}(\cdot) - 0\|_{\infty}$ and $\|\what{\pi}(\cdot) - 1\|_{\infty}$ are both bounded away from 0 with high probability. To satisfy Assumption~\ref{assume:nuisance}(c), a norm of $m=2$ is often sufficient
(eg., if $w(\cdot)$ is bounded), but if the weight function $w(\cdot)$ heavily
upweights certain regions of $Q$, then we may require $m > 2$ (e.g.,
via H\"older's inequality). Assumption~\ref{assume:nuisance}(a) is a sufficient condition to allow the approximation error terms to be uncorrelated within each partition by allowing one to condition on the nuisance parameters fit using cross-fitting, although it can be generalized so long as $\cov(\delta_i, \delta_j \mid G_k, B_k) = o(1/n)$ for $i, j \in I_k$. This allows generalization to leave-one-out cross-fitting schemes under sufficiently regular nuisance parameter estimation methods.

Under these assumptions, we give our main result, showing that the estimated RATE metric is asymptotically linear (and therefore asymptotically normal).

\begin{theo}
\label{thm:asymp-linear}
Under Assumptions \ref{assume:weight-fn}-\ref{assume:nuisance},
$
    \sqrt{n}(\what{\theta} - \theta) = \frac{1}{\sqrt{n}} \sum_{i=1}^n \psi_i + o_P(1),
$
with $\psi_i = w(1-F_S(S(X_i))) (\Gamma^\ast_i - \overline{\tau}(S(X_i))) + \int_{S(X_i)}^\infty w(1-F_S(q)) \dif{\overline{\tau}(q)} - \theta$, where $\var(\psi_1) < \infty$. Thus,
\begin{equation}
  \sqrt{n}(\what{\theta} - \theta) \cd \normal\left(0, \var(\psi_1)\right).  
\end{equation}
\end{theo}

There are two key ideas in Theorem~\ref{thm:asymp-linear}. The first is quantifying the variation that comes from weighting observations by estimated quantiles of the prioritization rule instead of the true quantiles. This we connect to the theory of $L$-statistics \citep{VanDerVaart98,ShorackWe09}, by noticing that
$
    (1/n) \sum_{j=1}^n w_n\p{\frac{j}{n}} \overline{\tau}(Q_i),
$ is an $L$-statistic.

The second is the use of the scores \smash{$\hGamma_i$} in place of \smash{$\overline{\tau}(Q_i)$} in the estimator \smash{$\what{\theta}$}. Such scores are helpful for estimating any parameter that is a linear functional of $\tau(X_i)$. By the law of iterated expectations, the same is true for linear functionals of \smash{$\overline{\tau}(Q_i)$}, so long as the weights only depend on $Q_i$. The representation of RATEs in \eqref{eq:as-Lstat} show that RATEs satisfy this requirement. The approach for replacing $\tau(X_i)$ with \smash{$\hGamma_i$} has been used for average treatment effect estimation \citep{ChernozhukovChDeDuHaNeRo18}, conditional average treatment effect estimation \citep{kennedy2023towards}, policy learning \citep{athey2021policy}, and structural functions of treatment effects \citep{semenova2017debiased}. The key to the proof of Theorem~\ref{thm:asymp-linear} in the Supplementary Materials Section B.2 is extending this technique to $L$-statistics. The statistical efficiency of the estimator depends greatly on the form of the score used; we discuss the practical implications of these choices in Section~\ref{sec:comparing-weighting-functions}.

The CLT implies that Wald confidence intervals based on a consistent estimator of $\var(\psi_1)$ would have asymptotically correct coverage. However, the variance term is cumbersome to understand and estimate in closed form. Thus, we pursue inference using resampling methods instead. Asymptotic linearity with an influence function with finite variance implies the
validity of a wide variety of inference methods, such as many boostrap resampling methods
for inference \citep{Mammen12}. Typical proofs (eg., via Hadamard differentiability as
in \citet{VanDerVaart98}) would require additional assumptions on the weight function
for RATE metrics. However, certain inference procedures based on half-sampling can
be justified via asymptotic linearity alone \citep{chung2013exact}. Building on this
observation, we here show that the half-sample bootstrap \citep{Efron82,praestgaard1993exchangeably}
enables us to build confidence intervals for the RATE based on Theorem \ref{thm:asymp-linear}. We conjecture that the standard nonparametric bootstrap would provide valid inference as well; however due to technical challenges in the proof, we do not pursue this result here.

\begin{lemm}
\label{lem:half-boot}
For some parameter $\beta$, let $\what{\beta}$ be an estimator from iid data $(Z_i)_{i=1}^n$ that satisfies $\sqrt{n}(\what{\beta} - \beta) = \frac{1}{\sqrt{n}}\sum_{i=1}^n \psi(Z_i) + o_P(1) \cd Z$ for a fixed, measurable function $\psi$ such that $\E[\psi^2(Z_i)] < \infty$ and $\E[\psi(Z_i)] = 0$. Let $\what{\beta}^\ast$ be the estimate using a random sample (without replacement) of $\lfloor n/2 \rfloor$ of the $Z_i$. Then, conditionally on $(Z_i)_{i=1}^n$, $\sqrt{n}(\what{\beta}^\ast - \what{\beta}) \cd Z$, as well.
\end{lemm}

Theorem~\ref{thm:asymp-linear} shows that $\what{\theta}$ satisfies the assumptions of Lemma \ref{lem:half-boot}, so  we immediately get:

\begin{coro}
Under the assumptions of Theorem~\ref{thm:asymp-linear}, let $V$ be the asymptotic distribution of
$ \sqrt{n}(\what{\theta} - \theta). $
Let $\what{\theta}^\ast$ be an estimate of a RATE using a random sample (without replacement) of $\lfloor n/2 \rfloor$ of the observations. Under the same assumptions as Theorem~\ref{thm:asymp-linear},
$
  \sqrt{n}(\what{\theta}^\ast - \what{\theta}) \cd V,
$
conditionally on $(X_i, Y_i, W_i)_{i=1}^n.$
\label{coro:bootstrap}
\end{coro}

From these results, we get a few corollaries justifying the asymptotic linearity of, and confidence intervals for, the previously proposed evaluation metrics. This generalizes existing results for these metrics that applied primarily to fully randomized trials. \citet{ImaiL19} show results for the Qini metric under fully randomized experiments using Neyman repeated sampling techniques. \citet{ZhaoTiCaClWe13} show consistency of the AUTOC under fully randomized experiments, but not the asymptotic distribution. We significantly generalize these results, by showing that they would hold for any CATE score $\hat{\Gamma}$ satisfying Assumption~\ref{assume:nuisance}. This allows application of our asymptotic results to observational data under unconfoundedness or using various other identification techniques.

First, we show that any weight function $w : [0, 1] \to \R^+$ that is squared-integrable and uniformly continuous satisfies the assumptions on the weight function from Theorem~\ref{thm:asymp-linear}.

\begin{lemm}
Assume that $w : [0, 1] \to \R^+$ is squared-intergrable and uniformly continuous on its range, $[0, 1]$. Then, $ E[\frac{1}{n}\sum_{j=1}^n (w(j/n) - w(1-F_S(S(X_{i(j)}))))^2] \to 0.$
\label{lem:unif-cts}
\end{lemm}
See Section~B.4 in the Supplementary Materials for proof. This immediately implies the following corollary for the Qini score.
\begin{coro}
The Qini coefficient, for completely randomized trials as initially described in \cite{Radcliffe07}, or for observational data satisfying the no unobserved confounding assumption using the AIPW score, is $\sqrt{n}$-consistent and asymptotically linear, as long as the data satisfies Assumption~\ref{assume:regularity} and the nuisance parameter estimates satisfy Assumption~\ref{assume:nuisance}.
\end{coro}

The weight function for the AUTOC is not uniformly continuous. Nonetheless, it still satisfies Assumption~\ref{assume:weight-fn}, according to the following proposition, also proved in Section~B.5.
\begin{prop}
\label{prop:autoc-weight}
The weight function $w_n(t) = H_n - H_{\lfloor nt \rfloor + 1} - 1$ for the AUTOC satisfies the conditions of Assumption~\ref{assume:weight-fn} with $w(t) = -\log(t) - 1$. Therefore, above estimator of the AUTOC is $\sqrt{n}$-consistent and asymptotically linear, as long as the data satisfies Assumption~\ref{assume:regularity} and the nuisance parameter estimates satisfy Assumption~\ref{assume:nuisance} with any $r > 2$.
\end{prop}

\section{Choosing a RATE Metric and Score}
\label{sec:comparing-weighting-functions}
The generality of RATE metrics and their estimators discussed in Section~\ref{sec:estimation} mean that there are important choices that need to be made when using RATE metrics in an application. In particular, one must choose a weighting function $\alpha(\cdot)$ and a form for the score $\what{\Gamma}_i$ as in Section~\ref{sec:scoring-rules}. These choices affect the signal-to-noise ratio of the RATE metric, therefore increasing or decreasing the power of resulting the hypothesis tests.

We begin by considering which weighting function one should use. 
Overall, we conclude that the RATE metric with the best signal-to-noise ratio (and thus power for testing whether the RATE for a given prioritization rule is different from $0$) depends on how the treatment effects differ across quantiles of the prioritization rule. To illustrate this, we construct simulations representing (1) a scenario in which almost all subjects exhibit a linearly varying CATE, (2) a scenario in which only a small portion of the population experiences a varying CATE, and (3) a scenario in between (1) and (2).

We draw $n=400$ samples from a standard uniform distribution $X_i \sim \textnormal{Unif}(0, 1)$, and generate potential outcomes according to the following model:
\begin{gather}
\begin{cases}
    Y_i(w) = \mu_w(X_i) + \varepsilon_{i}(w),~\text{where} \\
    \mu_0(x) = 0 ~\mbox{and}~
    \mu_1(x) = \max{\p{-\frac{2}{p^2} x + \frac{2}{p}, 0}} 
\end{cases}
\label{eqn:mu0-qini-vs-autoc}
\end{gather}
and $\varepsilon_{i}(w) \sim \normal{}(0, 0.2)$ represents i.i.d. random noise and $p$ is a simulation-specific parameter representing the proportion of individuals for whom the CATE is non-zero. We draw the treatment assignment randomly with probability 0.5, so that $e = P\p{W_i = 1} = P\p{W_i = 0} = 0.5$.
Then, we consider the prioritization rule $S(X_i) = 1 - X_i$, which is ``perfect'' in the sense that, for all $x_i$ and $x_j$ in $[0, 1]$, $S(x_i) \geq S(x_j)$ implies $\tau(x_i) \geq \tau(x_j)$. When $p=1.0$, the CATE is nonzero for all subjects and varies linearly over quantiles of the prioritization rule. When $p=0.1$, only a small subset of the population has a nonzero CATE, but the treatment effect is large and changes quickly with the quantile.

\begin{figure}[ht]
    \begin{subfigure}{0.32\textwidth}
        \includegraphics[scale=0.6]{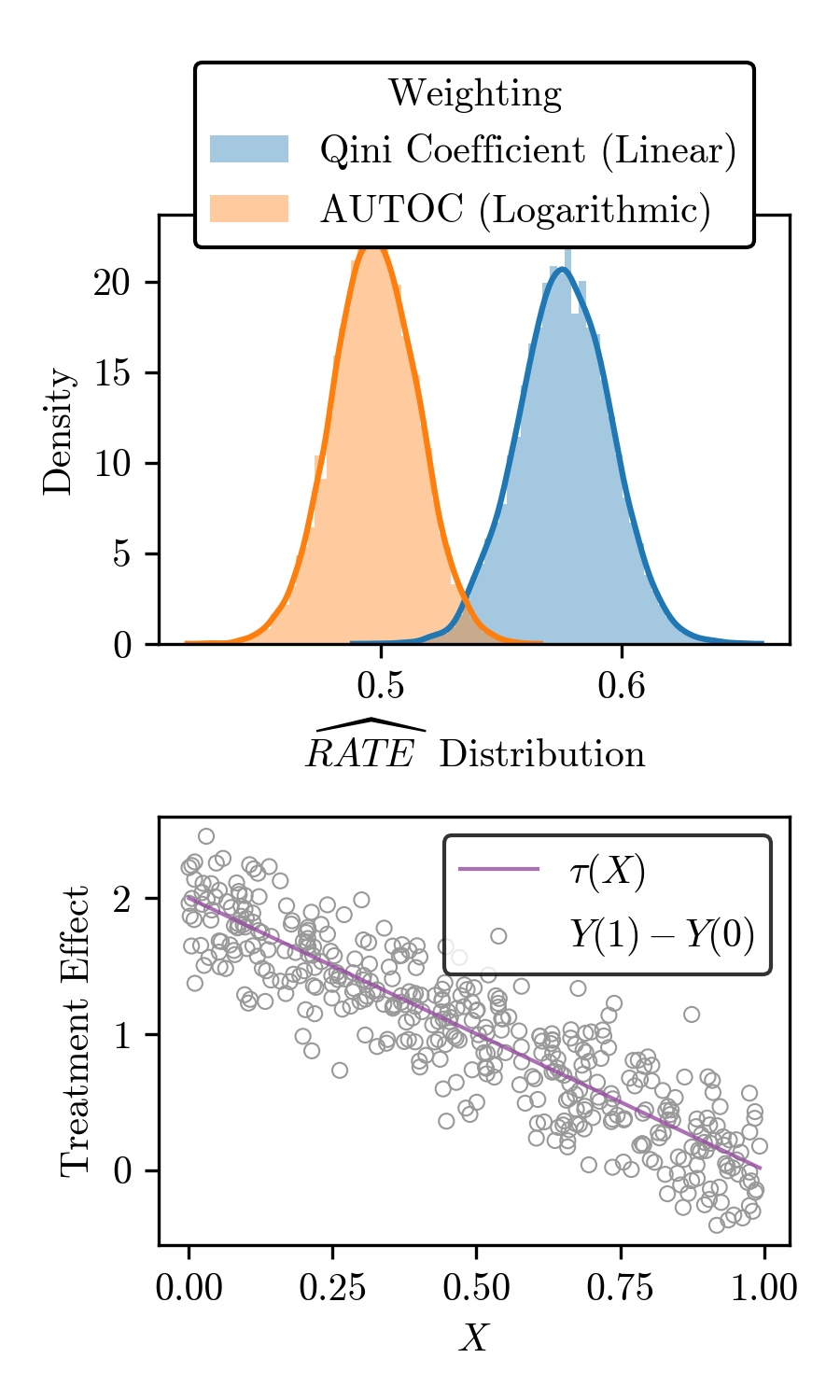} 
        \caption{100\% have $|\tau(X_i)| > 0$}
    \end{subfigure}
    \begin{subfigure}{0.32\textwidth}
        \includegraphics[scale=0.6]{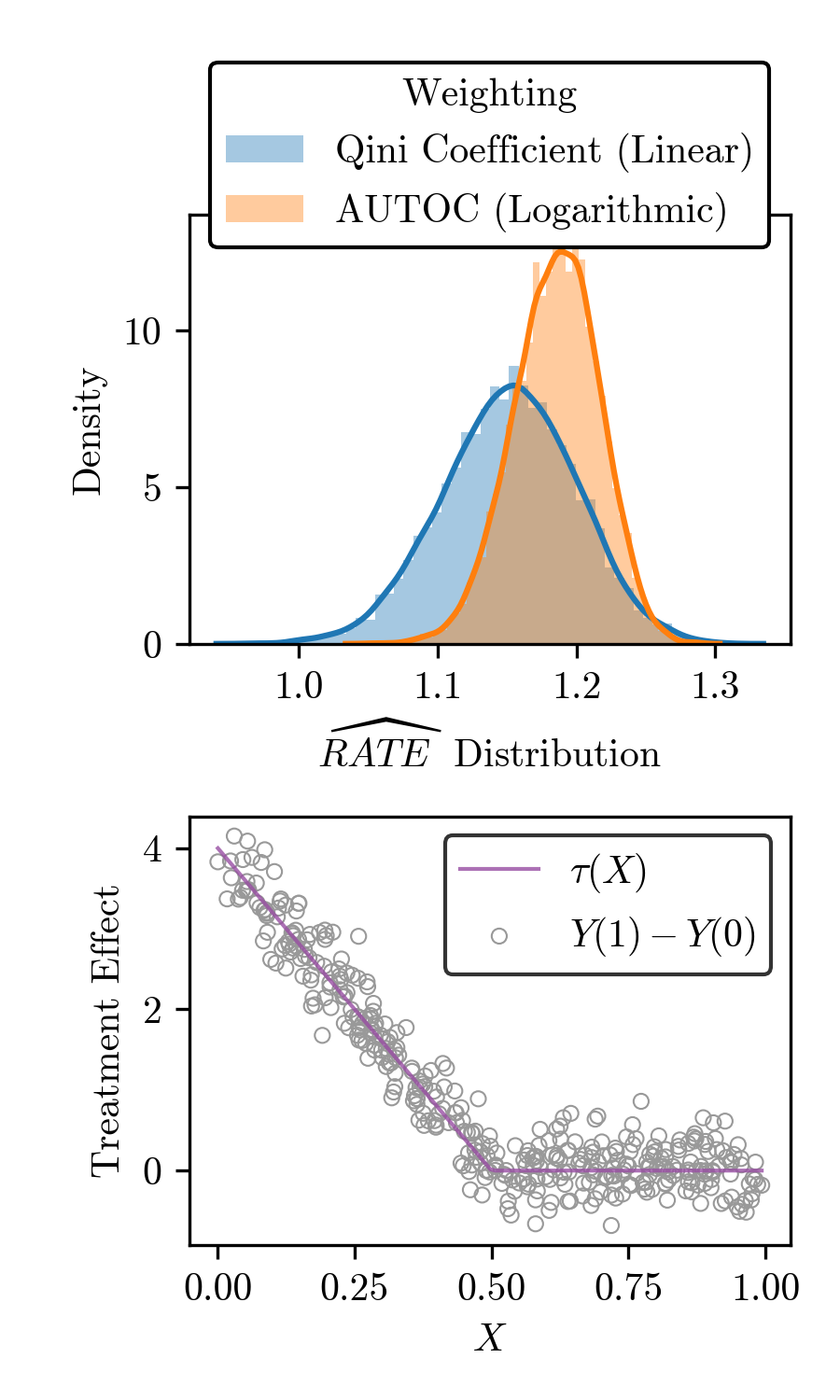}
        \caption{50\% have $|\tau(X_i)| > 0$}
    \end{subfigure}
    \begin{subfigure}{0.32\textwidth}
        \includegraphics[scale=0.6]{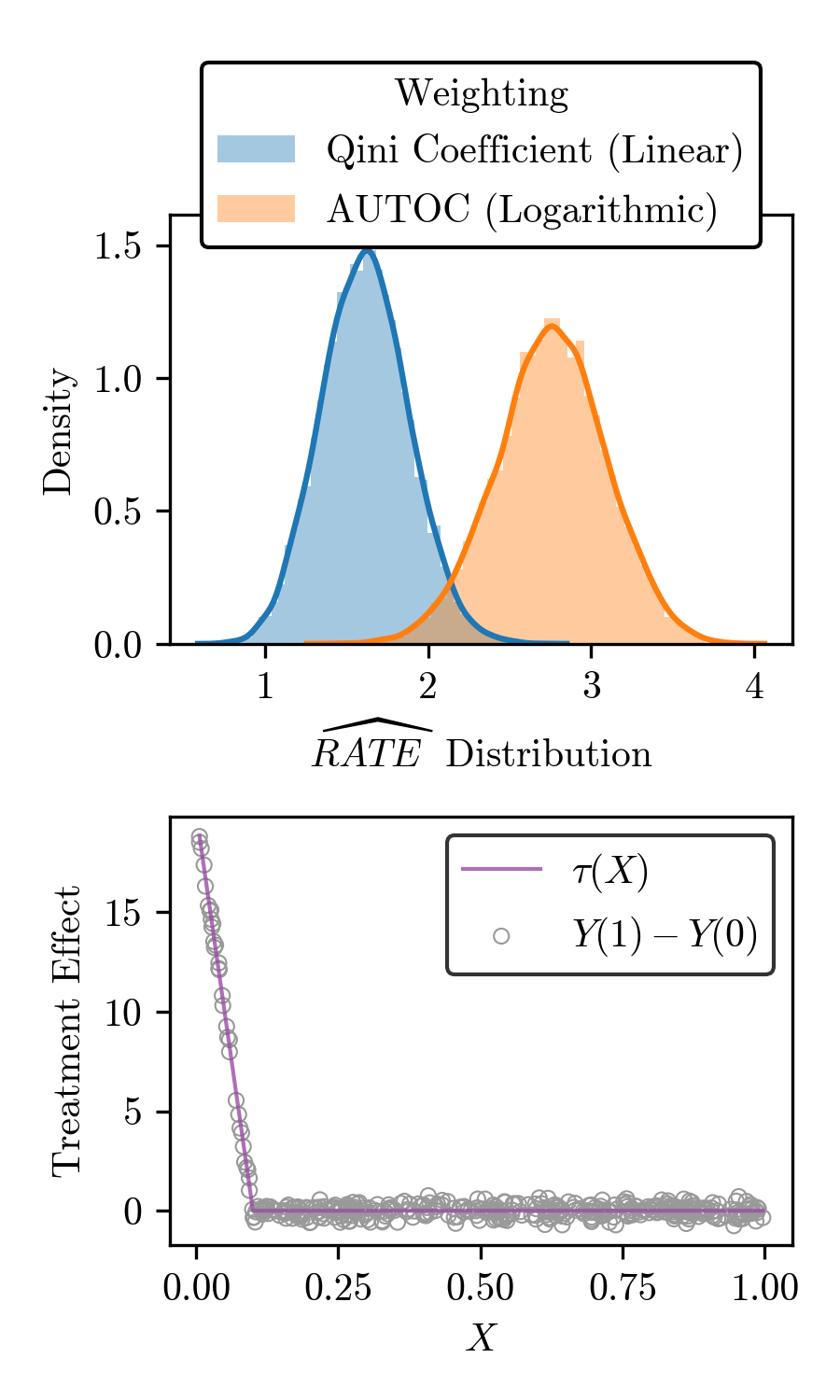}
        \caption{10\% have $|\tau(X_i)| > 0$}
    \end{subfigure}
    \caption{Comparison of linear (Qini) vs. logarithmic (AUTOC) weighting functions. If treatment effects are nonzero for a significant proportion of the population (e.g., left figure) the power of the estimated RATE when using linear weighting (Qini) tends to be greater than when using logarithmic (AUTOC) weighting. Conversely, if nonzero treatment effects are concentrated among a small proportion of the population (e.g., right figure) we see that using logarithmic (AUTOC) weighting leads to a greater power for the estimated RATE relative to using linear (Qini) weighting.}
    \label{fig:comparing_autoc_to_qini}
\end{figure}

For each $p$ in $\{1.0, 0.5, 0.1\}$, we simulate a dataset and calculate both the AUTOC and Qini metrics using oracle AIPW scores $\Gamma^{\ast}_i$ as given in Equation \ref{eqn:oracle-aipw-score}. Figure \ref{fig:comparing_autoc_to_qini} shows the distribution of the estimates of each metric over $10000$ Monte Carlo simulations\footnote{
\spacingset{1}\footnotesize
We note that, given some fixed sample size, $n$, the variance of the weights applied to each doubly robust score in the AUTOC is 1 while the variance of the weights applied to these same scores for the Qini coefficient approaches 0.5 for large $n$. (This limit can be calculated by computing the integral for variance of the weights in the representation in Eq.~\eqref{eq:wrep}, with the helpful observation that $F_S(Q_i)$ is uniform on $(0,1)$). We rescaled the Qini coefficient weights to also have variance 1 for the purposes of this analysis, in order to fairly compare the two metrics. While this rescaling changes the value of the point estimate for these methods, it does not change the statistical power of the two approaches.}.

When all or a substantial portion (e.g., $>50\%$) of the treatment effects are both heterogeneous and non-zero, using the Qini metric can lead to a higher RATE compared to the AUTOC. However, when only a small subset of individuals have a non-zero heterogeneous treatment effect (e.g., $
\leq 10\%$), using the AUTOC can lead to a higher RATE and thus be advantageous. This example highlights a broader intuition about when and why one might choose the AUTOC vs Qini coefficient as a RATE metric: If a researcher believes that only a small subset of their study population experiences nontrivial heterogeneous treatment effects, they should use logarithmic weighting; if they believe that heterogeneous treatment effects are diffuse and substantial across the entire study population, they should use linear weighting. Following this guideline should, in general, increase the statistical power of tests against the absence of heterogeneous treatment effects.

A second important consideration for researchers looking to use RATE metrics concerns how best to generate score estimates, $\hGamma_i$, for each individual the evaluation dataset. Here, we focus on choosing a score in the context of data with no unobserved confounding.

As discussed in Section \ref{sec:scoring-rules}, both IPW scores and AIPW scores satisfy necessary conditions (Assumption~\ref{assume:nuisance}) for RATE estimation in randomized trials and sometimes do for observational study settings with no unobserved confounding, as well;\footnote{
\spacingset{1}\footnotesize
The conditions under which the IPW satisfies Assumption~\ref{assume:nuisance} are typically more restrictive than for the AIPW, therefore making the comparison meaningful primarily when the IPW estimator is $\sqrt{n}$-consistent.} however, AIPW scores can yield lower variance score estimates in both settings. In a simple experiment, we highlight the significant practical impact of using AIPW score estimates to reduce RATE estimate variance and improve power in a randomized trial.

\begin{wrapfigure}{R}{0.65\textwidth}
    \begin{center}
    \includegraphics[trim=0mm 10mm 0mm 15mm, width=0.6\textwidth]{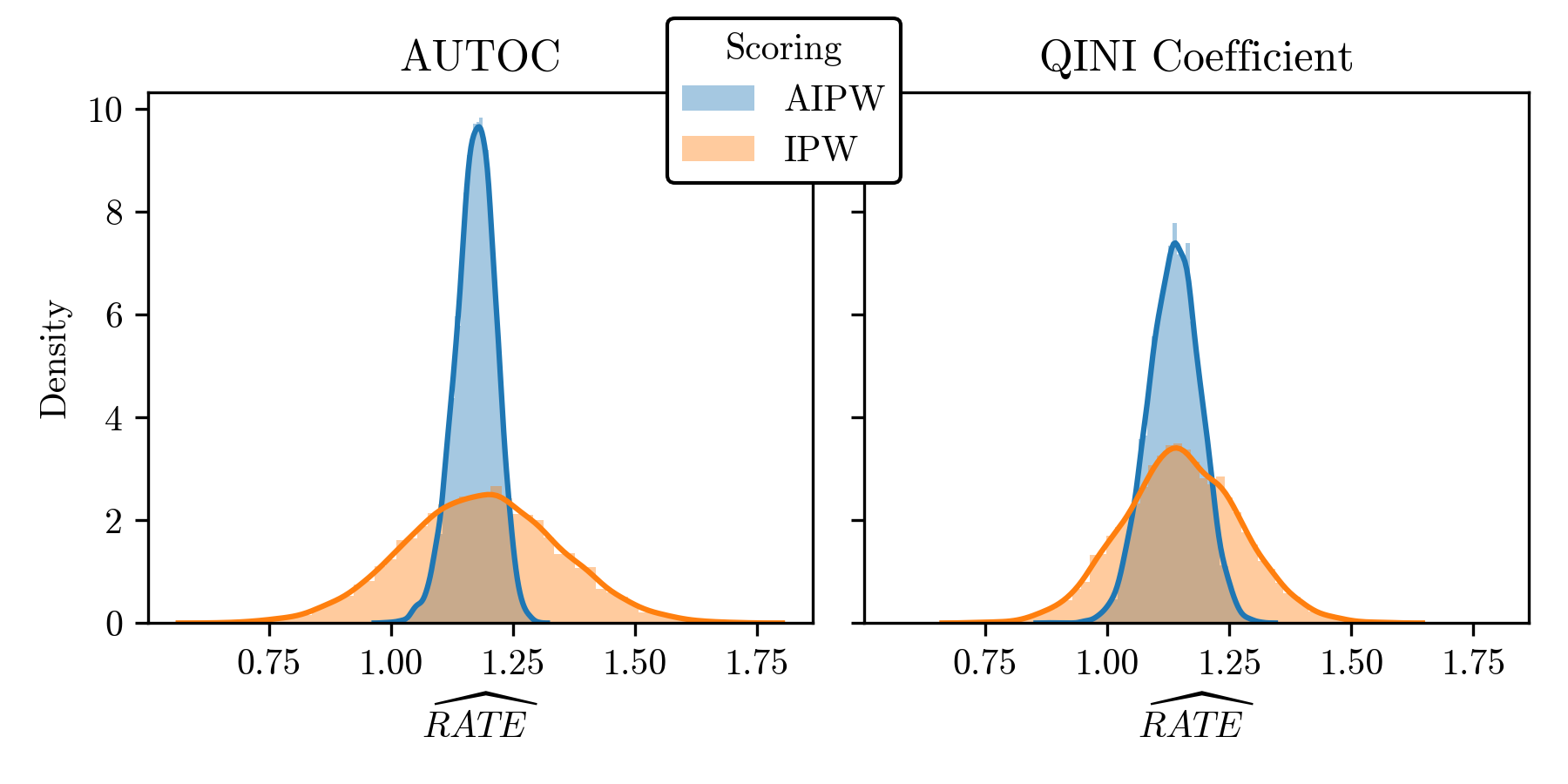}
    \end{center}
    \caption{Comparison of RATE estimates obtained using Inverse-Propensity Weighting (IPW) vs. Augmented IPW scores. Using AIPW scores yield RATE estimates with lower variance and greater statistical power compared to RATE estimates using IPW scores.\label{fig:comparing-aipw-to-ipw}}
\end{wrapfigure}

We consider the same simulation setup as in Section \ref{sec:comparing-weighting-functions}.
 Here, however, instead of using oracle AIPW scores (\ref{eqn:oracle-aipw-score}) to calculate the RATE, we use estimated IPW scores and AIPW scores. In calculating AIPW scores, we use nuisance parameters 
 $\hat{m}(x, w)$ estimated using a random forest regression model with cross-fitting \citep{chernozhukov2018generic}. Thus in each of 10000 Monte Carlo simulations, we generate a synthetic sample according to \eqref{eqn:mu0-qini-vs-autoc}, estimate both the AUTOC and Qini coefficient using IPW and AIPW scores on that sample, and compare the final distributions of 10000 RATE point estimates.

Figure \ref{fig:comparing-aipw-to-ipw} shows the resulting distributions. The AIPW score yields lower variance estimates of the RATE compared to the IPW score, consistent with the improved efficiency of AIPW scores for estimating the average treatment effect \citep{robins1994estimation}. The AIPW score improves the precision of the point estimate and the statistical power of derived hypothesis tests. We recommend using AIPW scores for estimating the RATE in both randomized trial and observational study settings.

\section{Aspirin and Stroke}
\label{sec:ist}

As a first application of RATE, we revisit the International Stroke Trial (IST) introduced
in Section \ref{sec:motivation} \citep{group1997international,sandercock2011international}.
It assessed the treatment effects of Aspirin, Heparin (an anticoagulant), both, or neither for
patients with presumed acute ischaemic stroke in a factorial design on 19,435 patients across 36 countries.
Primary outcomes included (1) death within 14 days of stroke onset, and (2) death or dependency at 6 months.
Follow-up for the primary outcomes was 99\% complete. We focus on the outcome of death or
dependency at 6 months and restrict our analysis to those patients for whom this outcome was recorded.
We also only analyze the effect of Aspirin, irrespective of Heparin assignment status. Consistent with
the original trial paper, we employ an intention-to-treat analysis approach \citep{group1997international}.

We estimate RATE metrics for several prioritization rules on IST, comprising both CATE-based and risk-based approaches. We evaluate 
\begin{enumerate}
    \item  A prognostic random forest risk model \citep{breiman2001random} trained on individuals assigned to the control arm (this model was used to create quintiles in Figure \ref{fig:IST:CATE-vs-risk});
    \item A logistic regression risk model trained on 
    \item A causal forest model \citep{athey2019generalized} that specifically targets the CATE; and
    \item A logistic regression T-learner \citep{kunzel2019metalearners}.
\end{enumerate}
Here, treatment effects with a negative sign are desirable (i.e., they correspond to reduced mortality).
Our goal is to investigate whether we can identify a subgroup who benefits less than average from aspirin
(and so should potentially not be prescribed aspirin), and we hypothesize that patients with a high baseline
risk may form such a group. Thus, when evaluating prognostic rules we rank patients with high risk first and
when evaluating CATE-based rules we rank patients with positive CATE first. We expect to get a positive RATE
in doing so, meaning that we expect to find that patients we believe should benefit less from aspirin in fact do.

In order to avoid overfitting, rules are first trained on a 50\% train split, then evaluated on a 50\% test split. 
For each prioritization rule considered, we learn the optimal parameters of the prioritization rule on the train set, and then generate a point estimate and Gaussian half-sample bootstrap 95\%-confidence interval for the RATE using 10,000 bootstrap samples from data in the test set. We  report associated two-sided $p$-values against the RATE being 0. This test serves as a test against the absence of heterogeneous treatment effects.
Given the high (99\%) follow-up rate, we simply discarded observations with missing follow-up and did not attempt to correct for censoring.
Details on each of these methods can be found in Section C.1.2 of the Supplement and the papers cited above. All prioritization rules are trained and evaluated on the IST data.


\begin{table}
\centering
\begin{tabular}{|r|cc|}
	\hline
	Prioritization Rule & AUTOC (95\% CI) & $p$-value \\
	\hline
	Random Forest Risk (grf) & 0.014 (0.002, 0.027) & 0.022 \\
	Logistic Regression Risk & 0.011 (0.000, 0.230) & 0.057 \\
        Causal Forest (grf) & 0.006 (-0.012, 0.024) & 0.493 \\
        Logistic Regression T-learner & 0.005 (-0.013, 0.023) & 0.588 \\
	\hline
\end{tabular}
\caption{RATE estimates obtained using data from the International Stroke Trial. We also show 95\% confidence intervals obtained using the half-sample bootstrap,  along with associated $p$-values. The $p$-values are not corrected for multiple testing.}
\label{tab:IST-RATE-vs-model}
\end{table}

AUTOC estimates from our considered methods, along with 95\% confidence intervals,
are shown in Table \ref{tab:IST-RATE-vs-model}. Our first finding is that the
AUTOC with prognostic-based targeting (i.e., with the random forest risk model)
is statistically significant, with a $p$-value of 0.022. This is an interesting
finding since, as discussed in Section \ref{sec:motivation}, \citet{group1997international}
and \citet{chen2000indications} had considered baseline prognostic index as a potential
effect modifier; however, their subgroup-based analysis was not powerful enough to
obtain statistical significance. Here, in contrast, we find that if we combine
machine learning-based estimation of the prognostic index (here using random forests)
with RATE-based evaluation, we get significance.\footnote{
\spacingset{1}\footnotesize
We do not apply a multiplicity
correction to the $p$-values here since the random forest risk / AUTOC analysis was our
primary analysis; the other analyses are provided for comparison and further insight.}
In other words, from a statistical point of view, this result highlights the gain in power
from using RATE relative to estimating treatment effects separately for different risk quantiles.

\section{Prioritization Rules for Antihypertensive Treatment}
\label{sec:CVD}

As a second application of our framework, we consider the problem of personalized antihypertensive treatment. Hypertension, or elevated blood pressure (BP), is implicated in 14\% of all deaths across the globe \citep{fisher2018hypertension} and has a high prevalence around the world \citep{muntner2018potential}. Effective hypertension treatment significantly reduces the risk of negative cardiovascular disease outcomes \citep{psaty2003health}. However, it is less clear the degree to which benefits of intensively targeting a low BP with antihypertensive medications are uniform across patients, as antihypertensives carry risks and cause adverse side effects \citep{accord2010effects, sprint2015randomized}.

Here, we study personalized treatment rules for hypertension using two large randomized controlled trials on the effectiveness of intensive blood-pressure control: the Action to Control Cardiovascular Risk in Diabetes Blood Pressure trial (ACCORD-BP) \citep{accord2010effects} and the Systolic Blood Pressure Intervention Trial (SPRINT) \citep{sprint2015randomized}. The ACCORD-BP and SPRINT trials share many similarities such as the treatment (intensive blood-pressure control, which aims for SBP $< 120$ mm Hg vs. standard blood-pressure control, which aims for SBP $< 140$ mm Hg), study population (see Section D.1 of the Supplementary Materials), and primary outcomes (e.g., myocardial infarction, stroke, death from cardiovascular causes). However, the trials differ in two key aspects: 
(1) The ACCORD-BP trial was conducted on participants with type 2 diabetes, while SPRINT was conducted on participants without diabetes; and (2) the ACCORD-BP trial found that risk reduction under the intensive blood-pressure treatment arm was non-significant (12\%, 95\% CI -6\% to 27\%), while the SPRINT trial found that intensive blood-pressure treatment led to significantly lower risk relative to controls (25\% , 95\% CI 11\% to 36\%). The literature has hypothesized that the differences in reported significance of treatment effects may be due to treatment effect heterogeneity \citep{beddhu2018effects, basu2017detecting, kaul2017tale}. This is a hypothesis we can test using the RATE.
Baseline characteristics for study participants are given in Table \ref{tab:baseline-characteristics} of the Supplementary Material.

We use RATE to study the ability of a number of prioritization rules (i.e., $S$ in our notation, as per Definition 1 in Section \ref{sec:RATE}) to identify patients who benefit more than average from intensive blood pressure medication in SPRINT and ACCORD-BP. First, we consider two risk scores for cardiovascular disease that are widely used in clinical practice:\footnote{
\spacingset{1}\footnotesize
For example, the American College of Cardiology and American Heart Association 2017 guidelines recommend that, in patients with Stage 1 hypertension (SBP of 130-139 mm Hg), antihypertensive medications should be used for intensive blood pressure targeting only if the estimated 10-year CVD risk is $\geq10\%$, as measured by the ACC/AHA Pooled Cohort Equations \citep{whelton20182017, goff2014accaha}.
}
The Framingham Risk Score \citep{DAgostinoVaPeWoCoMaKa08}\footnote{
    \spacingset{1}\footnotesize
    The Framingham Risk Score predicts an individual's 10-year risk of developing atherosclerotic cardiovascular disease in patients without a prior cardiac event \citep{DAgostinoVaPeWoCoMaKa08}. 
    It consists of two separate multivariate models for males and females. Each model takes as input a patient's age, sex, current smoking status, systolic blood pressure, HDL cholesterol, total cholesterol, diabetes diagnosis, and whether the patient is currently being treated with medications to reduce their blood pressure. 
    Multiple studies have validated the ability of the Framingham Risk Score to discriminate between high- and low-risk patients \citep[e.g.,][]{artigao2013framingham}
    and criticized its moderate to poor calibration---particularly among younger patients, women, and ethnically diverse cohorts \citep[e.g.,][]{defilippis2015analysis}.
}
and the ACC/AHA Pooled Cohort Equations \citep{goff2014accaha}.\footnote{
    \spacingset{1}\footnotesize
    The ACC/AHA Pooled Cohort Equations
    use a set of risk estimates obtained from pooling multiple cohorts to calculate a patient's 10-year risk of developing atherosclerotic cardiovascular disease \citep{goff2014accaha}. The Pooled Cohort Equations consists of four separate multivariate models, stratified by race (White and Black) and sex (men and women). 
    In addition to the variables in the Framingham Risk Score, the Pooled Cohort Equations also incorporate race and diabetes status. The Pooled Cohort Equations have been validated for the overall risk prediction task, although their performance for certain subgroups has been criticized \citep{yadlowsky2018clinical}.
} 
Second, we consider prioritization based on a risk score we trained ourselves via a random forest. Finally we consider two prioritization rules based on CATE estimates, one that fits CATE via Cox proportional hazards regression using what \citet{kunzel2019metalearners} called the S-learner approach, and another that fits CATE using causal survival forests \citep{cui2023estimating}. 

We estimate RATE separately in the ACCORD-BP and SPRINT trials. Whenever we use scoring rules
that we fit ourselves (i.e., based on forests or Cox regression), we follow the approach of \citet{basu2017detecting} and train models on the other
trial than the one we are estimating a RATE on (e.g., we evaluate a causal survival forest trained
with ACCORD-BP on SPRINT, and vice-versa). We produce a 95\% confidence intervals and associated
two-sided $P$-values as in Section \ref{sec:ist}. 
Unlike in the IST discussed above, both the ACCORD-BP and SPRINT trials had non-negligible loss to follow-up (94.6\% were lost to follow-up in SPRINT, 85.2\% in ACCORD-BP)
and so we used right-censored time-to-event outcomes
throughout; see Section~A.3 of the Supplementary Materials for details.

Table \ref{train_accord_test_sprint_rmst} shows results from training prioritization rules on ACCORD-BP and evaluating them on SPRINT, and Table \ref{train_sprint_test_accord_rmst} for training on SPRINT and evaluating on ACCORD-BP.
In this experiment, the estimated RATE metrics for all risk-based prioritization rules were not significantly different from zero when evaluated on both ACCORD-BP and SPRINT. These findings suggest that, for these two trial populations, the risk-based estimators in consideration would not order patients in accordance with estimated treatment benefit. Additionally, the estimated RATE prioritization rules that directly target the CATE (the Causal Survival Forest and Cox Proportional Hazards S-Learner) were
not significantly different from 0 at level $\alpha = 0.05$ in both ACCORD-BP and SPRINT, using Gaussian half-sample bootstrap confidence intervals.

\begin{table}
\centering
\begin{tabular}{|r|cc|}
	\hline
	Prioritization Rule & AUTOC (95\% CI) & $p$-value \\
	\hline
	Causal Survival Forest (grf) & 2.37 (-3.52, 8.27) & 0.43 \\
	Cox PH S-learner & 3.57 (-1.92, 9.07) & 0.20 \\
	\hline
	Random Survival Forest Risk (grf) & 4.21 (-2.46, 10.88) & 0.22 \\
	Framingham Risk Score & 1.83 (-4.28, 7.94) & 0.55 \\
	ACC/AHA Pooled Cohort Equations & -1.76 (8.26, 4.73) & 0.59 \\
	\hline
\end{tabular}
\caption{AUTOC estimates obtained using data from SPRINT ($n = 9069$), with prioritization rules trained on ACCORD ($n = 4535$), if necessary. We also show 95\% confidence intervals obtained using the half-sample bootstrap,  along with associated $p$-values.}
\label{train_accord_test_sprint_rmst}
\end{table}

\begin{table}
\centering
\begin{tabular}{|r|cc|}
	\hline
	Prioritization Rule & AUTOC (95\% CI) & $p$-value \\
	\hline
	Causal Survival Forest (grf) & -5.29 (-17.58, 7.00) & 0.40 \\
        Cox PH S-learner & 8.90 (-2.62, 20.43) & 0.13 \\
	\hline
	Random Survival Forest Risk (grf) & 2.93 (-8.82, 14.67) & 0.62 \\
	Framingham Risk Score & -7.01 (-18.45, 4.48) & 0.25 \\
	ACC/AHA Pooled Cohort Equations & -1.43 (-13.22, 10.36) & 0.81 \\
	\hline
\end{tabular}
\caption{AUTOC estimates obtained using data from ACCORD ($n = 4535$), with prioritization rules trained on SPRINT ($n = 9069$), if necessary. We also show 95\% confidence intervals obtained using the half-sample bootstrap,  along with associated $p$-values.}
\label{train_sprint_test_accord_rmst}
\end{table}

Our findings show no significant evidence of heterogeneous treatment effects in the SPRINT and ACCORD-BP trials in terms of the restricted mean survival time (RMST). We also assessed the RATE using a combined version of the SPRINT/ACCORD-BP data, but these results were similarly not significant (see Section D.2 of the Supplementary Materials).

From the results on these data, it remains ambiguous whether or not clinical use of risk scores like the Framingham Risk Score to guide blood pressure control treatment 
is in fact more beneficial than simply using the original trial's inclusion/exclusion protocols to guide treatment (given that treatment benefit was significant in SPRINT for the primary outcome). It also remains ambiguous whether there is significant treatment effect heterogeneity at all in these trial populations. Our main results suggest that the SPRINT/ACCORD-BP results are not well-enough powered to support conclusive claims of treatment effect heterogeneity using an RMST outcome. However, we also note contemporaneous work \citep{oikonomou2022individualising, xu2022treatment} that obtain (marginally) significant detections for heterogeneity in absolute-risk-difference outcomes. Overall, these results highlight the fragility of any claims regarding treatment effect heterogeneity in these trials and demonstrates the importance of having a principled approach towards performing such tests. The RATE provides such a principled approach for evaluating and comparing prioritization rules, as well as testing for global treatment effect heterogeneity.

\section{Application to Uplift Modeling}
\label{sec:uplift-modeling}

Finally, to demonstrate versatility of the RATE approach outside of the medical domain,
we also briefly consider an application to marketing. Traditionally, marketing offers
have been targeted based on analogues to risk modeling (e.g., target retention offers
to customers predicted to be at risk of canceling the service); however,  there has been
 interest in methods based on CATE estimation, or \emph{uplift modeling}
\citep{Ascarza18,Radcliffe07,RadcliffeSu99}. One driver of recent interest in CATE
estimation is that randomized trials are particularly easy to design and implement
the world of digital marketing, and many tools exist to help data scientists run these
experiments \citep{HenningObTa15,JohnsonLeNu17,GordonZeBhCh19}.

We use RATE to compare different targeting rules on a large benchmark dataset released by Criteo for studying uplift modeling in online digital advertising, based on anonymized results from a number of randomized controlled trials on incrementality \citep{DiemertBeReAm18}.\footnote{\spacingset{1}\footnotesize
In combining trials, the interpretation of the treatment is subtle, corresponding to an intent to treat with one of a handful of arbitrarily chosen ads. The purpose of the data is to provide a benchmark for uplift modeling, and therefore, the results are not meant to be used in a particular application. See \cite{DiemertBeReAm18} for more information about the dataset construction and validation.}
For the purposes of anonymization, the features released are a random projection of the user features into a 12-dimensional space. The data also contain a binary indicator for treatment status, and two outcomes, visiting the site following the ad, and further conversion into a customer.

\begin{table}[]
    \centering
    \begin{tabular}{|l|c|c|c|c|c|c|}
    \hline
         & \multicolumn{2}{c|}{Full data} &  \multicolumn{2}{c|}{Training subsample} &  \multicolumn{2}{c|}{Test subsample} \\
         \hline
         & Visit & Conversion & Visit & Conversion & Visit & Conversion
         \\
    \hline
        Treatment & 0.0485 & 0.0031 & 0.0492 & 0.0031 & 0.0488 & 0.0031 \\
    \hline
        Control & 0.0382 & 0.0019 & 0.0395 & 0.0019 & 0.0377 & 0.0021 \\
    \hline
    \end{tabular}
    \caption{Visit and conversion rates in each arm of the Criteo Uplift Benchmark dataset, and our two randomly subsampled datasets for training and testing uplift models. Tests for difference in means between the rates of the three samples are insignificant at $5\%$-level.}
    \label{tab:criteo-avg-rates}
\end{table}

\begin{figure}[t]
    \begin{subfigure}{0.5\textwidth}
    \centering
        \includegraphics[width=0.75\columnwidth]{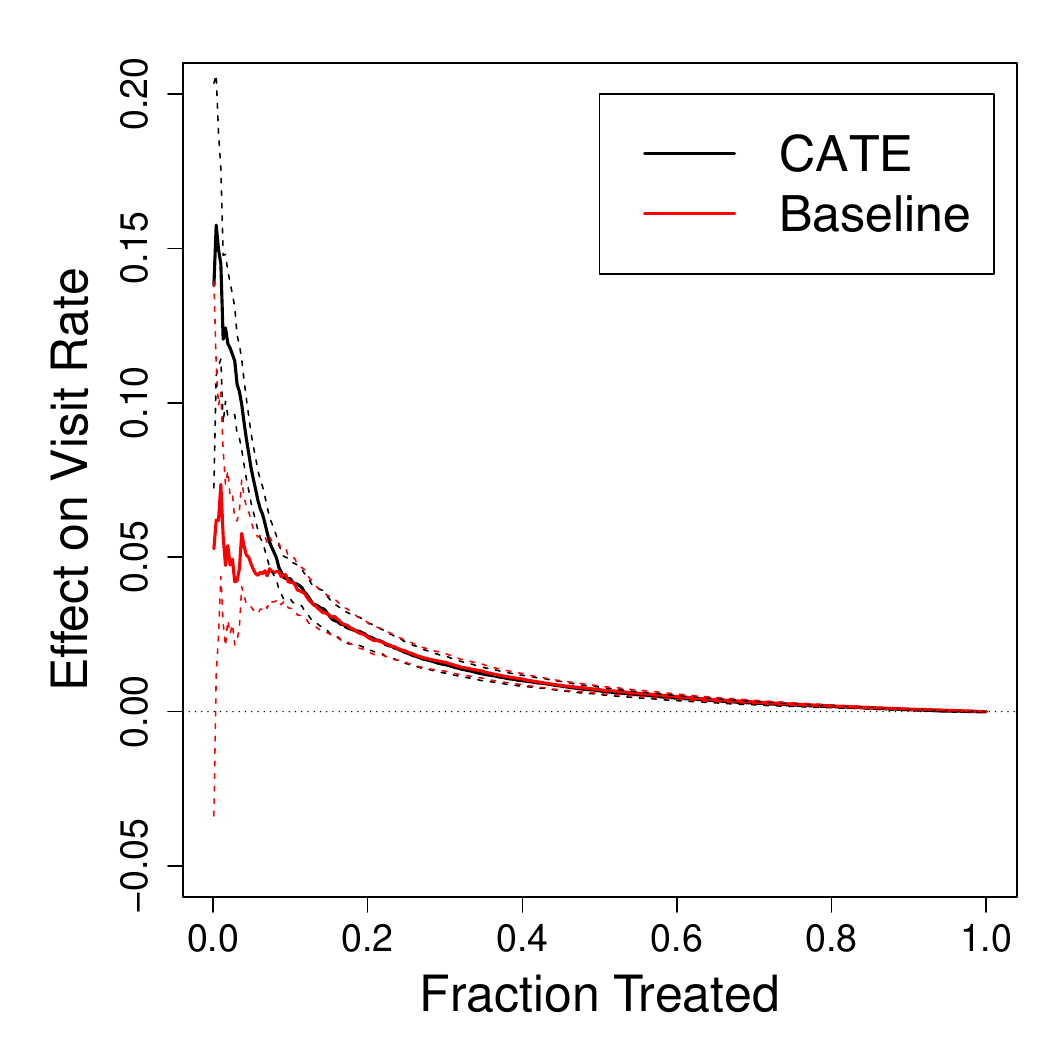}
        \label{fig:criteo:visit-compare}
    \end{subfigure}
    \begin{subfigure}{0.5\textwidth}
    \centering
        \includegraphics[width=0.75\columnwidth]{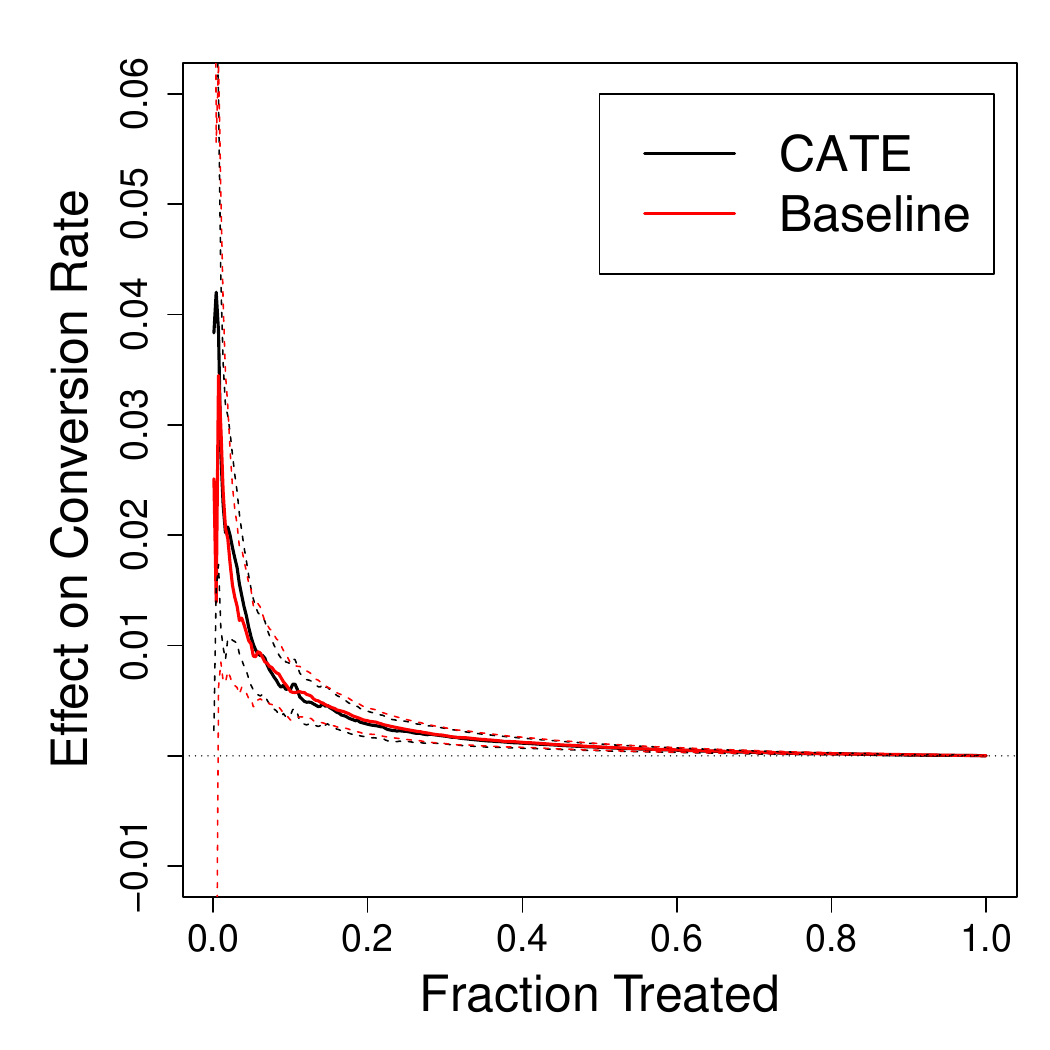}
        \label{fig:criteo:conversion-compare}
    \end{subfigure}
    \caption{TOC curves for two prioritization rules (baseline- and CATE-based) and two outcomes (rate of visits and rate of conversion) on Criteo Uplift benchmark dataset.}
    \label{fig:criteo}
\end{figure}

We considered two random-forest based prioritization rules, one fit to the baseline probability of the outcome if untreated (ie., how likely is the user to visit or convert without any advertising), and one fit using causal forests for the conditional average treatment effect. We will call these the baseline and CATE prioritization rules, respectively.\footnote{\spacingset{1}\footnotesize
In both cases, we used default hyperparameters, except for the minimum node size and \texttt{mtry} parameter. The smaller the minimum node size, the slower that fitting the random forests becomes, and this is particularly problematic for the large dataset size used here. We selected the minimum node size to be 5000, which was the smallest value such that the forest could be fit reasonably quickly on a modern laptop computer. We selected the optimal \texttt{mtry} parameter for both algorithms with respect to the AUTOC on the test dataset. Because there are only 12 possible values for this parameter, and because the test set is large, the bias in the final results introduced by doing this is negligible.}
The data contain 25,309,483 samples with an unbalanced fraction of treated vs control units;  however, we randomly selected two sets of 320,000 samples without replacement, each equally balanced with 160,000 treated and 160,000 control units. We used the first set for training prioritization rule models, and the second set for evaluating the RATE metrics of the learned models.
The average visit and conversion rates in the treatment and control arms of the updated dataset and our subsamples are reported in Table~\ref{tab:criteo-avg-rates}.

\begin{table}[t]
    \centering
    \begin{tabular}{|l|c|c|c|c|}
    \hline
         \multirow{2}{*}{} & \multicolumn{2}{c|}{AUTOC (95\% CI) } &  \multicolumn{2}{c|}{Qini coefficient (95\% CI)} \\
         \cline{2-5}
         & Visit & Conversion & Visit & Conversion
         \\
    \hline
        Baseline & \makecell{0.0136\\(0.0111, 0.0161)} & \makecell{0.0023\\(0.0011, 0.0035)} & \makecell{0.0059 \\(0.0048, 0.0069)} & \makecell{0.00074\\(0.00047, 0.00101)} \\
    \hline
        CATE & \makecell{0.0171 \\(0.0148, 0.0194)} & \makecell{0.0025\\(0.0015, 0.0035)} & \makecell{0.0058 \\(0.0049, 0.0067)} & \makecell{0.00070\\(0.00044, 0.00096)}  \\
    \hline
    \end{tabular}
    \caption{RATE metrics on Criteo Uplift Benchmark dataset. Confidence intervals are estimated using the half-sample bootstrap.}
    \label{tab:criteo-rate}
\end{table}

Figure~\ref{fig:criteo} shows the TOC curves for each of the prioritization rules for each of the outcomes, and Table~\ref{tab:criteo-rate} shows the RATE metrics (AUTOC and Qini coefficient) summarizing these curves, along with their confidence intervals.
Both methods on both outcomes work better than random prioritization, with strictly positive CIs. For visits, the AUTOC for the CATE-based prioritization rule is marginally higher than the baseline-based rule; the difference between their AUTOC statistics is 0.0034 (95\% CI 0.0019, 0.0049), which is a relative increase in AUTOC of 25\%. However, the Qini coefficient and the results for conversions are not significantly different between the two prioritization rules. Looking at the TOC curves in Figure~\ref{fig:criteo} gives suggestions as to why this is the case. We see that the average treatment effect among the highest prioritized individuals according to the CATE-based rule is much higher than for the baseline-based rule at similar fractions treated. However, the difference disappears once the fraction treated is more than 10\%. The Qini coefficient weights the TOC by the fraction treated, so it downweights the importance of these small groups with large treatment benefit. On the other hand, these groups are more influential in the AUTOC, explaining why we observe a statistically significant difference in this metric. Note that while the confidence intervals overlap between the two metrics, paired bootstrapping is more precise and detects a significant contrast for the AUTOC.

\section{Discussion}
\label{sec:discussion}

We discussed evaluation of treatment prioritization rules using a number of metrics motivated by the Area under the Receiver-operating Characteristic curve (AUROC) for classification \citep{green1966signal, zweig1993receiver}. Our approach enables threshold-agnostic evaluation and comparison of risk-based and CATE-based estimators alike. Researchers can use RATE metrics and half-sample bootstrap confidence intervals to (1) evaluate whether a prioritization rule performs significantly better than chance/random treatment allocation in stratifying subjects according to estimated treatment benefit; (2) directly, meaningfully, and rigorously compare the targeting performance of two different prioritization rules and determine whether one exhibits significant superiority; and (3) perform a global test for the presence of predictable heterogeneous treatment effects, by coupling estimation of the RATE with a powerful, flexible, non-parametric CATE estimator as the prioritization rule.
Our approach works in the context of binary outcomes, continuous, real-valued outcomes, and right-censored time-to-event outcomes, making it a practical tool for a number of different fields of study, from marketing and business to medicine and public policy. Additionally, it generalizes existing estimators of the Qini score~\citep{Radcliffe07} and AUTOC~\citep{ZhaoTiCaClWe13} to any setting where an identification strategy allows the construction of a CATE score, including observational studies with unconfoundedness. The representation of these metrics as a weighted average treatment effect that we provide in \eqref{eq:as-Lstat} implies that these metrics will be similar to the average treatment effect in terms of their sensitivity to model misspecification or unobserved confounding.

To apply a treatment prioritization rule in practice, one must convert the prioritization rule into a policy for whom to intervene on versus not. From this perspective, it's intuitive to consider directly evaluating the prioritization rule on the average reward of the policy implied by thresholding the prioritization rule at a specified level. The implied policy value of treating the top $u$-th of units can be measured using the ``high-vs-others'' estimator as discussed in Section~\ref{sec:RATE}. This estimator is a RATE and so the tools for inference as discussed in Section~\ref{sec:asymptotics} apply. If the threshold---e.g. treat with antihypertensive medications with 10-year risk of ASCVD $\ge 10\%$---is known at the time of evaluation, this is a reasonable approach. This can also be a reasonable approach if the author has expert knowledge (e.g., regarding the treatment and study population) that would justify assigning weights for all units other than the top $u$-th a weight of zero in calculating the RATE. 


However, in many cases, the policy is unclear or may still be subject to change at the time that the prioritization rule is designed. For instance, in healthcare applications, the policy for treatment decisions might additionally depend on individualized risk of serious adverse events from the treatment, patient and provider preferences, costs, and more. In marketing applications, they may depend on varying costs associated with advertising, concerns for advertisement blindness \citep{HenningObTa15} and other considerations.

\begin{wrapfigure}{r}{0.52\textwidth}
\centering
    \includegraphics[trim={1.5cm 1.5cm 1.5cm 1.5cm},width=0.45\textwidth]{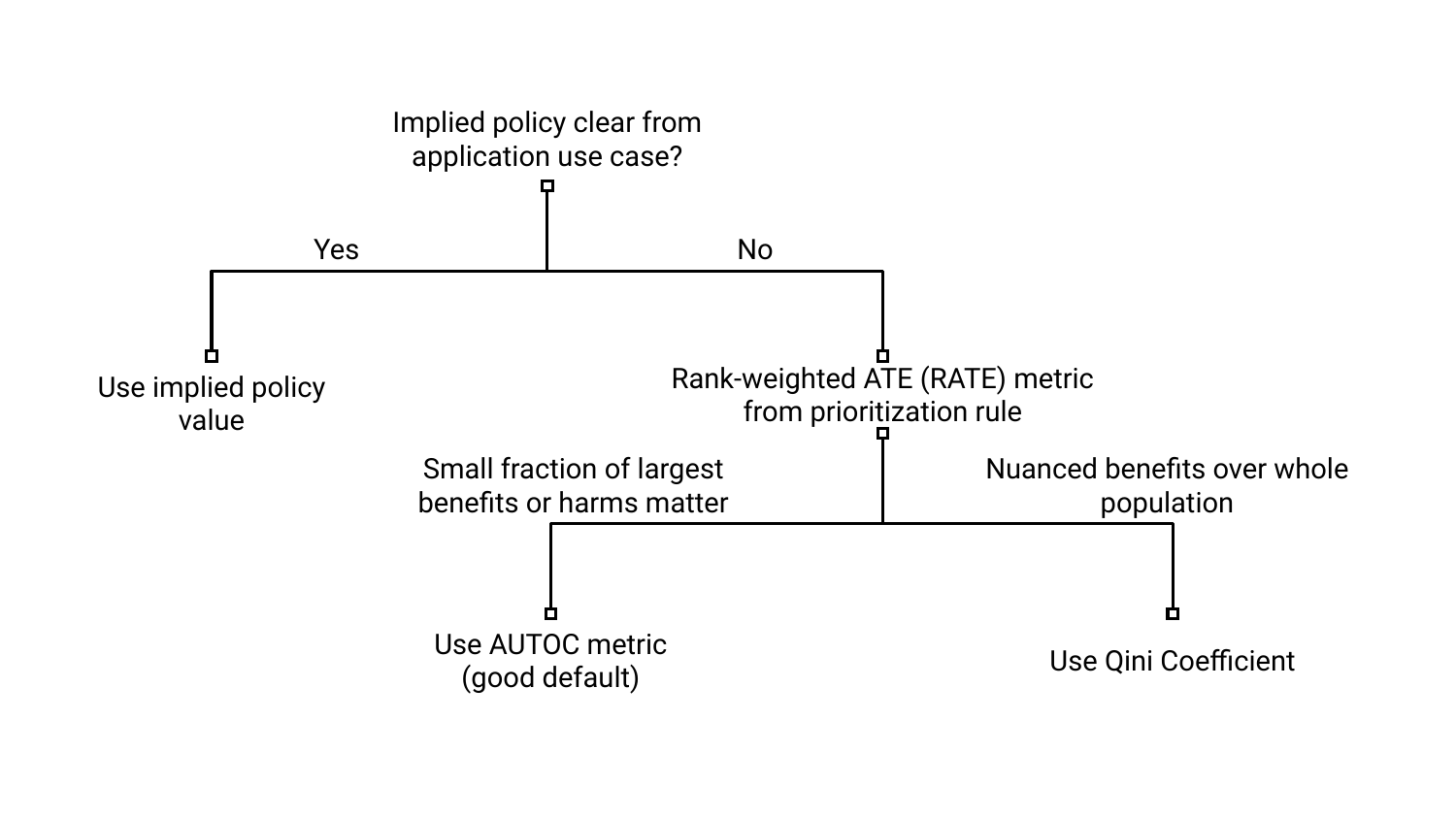}
    \caption{A flowchart for choosing the best evaluation metric to use in an application.\label{fig:outline}}
\end{wrapfigure}

RATE metrics provide a way to summarize the quality of a treatment prioritization rule in ranking units according to potential without specifying a particular treatment policy. This can be especially beneficial when evaluating the degree to which risk-based and other non-centered targeting algorithms optimize treatment benefit.
Choosing the appropriate RATE metric for a problem is important.  As noted in Section~\ref{sec:comparing-weighting-functions}, the precision (and thus, statistical power) of a RATE metric depends on the weighting scheme and the shape of the true, underlying TOC curve. In applications where the analyst anticipates that the prioritization rule will be used to find a relatively small group of individuals with a large treatment benefit, they should use the AUTOC curve; the digital advertising example in Section~\ref{sec:uplift-modeling} illustrates the practical impacts here. However, in cases where treatment effects are more evenly spread out over the population, the Qini coefficient might be more powerful. From our experience so far working with these metrics, we recommend starting with the AUTOC metric if one is unsure which would be more appropriate. We summarize the above discussions of how to choose an effective metric for heterogeneous treatment effects in applications as a flowchart in Figure~\ref{fig:outline}.

\newpage
\appendix

\section*{Supplementary Material}

\vspace{-4mm}

\section{Additional Results}

\subsection{Converse to Proposition \ref{prop:to-Lstat}}
\begin{prop}
\label{prop:from-Lstat}
Suppose that $S(X_i)$ has no ties, i.e., that $F_S(X_i))$ has a uniform distribution on $[0, \, 1]$,
and that \smash{$\tau(x) = \mathbb{E}[Y_i(1) - Y_i(0) \cond X_i = x]$} is uniformly
bounded. Furthermore, suppose that $w: [0, \, 1] \rightarrow \RR$ admits a (Radon-Nikodym) derivative
$w'(u)$ such that $w'(u)(1 - u)$ is absolutely integrable with respect to Lebesgue measure on
$[0, \, 1]$, and that $\int_0^1 w(u) \ du = 0$. Then $\eta_w(S)$ as defined in \eqref{eq:etaw} is
a RATE in the sense of Definition~\ref{def:rate}, with the weight function
\begin{equation*}
\eta_w(S) = \theta_{\alpha_w}(S), \ \ \ \ \alpha_w(t) = -t w'(1-t).
\end{equation*}
\end{prop}
\begin{proof}
We first note that, by the chain rule, \smash{$\eta_w = \mathbb{E}[w\p{1 - F_S(S(X_i))}\tau(X_i)]$}, and furthermore
\begin{align*}
\eta_w(S) &= \int w\p{1 - F_S(S(x))} \tau(x) \, d \mathbb{P}(x) \\
&= w(1) \EE{\tau(X)} - \int \int_{0}^{F_S(S(x))} w'(1 - u) \, du \,  \tau(x) \, d \mathbb{P}(x) \\
&= w(1) \EE{\tau(X)} - \int \int_{0}^{1} 1\p{\cb{u \leq F_S(S(x))}} w'(1 - u) \, du \, \tau(x) \, d\mathbb{P}(x).
\end{align*}
Next, because $(1 - u) w'(u)$ is absolutely integrable and $\tau(x)$ (and thus also the TOC) is uniformly bounded, we can
use Fubini's theorem to continue, and verify that
\begin{align*}
\eta_w(S) &= w(1) \EE{\tau(X)} - \int_0^1 w'(1 - u) \int 1\p{\cb{u \leq F_S(S(x))}} \tau(x) \, d \mathbb{P}(x) \, du \\
&= w(1) \EE{\tau(X)} - \int_0^1 w'(1 - u) \p{1 - u} \p{\TOC(u; \, S) + \EE{\tau(X)}} \, du \\
&= \theta_{\alpha_w}(S) + \EE{\tau(X)} \p{w(1) - \int_0^1 w'(1 - u) \p{1 - u} \, du},
\end{align*}
where $\theta_{\alpha_w}(S)$ is as in Definition~\ref{def:rate}. Finally, integrating by parts, we see that
$\int_0^1 w'(1 - u) \p{1 - u} \, du = \int_0^1 w'(u) u \, du=  w(1) - \int_0^1 w(u) \, du = w(1)$
because $\int_0^1 w(u) \, du = 0$ by assumption, and so $\eta_w(S) = \theta_{\alpha_w}(S)$.
\end{proof}

\subsection{Tiebreaking}
\label{sec:tiebreaking}

So far, we have assumed that the priority score $S(\cdot)$ provides a strict priority ranking.
In many settings of interest, however, $S(\cdot)$ may have ties. In these cases, we break ties by considering
all possible permutations of orderings within each set of tied observations, and then
averaging the resulting TOC estimates. Thus, in the presence of ties, our definition of the TOC
from \eqref{eq:toc} generalizes to
\begin{equation}
\label{eq:toc_tie}
\begin{split}
&\TOC(u; \, S) = 
\frac{P(S(X_i) > q_u)}{u} \EE{Y_i(1) - Y_i(0) \cond S(X_i) > q_u} \\
&\quad\,\quad\,\quad + \left(1 - \frac{P(S(X_i) > q_u)}{u}\right)\EE{1\p{\cb{S(X_i) = q_u}}\p{Y_i(1) - Y_i(0)} } \\
&\quad\,\quad\,\quad -  \EE{Y_i(1) - Y_i(0)},
\end{split}
\end{equation}
where $q_u = \sup\{ q : F_S(q) \le u\}$. Notice that when there are no ties, so that $P(S(X_i) > q_u) = u$, this simplifies to the previous expression for the TOC. Similarly, the sample-average TOC estimator \eqref{eq:TOC_est}
averages the $\hGamma_i$ across tied observations: If $S(X_{i(k)}) > S(X_{i(k+1)}) = \ldots = S(X_{i(k')})$ for
some $k < m < k'$, then
\begin{equation}
\label{eq:htoc_tie}
\hTOC\p{\frac{m}{n}; s} = \frac{1}{m} \p{\sum_{j=1}^{k} \hGamma_{i(j)} + \frac{m - k}{k' - k} \sum_{j=k + 1}^{k'} \hGamma_{i(j)}} - \frac{1}{n} \sum_{i=1}^n \hGamma_i.
\end{equation}
Given these these tie-robust definitions of the TOC, our definition of the RATE and as associated
estimator \eqref{eq:RATE_est} remain as is.

The weighted ATE representation in \eqref{eq:as-Lstat} can also be adapted by taking the average of the weights over all possible tie-breaking orders. That is, by replacing $w(t)$ at any $t$ with $P(S(X_i) = t) = v > 0$ with its average over the ties, 
\begin{equation*}
    \frac{1}{v}\int_0^v w(F_S(t)-\epsilon) \dif{\epsilon}.
\end{equation*}
Similarly, for the empirical estimate, by replacing $w(m/n)$ at any $k < m < k'$ with $S(X_m) = S(X_{k'})$ with 
\begin{equation*}
    \frac{1}{k' - k}\sum_{j = k+1}^{k'} w\p{j/n}.
\end{equation*}

\subsection{A Score for Time-to-event Outcomes}
The AIPW approach for both randomized trials and observational studies can be generalized to handle time-to-event outcomes with right-censored outcomes. To do so, we first need to introduce some additional notation to specify the event times and censoring times.
We denote the counterfactual event times as $T_i(1)$ and $T_i(0)$, and the counterfactual censoring times as $C_i(1)$ and $C_i(0)$. With censored data, we only observe the factual event time $T_i = T_i(W_i)$ up to the factual censoring time $C_i = C_i(W_i)$. That is, the observed data are $U_i = \min\{T_i, C_i\}$ and $\Delta_i = 1\{U_i = T_i\}$.

Two endpoints that are often used with time-to-event outcomes can be easily adapted to the methods developed here.  To study absolute risk, we use the outcome
\begin{equation*}
    Y_i(w) = 1\{ T_i(w) \le t_0\},
\end{equation*}
and to study the restricted mean survival time, we use
\begin{equation*}
    Y_i(w) = \min\{ T_i(w), t_0\},
\end{equation*}
for a pre-specified time $t_0$ that denotes the end of observation. The causal effect estimated in each case is the absolute risk reduction
\begin{equation*}
    P\p{T_i(1) \le t_0 \cond X_i=x} - P\p{T_i(0) \le t_0 \cond X=x},
\end{equation*}
and the difference in restricted mean survival time,
\begin{equation*}
    \EE{ \min\{T_i(1), t_0\} - \min\{T_i(0), t_0\} \cond X=x}.
\end{equation*}

With either of these endpoints, we use the augmented IPW approach described in \cite{tsiatis2007semiparametric}, based off of the construction from \cite{robins1994estimation}, and used as a score for CATE estimation in \cite{cui2023estimating}.
For either of these endpoints, the score is the same. Let $\tilde{U}_i = \min\{U_i, t_0\}$, $\tilde{\Delta}_i = 1\{U_i = T_i~\text{or}~C_i > t_0\}$ and $Y_i = Y_i(W_i)$ as defined above for each endpoint. Then,
the score is
\begin{gather}
    \hGamma_i = \hat{m}(X_i, 1) - \hat{m}(X_i, 0) + \frac{W_i - \hat{e}(X_i)}{(1-\hat{e}(X_i))\hat{e}(X_i)}\left(\frac{\tilde{\Delta}_i Y_i - \p{1-\tilde{\Delta}_i}\hat{q}(\tilde{U}_i, X_i, W_i)}{\hat{S}_C(\tilde{U}_i, X_i, W_i)} \right. \label{eq:dr-score-censoring} \\
    \quad~\quad \left. - \int_0^{\tilde{U}_i} \frac{\hat{q}(s, X_i, W_i)}{\hat{S}_C(s, X_i, W_i)} \dif{\hat{\Lambda}}_C(s, X_i, W_i) - \hat{m}(X_i, W_i) \right),\nonumber
\end{gather}
with
\begin{align*}
    &m(x, w) = \EE{Y_i \cond X_i=x, W_i=w},\\
    &q(u, x, w) = \EE{Y_i \cond X_i=x, U_i \ge u, W_i = w},~\mbox{and}\\
    &S_C(s, x, w) = P(C_i \ge s \cond X_i=x, W_i=w).
\end{align*}
Similarly to the observational study setting, $\hat{e},$ $\hat{m}$, $\hat{q}$, and $\hat{S}_C$ refer to nonparametrically estimated versions of the corresponding functions defined without a $\hat{~}$. The oracle score $\Gamma_i^\ast$ is the equivalent with all of the estimated functions replaced with their corresponding true function.

\subsection{Hypothesis Testing and Confidence Intervals}
\label{sec:hypo-test}
An important aspect of RATE metrics is that they provide a natural null hypothesis for testing whether a prioritization rule is useful for identifying subpopulations with heterogeneous treatment effects. Consider the null hypothesis $H_0$: $\theta_\alpha(S) \le 0$. A test of this hypothesis is sensitive to treatment heterogeneity aligned with the prioritization rule $S$.

Additionally, a test of $H_0$ can be interpreted as a test of treatment effect heterogeneity. If there is no treatment effect heterogeneity, meaning that $\tau(X) = \tau $ is constant, then $\TOC(u; S) = 0$ for any $u \in (0, 1]$, and so $\theta_\alpha(S) = 0$. However, it will not be a powerful test for treatment effect heterogeneity without a good prioritization rule. The statistic $\theta_\alpha(S)$ can be zero if there is no treatment effect heterogeneity, or if the prioritization rule $S$ has no relationship with the heterogeneous treatment effects. Some may find it advantageous, that instead of simply testing whether heterogeneous treatmenet effects exist, it tests whether the treatment effect is predictable by some prioritzation rule. However, those interested in studying whether heterogeneous treatment effects exist regardless of whether or not they can be identified by a prioritization rule should consider tests such as \citet{DingFeMi16}.

Using the sampling distribution for $\what{\theta}$ from Section 3, 
we can easily develop a valid test for $H_0$ using the bootstrap. Let $\sigma^\ast$ denote the standard deviation of the bootstrap estimates $\widehat{\theta}^\ast$. Then, an asymptotically valid $P$-value for $H_0$ is
\begin{equation*}
    2 \Phi\p{-\frac{|\widehat{\theta}|}{\sigma^\ast}},
\end{equation*}
where $\Phi$ is the CDF of the unit normal distribution, $\normal{}(0, 1).$
Similarly, we can use the sampling distribution to derive asymptotically valid confidence intervals for $\theta_\alpha(S)$. To get a $1-\alpha$ confidence interval, we can use
\begin{equation*}
    \what\theta \pm z_{1-\alpha/2} \sigma^\ast,
\end{equation*}
where $z_{1 - \alpha/2} = \Phi^{-1}(1 - \alpha/2)$, and, as above, $\sigma^\ast$ is the standard deviation of the bootstrap estimates $\widehat{\theta}^\ast$.

\section{Proofs}
\label{sec:proofs}

\subsection{Proof of Proposition \ref{prop:to-Lstat}}
\begin{proof}
Let $A = \int_0^1 \alpha(u) \ du.$ Now, note that
\begin{align*}
    \theta_a(S) &= \int_{0}^1 \alpha(u) \EE{ Y_i(1) - Y_i(0) \cond F_S(S(X_i)) \ge 1 - u} \dif{u} - A\EE{Y_i(1) - Y_i(0)} \\
    &= \int_{0}^1 \frac{\alpha(u)}{u} \EE{ \p{Y_i(1) - Y_i(0)} 1\p{\cb{F_S(S(X_i)) \ge 1 - u}}} \dif{u} - A\EE{Y_i(1) - Y_i(0)} \\
    &= \int_{0}^1 \frac{\alpha(u)}{u} \EE{ \tau(X_i) 1\p{\cb{F_S(S(X_i)) \ge 1 - u}}} \dif{u} - A\EE{Y_i(1) - Y_i(0)}.
\intertext{By recalling that we assumed $\tau(X)$ is uniformly bounded and $\frac{\alpha(u)}{u}$ is absolutely integrable on $1 \ge u \ge 1 - F_S(S(X_i))$ for almost every $S(X_i)$, we can apply Fubini's Theorem to get}
    &= \EE{ \p{Y_i(1) - Y_i(0)} \int_{0}^1 \frac{a(u)}{u} 1\p{\cb{F_S(S(X_i)) \ge 1 - u}} \dif{u}} - A\EE{Y_i(1) - Y_i(0)} \\
    &= \EE{ \p{Y_i(1) - Y_i(0)} \int_{1 - F_S(S(X_i))}^1 \frac{a(u)}{u} \dif{u}} - A\EE{Y_i(1) - Y_i(0)}.
\end{align*}
\end{proof}

\subsection{Proof of Theorem~\ref{thm:asymp-linear}}

\begin{proof}
Let $\what{F}_S(q^-)$ denote the empirical distribution at the left limit, instead of the right, and define $\what{G}_S(q) = 1 - \what{F}_S(q^-)$, which is the empirical distribution function of the observations in the reverse order.
Note that the quantile $j/n$ of the $i(j)$-th example is simply \smash{$\what{G}_S(Q_{i(j)})$}.
Then, we decompose $\sqrt{n}(\what{\theta} - \theta)$ into two leading terms that we show are asymptotically linear, followed by remainder terms that we show vanish under Assumptions A-C.
\begin{align*}
    \frac{1}{\sqrt{n}} \sum_{i=1}^n w_n(\what{G}_S(Q_i)) \what{\Gamma}_i -\theta &= \underbrace{\frac{1}{\sqrt{n}} \sum_{i=1}^n w(1-F_S(Q_i) (\Gamma^\ast_i - \overline{\tau}(Q_i))}_{\text{Term 1}}
    \\
    &~~~~+ \underbrace{\frac{1}{\sqrt{n}} \sum_{i=1}^n w(\what{G}_S(Q_i)) \overline{\tau}(Q_i) - \theta}_{\text{Term 2}}
    \\
    &~~~~+ \underbrace{\frac{1}{\sqrt{n}} \sum_{i=1}^n (w_n(\what{G}_S(Q_i)) - w(\what{G}_S(Q_i))) \overline{\tau}(Q_i)}_{\text{Term 3}}
    \\
    &~~~~+ \underbrace{\frac{1}{\sqrt{n}} \sum_{i=1}^n (w_n(\what{G}_S(Q_i)) - w(1-F_S(Q_i))) (\Gamma^\ast_i - \overline{\tau}(Q_i))}_{\text{Term 4}}
    \\
    &~~~~+ \underbrace{\frac{1}{\sqrt{n}} \sum_{i=1}^n w_n(\what{G}_S(Q_i)) \delta_i.}_{\text{Term 5}}
\end{align*}

\noindent \textbf{Term 1}

The first term is trivially a sum of iid random variables.

\noindent \textbf{Term 2}

The second term (which we will denote as $T_n$)
is an $L$-statistic \citep{VanDerVaart98,ShorackWe09} with $\overline{\tau}(q)$ of bounded variation. To show that it is asymptotically linear, we use the following assumptions and theorem from \citet{ShorackWe09}:

\setcounter{assumption}{3}
\begin{assumption}{[Assumption 1, \citet{ShorackWe09}]}
(a) The weight function $w : (0, 1) \to \R$ satisfies $|w(1-t)| \le B(t) = M t^{-b_1} (1-t)^{-b_2}$ for $M < \infty$ and $\max\{b_1, b_2\} < 1$, and (b) $\overline{\tau}(q) = \overline{\tau}_+(q) - \overline{\tau}_-(q)$, each monotone non-decreasing, with $|\overline{\tau}_\pm(F_S^{-1}(t))| \le D(t) = Mt^{-d_1}(1-t)^{-d_2}$.
\label{assume:sw-growth}
\end{assumption}

\begin{assumption}{[Assumption 2', \citet{ShorackWe09}]}
(a) $w$ has a continuous derivative on $(0, 1)$ satisfying $|w'(1-t)| \le \frac{B(t)}{t(1-t)}$ for $B$ defined as in Assumption~\ref{assume:sw-growth}.
\label{assume:sw-smooth}
\end{assumption}

\begin{theo}{[Theorem 1, \citet{ShorackWe09}]}
Suppose Assumptions~\ref{assume:sw-growth} and~\ref{assume:sw-smooth} hold with $\alpha = \max\{b_1 + d_1, b_2 + d_2\} < 1/2$. Then,
\begin{equation*}
    \sqrt{n}(T_n - \theta) = -\int_{0}^1 w(1-t) \mathbb{U}_n(t) \dif{h(F_S^{-1}(t))} + o_P(1).
\end{equation*}
\label{thm:sw-l-stat}
\end{theo}

Assumption~\ref{assume:sw-growth}(a) follows directly from Assumption~\ref{assume:weight-fn}. From Assumption~\ref{assume:regularity}, we know that $\overline{\tau}(q)$ is of bounded variation, and that $\E[|\overline{\tau}(Q)|^r] < \infty$, which implies that we can write $\overline{\tau}(q) = \overline{\tau}_+(q) - \overline{\tau}_-(q)$, with $\E[|\overline{\tau}_\pm(Q)|^r] < \infty$. Now, Markov's inequality implies that for $s > 0$,
\begin{equation*}
    P( |\overline{\tau}_{\pm}(Q)| \ge s ) \le \frac{\E[|\overline{\tau}_{\pm}(Q)|^r]}{s^r}.
\end{equation*}
Inverting this statement, using the fact that $q \mapsto \overline{\tau}_\pm(q)$ is monotone, we see that Assumption~\ref{assume:sw-growth}(b) holds with $d_1 = d_2 = 1/r$. As stipulated in Assumption~\ref{assume:regularity}, $\alpha = 1/r + b < 1/2$. Therefore, $T_n$ satisfies the assumptions of Theorem~\ref{thm:sw-l-stat}, implying that it is asymptotically linear with influence function
\begin{equation*}
    Y_i = \int_{Q_i}^\infty w(1-F(q)) \dif{\overline{\tau}(q)} - \theta.
\end{equation*}

The rest of the proof involves showing the remaining terms are asymptotically negligible.

\noindent \textbf{Term 3}

For any fixed $\epsilon > 0$, applying Markov's inequality followed by H\"older's inequality gives
\begin{align}
    &P\left( \bigg|\frac{1}{\sqrt{n}}\sum_{i=1}^n \left(w_n(\what{G}_S(Q_i)) - w(\what{G}_S(Q_i))\right) \overline{\tau}(Q_i))\bigg| \ge \epsilon \right) \nonumber
    \\
    &~~~~~\le
    \frac{\E\left[\bigg|\sum_{i=1}^n \left(w_n(\what{G}_S(Q_i)) - w(\what{G}_S(Q_i))\right) \frac{\overline{\tau}(Q_i)}{\sqrt{n}}\bigg| \right]}{\epsilon} \nonumber
    \\
    &~~~~~\le
    \frac{\E\left[ \left(\sum_{i=1}^n \left(w_n(\what{G}_S(Q_i)) - w(\what{G}_S(Q_i))\right)^p \right)^{1/p} \left(\frac{1}{n^{r/2}}\sum_{i=1}^n |\overline{\tau}|^r(Q_i) \right)^{1/r} \right]}{\epsilon} \nonumber
\end{align}
Assumption~\ref{assume:weight-fn} gives that $\left(\sum_{i=1}^n \left(w_n(\what{G}_S(Q_i)) - w(\what{G}_S(Q_i))\right)^p \right) < C$ eventually.
Applying Jensen's inequality shows that $\E\left[\left(\frac{1}{n}\sum_{i=1}^n |\overline{\tau}|^r(Q_i) \right)^{1/r} \right] \le \E[ \overline{\tau}^r(Q_i) ]^{1/r}$, which is bounded by Assumption~\ref{assume:regularity}. Therefore, the numerator is $O(n^{1/2 - 1/r}) = o(1)$.

\noindent \textbf{Term 4}

For any fixed $\epsilon > 0$, applying Chebyshev's inequality conditionally on $\{Q_i\}_{i=1}^n$ shows that the fourth term satisfies
\begin{align}
    &P\left( \bigg|\frac{1}{\sqrt{n}}\sum_{i=1}^n \left(w_n(\what{G}_S(Q_i)) - w(1-F_S(Q_i))\right) (\Gamma^\ast_i - \overline{\tau}(Q_i))\bigg| \ge \epsilon \right) \nonumber
    \\
    &~~~~~=
    \E\left[P\left( \bigg| \frac{1}{\sqrt{n}}\sum_{i=1}^n \left(w_n(\what{G}_S(Q_i)) - w(1-F_S(Q_i))\right) (\Gamma^\ast_i - \overline{\tau}(Q_i)) \bigg| \ge \epsilon \mid \{Q_i\}_{i=1}^n \right)\right] \nonumber
    \\
    &~~~~~=
    \E\left[ \frac{1}{\epsilon^2} \var\left( \frac{1}{\sqrt{n}}\sum_{i=1}^n \left(w_n(\what{G}_S(Q_i)) - w(1-F_S(Q_i))\right) (\Gamma^\ast_i - \overline{\tau}(Q_i)) \mid \{Q_i\}_{i=1}^n \right)\right] \label{eq:conditional-var-1}
    \\
    &~~~~~=
     \frac{1}{\epsilon^2} \E\left[ \frac{1}{n}\sum_{i=1}^n \left(w_n(\what{G}_S(Q_i)) - w(1-F_S(Q_i))\right)^2 \var(\Gamma_i^\ast \mid Q_i)  \right] \label{eq:conditional-var-2}
    \\
    &~~~~~\le
     \frac{C_g}{\epsilon^2} \E\left[ \frac{1}{n}\sum_{i=1}^n \left(w_n(\what{G}_S(Q_i)) - w(1-F_S(Q_i))\right)^2 \right], \nonumber
\end{align}
where \eqref{eq:conditional-var-1} to \eqref{eq:conditional-var-2} follows from the rule that $\var(aX + bY) = a^2 \var(X) + b^2 \var(Y),$ when $X$ and $Y$ are independent, and noting that each term in the sum is independent given $\{Q_i\}_{i=1}^n$.
By assumption, $\E[ \tfrac{1}{n}\sum_{i=1}^n (w_n(\what{G}_S(Q_i)) - w(1-F_S(Q_i)))^2 ] \to 0$.
The last inequality follows from Assumption~\ref{assume:regularity}.

\noindent \textbf{Term 5}

Finally, we control the last term using Assumption~\ref{assume:nuisance}. We begin by splitting the sum up into the $K$ partitions,
\begin{equation*}
    \frac{1}{\sqrt{n}}\sum_{i=1}^n w_n(\what{G}_S(Q_i)) \delta_i = \sum_{k=1}^K \frac{1}{\sqrt{n}} \sum_{i \in I_k} w_n(\what{G}_S(Q_i)) \delta_i,
\end{equation*}
and note that because these partitions are selected uniformly at random, it suffices to show that
\begin{align*}
    P\left( \bigg| \frac{1}{\sqrt{n}}\sum_{i \in I_1} w_n(\what{G}_S(Q_i)) \delta_i \bigg| \ge \epsilon \right) \to 0
\end{align*}
for any $\epsilon > 0$. Recalling $G_1$ from Assumption~\ref{assume:nuisance}, we divide this event into a piece that happens on $G_1$ and a piece that happens on $G_1^C$,
\begin{align*}
    P\left( \bigg| \frac{1}{\sqrt{n}}\sum_{i \in I_1} w_n(\what{G}_S(Q_i)) \delta_i \bigg| \ge \epsilon \right) \le  P\left( \bigg| \frac{1}{\sqrt{n}}\sum_{i \in I_1} w_n(\what{G}_S(Q_i)) \delta_i \bigg| \ge \epsilon \mid G_1 \right) P\left( G_1 \right) + P\left(G_1^C\right).
\end{align*}
Because $P(G_1) \to 1$, we need to show that 
\begin{align*}
    P\left( \bigg| \frac{1}{\sqrt{n}}\sum_{i \in I_1} w_n(\what{G}_S(Q_i)) \delta_i \bigg| \ge \epsilon \mid G_1 \right) \to 0.
\end{align*}
Denote $\{Q_i\}_{i \in I_1}, \{(W_i, X_i, Y_i)\}_{i \not\in I_1}$ as $B_1$. Using $\cas$ here to denote convergence almost everywhere on $G_1$, it suffices to show that
\begin{equation*}
    P\left( \bigg| \frac{1}{\sqrt{n}}\sum_{i \in I_1} w_n(\what{G}_S(Q_i)) \delta_i \bigg| \ge \epsilon \mid B_1, G_1 \right) \cas 0,
\end{equation*}
because the bounded convergence theorem implies that conditional convergence implies unconditional convergence.

If we replace $w_n(\what{G}_S(Q_i))$ with $w(1-F_S(Q_i))$ this would follow easily by Chebyshev's inequality from the bias and variance conditions of Assumption~\ref{assume:nuisance}:
\begin{align*}
    |\E\left[ \frac{1}{\sqrt{n}}\sum_{i \in I_1}^n w(1-F_S(Q_i)) \delta_i \mid G_1, B_1 \right]| \le |\sqrt{n}\E[ w(1-F_S(Q_i)) \delta_i \mid B_1, G_1 ]| \to 0,
\end{align*}
and, because we assumed in Assumption~\ref{assume:nuisance} that $\{\delta_i\}_{i \in I_1}$ are independent conditionally on $B_1, G_1$,
\begin{align*}
    \var\left( \frac{1}{\sqrt{n}}\sum_{i \in I_1}^n w(1-F_S(Q_i)) \delta_i \mid G_1, B_1 \right) &= \frac{1}{n}\sum_{i \in I_1} \var\left( w(1-F_S(Q_i)) \delta_i \mid G_1, B_1 \right) \to 0.
\end{align*}
This, along with the following lemma proves that Term 5 converges to zero.
\begin{lemm}
\label{lem:approx-error-weights}
Under Assumptions~\ref{assume:weight-fn} and~\ref{assume:nuisance},
\begin{equation}
   P\left(\bigg| \frac{1}{\sqrt{n}}\sum_{i \in I_1} (w_n(\what{G}_S(Q_i))-w(1-F_S(Q_i))) \delta_i \bigg| \ge \epsilon \mid B_1, G_1 \right) \cas 0.
   \label{eq:approx-error-weights}
\end{equation}
\end{lemm}
Proving Lemma~\ref{lem:approx-error-weights} is similar to how we control Terms 2-4 in in the above decomposition, however now with an $L$-statistic whose variance goes to 0. See Section~\ref{proof:approx-error-weights} in the Supplementary Materials for the details.
\end{proof}

\subsection{Proof of Lemma\ref{lem:half-boot}}
\label{proof:bootstrap}

\begin{proof}
Throughout, we will assume that $n$ is divisible by two for notational simplicity.
Because $\sqrt{n}(\what{\beta} - \beta) = \frac{1}{\sqrt{n}}\sum_{i=1}^n \psi(Z_i) + o_P(1)$, when fit with $n/2$ samples, we have that $\sqrt{n}(\what{\beta}^\ast - \beta) = \sqrt{n}\frac{2}{n}\sum_{i=1}^{n/2} \psi(Z_i) + o_P(1)$.

Write $\sqrt{n}(\what{\beta}^\ast - \what{\beta})$ as the difference between the two terms above,
\begin{align}
    \sqrt{n}(\what{\beta}^\ast - \what{\beta}) &= \sqrt{n}(\what{\beta}^\ast - \beta) - \sqrt{n}(\what{\beta} - \beta) \nonumber \\
    &= \frac{1}{\sqrt{n}}\sum_{i=1}^{n} V_i \psi(Z_i) + o_P(1),
    \label{eq:corr-mult-boot}
\end{align}
where $V_i$ is $1$ if the $i$-th example is in the subset of $n/2$ examples in the bootstrap sample, and otherwise is $-1$.

Our goal is to show that this is (conditionally) asymptotically normal with distribution $\normal{}(0, \var(\psi(Z_i))$. Because $\sqrt{n}(\what{\beta} - \beta) \cd Z,$ with $Z \sim \normal{}(0, \var(\psi(Z_i))$, this will prove the result.

To show that this is true, we first consider a ``binomialized'' version of the term \eqref{eq:corr-mult-boot}.\footnote{This technique is inspired by the Poissonization proof technique for the standard empirical bootstrap as in \cite[Chp. 3.6]{VanDerVaartWe96}.}
Let $N_n$ be a random draw from a $\operatorname{Ber}(n, \half)$ distribution. If $N_n > n/2$, randomly choose $N_n - n/2$ samples with $V_i = -1$ and set $\tV_i = 1$. If $N_n < n/2$, randomly choose $n/2 - N_n$ samples with $V_i = 1$ and set $\tilde{V}_i = -1$. Set $\tV_i = V_i$ for all other samples. With this construction, $\{\tV_i\}_{i=1}^n$ are i.i.d. Rademacher random variables,
\begin{equation*}
    \tV_i = \begin{cases}
    1 & \text{w.p.}~\half \\
    -1 & \text{w.p.}~\half.
    \end{cases}
\end{equation*}
Therefore, the multiplier central limit theorem \citep[Lemma 2.9.5]{VanDerVaartWe96} implies that
\begin{equation*}
    \frac{1}{\sqrt{n}}\sum_{i=1}^{n} \tV_i \psi(Z_i)
\end{equation*}
converges to $\normal{}(0, \var(\psi(Z_i))$, conditionally on $(Z_i)_{i=1}^n$ almost surely.

To complete the proof, we need to show that conditionally (almost surely) on $(Z_i)_{i=1}^n$,
\begin{equation*}
    R_n(N_n) = \frac{1}{\sqrt{n}}\sum_{i=1}^{n} (\tV_i - V_i) \psi(Z_i) = o_P(1).
\end{equation*}
Towards this goal, first notice that $V_i - \tV_i$ is only nonzero on the $|N_n - n/2|$ samples changed in the construction of $\tV_i$. Let $I_{n}$ be the indices of these samples. Then, observe that
\begin{align*}
    R_n(N_n) = \sign{}(N_n - n/2) \frac{2}{\sqrt{n}} \sum_{i \in I_n} \psi(Z_i).
\end{align*}
Inspecting the symmetry of $R_n(N_n)$, we can see that because $N_n$ is symmetric about $n/2$, $\aEE{R_n(N_n)} = 0$, with $\aEE{\cdot}$ shorthand for $\EE{\cdot \cond \{Z_i\}_{i=1}^n}.$ Our approach, then, is to show that $\condvar{}(R_n(N_n)) \cas 0$, where $\condvar{}(\cdot) = \var(\cdot \mid \{Z_i\}_{i=1}^n)$ which by Chebyshev's inequality, implies $P(|R_n(N_n)| > \epsilon \cond \{Z_i\}_{i=1}^n) \to 0,$ as claimed.

To bound the variance, we start by applying the law of total variance,
\begin{equation*}
    \condvar{}(R_n(N_n)) = \condvar{}\left(\aEE{R_n(N_n) \cond N_n}\right) + \aEE{\condvar{}(R_n(N_n) \cond N_n)},
\end{equation*}
and note that conditional on $N_n$, the remaining randomness is the choice of $|N_n - n/2|$ samples in $I_n$. In this way,
\begin{equation*}
    \sum_{i \in I_n} \psi(Z_i)
\end{equation*}
is naturally interpretable as $A_n \defeq |N_n - n/2|$ times the sample average of $A_n$ randomly drawn points $\psi(Z_i)$ from the finite population $\{\psi(Z_i)\}_{i=1}^n$. With this in mind,
\begin{align*}
    \aEE{\sum_{i \in I_n} \psi(Z_i) \cond N_n} &= \frac{A_n}{n} \sum_{i=1}^n \psi(Z_i)\\
    \condvar{}\left( \sum_{i \in I_n} \psi(Z_i) \cond N_n \right) &= A_n\frac{n - A_n}{n}\left(\frac{1}{n-1} \sum_{i=1}^n \psi^2(Z_i) - \left( \frac{1}{n-1} \sum_{i=1}^n \psi(Z_i)\right)^2  \right).
\end{align*}
Plugging these expressions into $R_n(N_n)$ gives
\begin{align}
    \aEE{R_n(N_n) \cond N_n} &= \frac{2}{\sqrt{n}} \frac{N_n - n/2}{n} \sum_{i=1}^n \psi(Z_i) \nonumber \\ 
\intertext{and so}
    \condvar{}\left(\aEE{R_n(N_n) \cond N_n}\right) &= \frac{4}{n} \condvar{}(N_n) \left( \frac{1}{n} \sum_{i=1}^n \psi(Z_i) \right)^2 = \left( \frac{1}{n} \sum_{i=1}^n \psi(Z_i) \right)^2
    \label{eq:exp-than-var}
\end{align}
and
\begin{equation}
    \aEE{\condvar{}\left(R_n(N_n) \cond N_n\right)} = \aEE{\frac{4A_n(n - A_n)}{n^2}}\underbrace{\left(\frac{1}{n-1} \sum_{i=1}^n \psi^2(Z_i) - \left( \frac{1}{n-1} \sum_{i=1}^n \psi(Z_i)\right)^2  \right)}_{s_\psi^2}.
    \label{eq:var-than-exp}
\end{equation}
Term \eqref{eq:exp-than-var} converges to zero almost surely by the Strong Law of Large Numbers. Term \eqref{eq:var-than-exp} also converges to zero almost surely, because $(n - A_n) / n \le 1$, $\aEE{A_n/n} \asymp 1/\sqrt{n}$, and $s_\psi^2 \to \var(\psi(Z_i))$ almost surely by the Strong Law.
\end{proof}

\subsection{Proof of Lemma~\ref{lem:unif-cts}}
\label{proof:unif-cts}
\begin{proof}
Define $\what{G}_S(q)$ as in the proof of Theorem~\ref{thm:asymp-linear} in Section~\ref{sec:proofs}, and let $G_S(q)$ be the corresponding distribution function of $S(X_i)$ in reverse order.
By continuity, we know that $|w(\what{G}_S(S(X_i))) - w(1-F_S(S(X_i)))| \le \rho_w(|\what{G}_S(S(X_i)) - G_S(S(X_i)|)$, where $\rho_w$ is the modulus of continuity of $w$. By the monotonicity of $\rho_w$, $\rho_w(|\what{G}_S(S(X_i) - G_S(S(X_i)|) \le \rho_w(\|\what{G} - G\|_{\infty}))$. Therefore,
$\frac{1}{n} \sum_{i=1}^n (w(\what{G}_S(S(X_i))) - w(G_S(S(X_i))))^2 \le \rho_w^2(\|\what{G} - G\|_{\infty})$. The Glivenko-Cantelli Theorem shows that $\|\what{G} - G\|_{\infty} \to 0$ almost surely, and so the claimed statement holds due to the bounded convergence theorem, as $\what{G}$ and $G$ are bounded in $[0,1]$, implying that $\rho_w(\|\what{G} - G\|_\infty)$ is bounded, as well.
\end{proof}

\subsection{Proof of Proposition~\ref{prop:autoc-weight}}
\label{proof:autoc-weight}
\begin{proof}
Define $\what{G}_S(q)$ as in the proof of Theorem~\ref{thm:asymp-linear} in Section~\ref{sec:proofs}, and let $G_S(q)$ be the corresponding distribution function of $S(X_i)$ in reverse order.
Writing
\begin{align*}
    &\E[\frac{1}{n}\sum_{i=1}^n (w_n(\what{G}_S(S(X_i))) - w(G_S(S(X_i))))^2] \\
    &\le \frac{2}{n}\underbrace{\sum_{i=1}^n (w_n(i/n) - w(i/n))^2}_{\text{Part 1}} + \underbrace{\E\left[ \frac{2}{n}\sum_{i=1}^n (-\log(\what{G}_S(Q_i)) + \log(G_S(Q_i)))^2 \right]}_{\text{Part 2}},
\end{align*}
we verify that each of these terms vanishes in the following steps:
\begin{enumerate}
    \item Showing that Part 1 is $O(1)$, then
    \item showing that Part 2 vanishes.
\end{enumerate}

\textbf{Step 1:}
Recall that $w_n(i/n) = H_n - H_i - 1$ and $w(i/n) = \log(n) - \log(i) - 1$. Noting that $|H_n - H_k - \log(n) + \log(k)| \le  \frac{1}{2k}$ for $1 \le k \le n$, averaging over these differences gives
\begin{equation*}
    \sum_{i=1}^n (w_n(i/n) - w(i/n))^2 \le \sum_{i=1}^n \frac{1}{(2i)^2} \le \frac{\pi^2}{24}.
\end{equation*}
This verifies the second part of Assumption~\ref{assume:weight-fn}.
Additionally, $\frac{1}{n}\sum_{i=1}^n (w_n(i/n) - w(i/n))^2 \to 0$, showing that the first term of the above decomposition vanishes.

\textbf{Step 2:}
In this step, we bound the squared difference between $w \circ \what{G}_S$ and $w \circ G_S$ in the following sense:
\begin{equation*}
  \E\left[ \frac{1}{n}\sum_{i=1}^n (-\log(\what{G}_S(Q_i)) + \log(G_S(Q_i)))^2 \right] = o(1).
\end{equation*}

We make the following observation, with proof in the following section, that shows that it is sufficient to control $(-\log(\what{G}_S(Q_i)) + \log(G_S(Q_i)))^2$ on the event $G_S(Q_i) \ge n^{\epsilon-1} / 2, \what{G}_S(Q_i) \ge n^{\epsilon - 1}$.
\begin{lemm}
\label{lem:log-terms}
For any $1/2 < \epsilon < 1$,
\begin{align*}
    &\E\left[ \frac{1}{n}\sum_{i=1}^n (-\log(\what{G}_S(Q_i)) + \log(G_S(Q_i)))^2 \right] \\
    &= \underbrace{\E\left[\indic{\what{G}_S(Q_i) \ge n^{\epsilon - 1}, G_S(Q_i) \ge n^{\epsilon-1} / 2}(-\log(\what{G}_S(Q_i)) + \log(G_S(Q_i)))^2 \right]}_{(\ast)} + o(1).
\end{align*}
\end{lemm}

To control the expectation $(\ast)$, we do a Taylor expansion and apply the DKW inequality. For some $0 < s < 1$,
\begin{align*}
    (-\log(\what{G}_S(Q_i)) + \log(G_S(Q_i)))^2 &\le \left( -\frac{\what{G}_S(Q_i) - G_S(Q_i)}{\what{G}_S(Q_i)} - \frac{(\what{G}_S(Q_i) - G_S(Q_i))^2}{(s\what{G}_S(Q_i) + (1-s)G_S(Q_i))^2} \right)^2.
\intertext{On the event $\{\what{G}_S(Q_i) \ge n^{\epsilon - 1}, G_S(Q_i) \ge n^{\epsilon-1} / 2\}$, this is bounded by}
    &\le n^{2(1-\epsilon)}(\what{G}_S(Q_i) - G_S(Q_i))^2 \\
    &\quad\quad + 16n^{4(1-\epsilon)}(\what{G}_S(Q_i) - G_S(Q_i))^4
    \\
    &\le n^{2(1-\epsilon)}\|\what{G}_S - G_S\|_\infty^2 + 16n^{4(1-\epsilon)}\|\what{G}_S - G_S\|_\infty^4.
\end{align*}
\cite{DvoretzkyKiWo56} shows that $P(\|\what{G}_S - G_S\|_\infty > t) \le 2\exp(-Cnt^2)$ for some $0 < C < \infty$. Using that for a random variable $V > 0$ and exponent $r \ge 1$, we have $\E[V^r] \le \int_0^\infty t^{r-1} P(V > t) \dif{t}$, we can bound
\begin{align*}
    (\ast) &\le \E[ n^{2(1-\epsilon)}\|\what{G}_S - G_S\|_\infty^2 + 16n^{4(1-\epsilon)}\|\what{G}_S - G_S\|_\infty^4]
    \\
    &\le n^{2(1-\epsilon)} \int_0^\infty t \exp(-Cnt^2) \dif{t} + 16n^{4(1-\epsilon)} \int_0^\infty t^3 \exp(-Cnt^2) \dif{t}
    \\
    &\le
    \frac{n^{2(1 - \epsilon)}}{2Cn} + \frac{16n^{4(1 - \epsilon)}}{2C^2n^2} = \frac{1}{2C}n^{1 - 2\epsilon} + \frac{8}{C^2}n^{2 - 4\epsilon}.
\end{align*}
This vanishes as $n \to \infty$ so long as $\epsilon > 1/2$.
\end{proof}

\subsection{Proof of Lemma~\ref{lem:log-terms}}
\label{proof:log-terms}
\begin{proof}
We continue to use the notation for $\what{G}_S(q)$ as in the proof of Theorem~\ref{thm:asymp-linear} in Section~\ref{sec:proofs}, and use $G_S(q)$ be the corresponding distribution function of $S(X_i)$ in reverse order. First, we split the term into two parts, 
\begin{align*}
    &\E\left[ \frac{1}{n}\sum_{i=1}^n (-\log(\what{G}_S(Q_{i})) + \log(G_S(Q_{i})))^2 \right] \\
    &=\underbrace{\E\left[ \frac{1}{n}\sum_{j=1}^{n^{\epsilon}} (-\log(\what{G}_S(Q_{i(j)})) + \log(G_S(Q_{i(j)})))^2 \right]}_{\eqdef R_1} + \E\left[ \frac{1}{n}\sum_{j=n^{\epsilon}}^n (-\log(\what{G}_S(Q_{i(j)})) + \log(G_S(Q_{i(j)})))^2 \right].
\end{align*}

The first step of the proof is to show that $R_1 = o(1)$, and the second step is to show that additionally conditioning the second term on $G_S(Q_{i(j)}) > n^{\epsilon - 1} / 2)$ has an asymptotically negligible effect on the expectation.

\textbf{Step 1:} 
Let $U_i$ be iid uniform random variables on $[0, 1]$, and $U_{i(1)}$ be the minimum of $U_1, \dots, U_n$. The key to this step is a change of variables from $Q_i$ to $U_i$. Let $f_{U_{i(1)}}(u)$ be it's density,
\begin{equation*}
    f_{U_{i(1)}}(u) = n(1-u)^{n - 1}.
\end{equation*}
By the monotonicity of the order statistics $Q_{i(j)}$, $t \mapsto G_S(t)$ and $t \mapsto \log t$,
\begin{align*}
    \E[ \log^2(G_S(Q_{i(j)}))]
    &\le \E[ \log( G_S(Q_{i(1)}))]
    \\
    &= \int \log^2(u) f_{U_{i(1)}}(u) \dif{u}
    \\
    &= \frac{6n H_n^2 - 6n \psi^{(1)}(n+1) + 6n\pi^2}{6n} \le H_n^2 + \pi^2.
\end{align*}

Along with the fact that at all of the observed samples $\{Q_i\}_{i=1}^n$, used to construct $\what{G}_S(\cdot)$, $1/n \le \what{G}_S(Q_i) \le 1$, this allows us to bound $R_1$ as
\begin{align*}
    R_1 &= \E\left[ \frac{1}{n}\sum_{j=1}^{n^{\epsilon}} (-\log(\what{G}_S(Q_{i(j)})) + \log(G_S(Q_{i(j)})))^2 \right]
    \\
    &\le
    \E\left[  \frac{1}{n}\sum_{j=1}^{n^{\epsilon}} 2\log^2(\what{G}_S(Q_{i(j)})) + 2\log^2(G_S(Q_{i(j)})) \right]
    \\
    &\le \frac{2 \log^2(n) + 2H_n^2 + 2\pi^2}{n^{1-\epsilon}} \to 0
\end{align*}
because $\epsilon < 1$.

\textbf{Step 2:}
Notice that the sum $\E\left[ \frac{1}{n}\sum_{j=n^{\epsilon}}^n (-\log(\what{G}_S(Q_{i(j)})) + \log(G_S(Q_{i(j)})))^2 \right]$ is over all indices $i$ such that $\what{G}_S(Q_i) \ge n^{\epsilon-1}$, therefore we can rewrite the sum as
\begin{align*}
    &\E\left[ \frac{1}{n}\sum_{j=n^{\epsilon}}^n (-\log(\what{G}_S(Q_{i(j)})) + \log(G_S(Q_{i(j)})))^2 \right]
    \\
    &= \E\left[ \frac{1}{n}\sum_{i=1}^n \indic{\what{G}_S(Q_{i}) \ge n^{\epsilon-1}} (-\log(\what{G}_S(Q_{i})) + \log(G_S(Q_{i})))^2 \right],
\end{align*}
so that each term of the sum is identically distributed. By linearity of expectation, this is equal to
\begin{align*}
    \E\left[ \indic{\what{G}_S(Q_{i}) \ge n^{\epsilon-1}} (-\log(\what{G}_S(Q_{i})) + \log(G_S(Q_{i})))^2 \right]
\end{align*}
Splitting this expectation into two events, $\{G_S(Q_{i}) \ge n^{\epsilon-1}/2 \}$ and $\{G_S(Q_{i}) < n^{\epsilon-1}/2 \}$, we show that the latter event is rare enough to contribute negligibly to the expectation:
\begin{align*}
    &\E\left[ \indic{\what{G}_S(Q_{i}) \ge n^{\epsilon-1}, G_S(Q_{i}) < n^{\epsilon-1}/2} (-\log(\what{G}_S(Q_{i})) + \log(G_S(Q_{i})))^2 \right]
    \\
    &\le \E\left[ 2 \log^2(n) + 2\log^2(U) \mid U \le n^{\epsilon-1}/2  \right] P\left(\what{G}_S(Q_{i}) \ge n^{\epsilon-1}, G_S(Q_{i}) < n^{\epsilon-1}/2 \right)
    \\
    &= O(\log(n) + \log^2(n)) P\left(\what{G}_S(Q_{i}) \ge n^{\epsilon-1}, G_S(Q_{i}) < n^{\epsilon-1}/2 \right)
    \\
    &\le O(\log(n) + \log^2(n)) P\left(\|\what{G}_S(\cdot) - G_S(\cdot)\|_\infty > n^{\epsilon-1}/2 \right)
\end{align*}
\cite{DvoretzkyKiWo56} shows that $P(\|\what{G}_S - G_S\|_\infty > t) \le 2\exp(-Cnt^2)$ for some $0 < C < \infty$, so the above is bounded by
\begin{equation*}
    O(\log(n) + \log^2(n)) \exp(-C n^{2\epsilon-1}/4)
\end{equation*}
which converges to $0$ as $n \to \infty$, because $\epsilon > 1/2$.
\end{proof}

\subsection{Proof of Lemma~\ref{lem:approx-error-weights}}
\label{proof:approx-error-weights}
\begin{proof}
Throughout this proof, we will do everything implicitly conditional on $B_1$ and $G_1$.
We split the term \eqref{eq:approx-error-weights} into three terms,
\begin{align*}
    &\frac{1}{\sqrt{n}}\sum_{i \in I_1} (w_n(\what{G}_S(Q_i))-w(1-F_S(Q_i))) \delta_i \\
    &\quad~\quad~\quad~\quad~\quad~\quad~= \underbrace{\frac{1}{\sqrt{n}}\sum_{i \in I_1} (w_n(\what{G}_S(Q_i))-w(1-F_S(Q_i))) (\delta_i - \E[\delta_i \mid Q_i])}_{\text{Term 1}} \\
    &\quad~\quad~\quad~\quad~\quad~\quad~\quad~+ \underbrace{\frac{1}{\sqrt{n}}\sum_{i \in I_1} (w_n(\what{G}_S(Q_i))-w(\what{G}_S(Q_i))) \E[\delta_i \mid Q_i])}_{\text{Term 2}}\\
    &\quad~\quad~\quad~\quad~\quad~\quad~\quad~+ \underbrace{\frac{1}{\sqrt{n}}\sum_{i \in I_1} (w(\what{G}_S(Q_i))-w(1-F_S(Q_i))) \E[\delta_i \mid Q_i]).}_{\text{Term 3}}
\end{align*}
Under the assumptions of this Lemma, each of these terms converges in probability to zero, as we now show. The first term is mean zero, and therefore, that it is negligible follows from Chebyshev's inequality and Assumption~\ref{assume:weight-fn} 
\begin{align*}
    &P\left(\left| \frac{1}{\sqrt{n}}\sum_{i \in I_1} (w_n(\what{G}_S(Q_i))-w(1-F_S(Q_i))) (\delta_i - \E[\delta_i \mid Q_i])\right| > \epsilon \right) \\
    &~\quad~\quad~\le \frac{\var\left(\frac{1}{\sqrt{n}}\sum_{i \in I_1} (w_n(\what{G}_S(Q_i))-w(1-F_S(Q_i))) (\delta_i - \E[\delta_i \mid Q_i])\right)}{\epsilon^2} \\
    &~\quad~\quad~= \frac{\E\left[\var\left(\frac{1}{\sqrt{n}}\sum_{i \in I_1} (w_n(\what{G}_S(Q_i))-w(1-F_S(Q_i))) (\delta_i - \E[\delta_i \mid Q_i])\mid \{Q_i\}_{i=1}^n \right)\right]}{\epsilon^2} \\
    &~\quad~\quad~= \frac{\E[\frac{1}{n}\sum_{i \in I_1} (w_n(\what{G}_S(Q_i))-w(1-F_S(Q_i)))^2 \E[\delta_i^2\mid Q_i]] }{\epsilon^2} \\
    &~\quad~\quad~\le \frac{C}{\epsilon^2}\E\left[\frac{1}{n}\sum_{i \in I_1} (w_n(\what{G}_S(Q_i))-w(1-F_S(Q_i)))^2 \right] \to 0.
\end{align*}

Similarly, the second term follows from applying H\"older's inequality to show that the second moment of the second term goes to zero, and applying Markov's inequality,
\begin{align*}
    &P\left(\left| \frac{1}{\sqrt{n}}\sum_{i \in I_1} (w_n(\what{G}_S(Q_i))-w(\what{G}_S(Q_i))) \E[\delta_i \mid Q_i]) \right| > \epsilon^2 \right) \\
    &~\quad~\quad~\le \frac{\E\left[\left|\frac{1}{\sqrt{n}}\sum_{i \in I_1} (w_n(\what{G}_S(Q_i))-w(\what{G}_S(Q_i))) \E[\delta_i \mid Q_i]) \right| \right]}{\epsilon} \\
    &~\quad~\quad~\le \frac{\E\left[\sqrt{\sum_{i \in I_1} (w_n(\what{G}_S(Q_i))-w(\what{G}_S(Q_i)))^2 \frac{1}{n}\sum_{i \in I_1}\E[\delta_i \mid Q_i]^2}\right]}{\epsilon} \\
    &~\quad~\quad~= \frac{\sqrt{\sum_{j=1}^n (w_n(j/n)-w(j/n))^2}}{\epsilon} \E\left[\sqrt{\frac{1}{n}\sum_{i \in I_1}\E[\delta_i \mid Q_i]^2} \right] \\
    &~\quad~\quad~\le \frac{\sqrt{\sum_{j=1}^n (w_n(j/n)-w(j/n))^2}}{\epsilon} \sqrt{\E\left[\frac{1}{n}\sum_{i \in I_1}\E[\delta_i \mid Q_i]^2\right]}.
\end{align*}
Assumption~\ref{assume:weight-fn} implies
\begin{equation*}
    \frac{\sum_{j=1}^n (w_n(j/n)-w(j/n))^2}{\epsilon^2}
\end{equation*}
is bounded eventually, and Assumption~\ref{assume:nuisance} implies
\begin{equation*}
    \E\left[\frac{1}{n}\sum_{i \in I_1}\E[\delta_i \mid Q_i]^2\right] \to 0.
\end{equation*}
Therefore, Term 2 is $o_P(1)$.

The third term is an $L$-statistic in the sense of \cite{Shorack72}, which differs from \cite{ShorackWe09} in that it allows the function of the $L$-statistic to vary with $n$, which is important when working with the approximation errors $\delta_i$. Specifically, they study estimators of the form
\begin{equation*}
    T_n = \frac{1}{n}\sum_{j=1}^n c_{nj} g_n(\xi_{nj})
\end{equation*}
where, for our purposes, $c_{ni}=w(j/n)$ and $\xi_{nj}$ is an order statistic of a uniform random variable. Because we have assumed that the distribution function $F_S(\cdot)$ of $Q$ is continuous, we can define $g_n(t) = \E[\delta_i \mid Q_i = F_{S}^{-1}(t)]$, and $g(t) = 0$. We proceed by showing that Assumptions~1-3 of \citet{Shorack72} hold (Assumption~4 is not needed, because we have assumed that what \citet{Shorack72} calls $\kappa$ is $0$). From this, we can conclude that $\sqrt{n}(T_n - \mu_n) \cp 0$, where $\mu_n = \E[w(1-F_S(Q_i))\E[\delta_i \mid Q_i]$. Notice that $\mu_n$ is also the mean of
\begin{equation*}
    T_n' = \frac{1}{n}\sum_{i=1}^n w(1-F_S(Q_i))\E[\delta_i \mid Q_i],
\end{equation*}
and that each of these terms is independent. Therefore, Chebyshev's inequality implies $\sqrt{n}(T_n' - \mu_n) \cp 0$, because
$|w(1-F_S(Q_i))\E[\delta_i \mid Q_i]|^2 \le \E[w^2(1-F_S(Q_i)) \delta_i^2 ] \to 0$ from Assumption~\ref{assume:nuisance}. Altogether, this implies that $\sqrt{n}(T_n - T_n') \cp 0$, which is exactly Term 3. Therefore, the proof follows from showing that Assumptions~1-3 of \citet{Shorack72} hold.

In this, to avoid conflicting notation, we shall use $\epsilon$ in place of the $\delta$ using in \citet{Shorack72}.
To satisfy Assumption~1 of \citet{Shorack72}, let $b_1 = b_2 = b$ as defined in Assumption~\ref{assume:weight-fn}, so that the conditions on $J$ hold. Using that $\E[|\delta_i|^r] < \infty$ in Assumption~\ref{assume:nuisance}, Markov's inequality gives
\begin{equation*}
    P( |\E[\delta_i \mid Q_i]| \ge s ) \le \frac{\E[|\E[\delta_i \mid Q_i]|^r]}{s^r} \le \frac{\E[|\delta_i|^r]}{s^r}.
\end{equation*}
Inverting this statement, using that $\E[\delta_i \mid Q_i]$ is of bounded variation and so can be split into two monotone functions, gives $|g_n(t)| \le C (t(1-t))^{1/r}$. Since the condition on $r$ and $b$ from Assumption~\ref{assume:regularity} implies there exists $\epsilon' > 0$ such that $1/r + b + \epsilon' \le 1/2$, Assumption~1 of \citet{Shorack72} holds with $\epsilon$ being the minimum of $\epsilon'$ and $\epsilon$ from Assumption~\ref{assume:nuisance}.

Here, $J_n = J$, so Assumption~2 of \citet{Shorack72} is satisfied.

Finally, to show Assumption~3, we re-write the integral using integration by parts and bound using H\"older's inequality as follows:
\begin{align*}
    &\int_{0}^1 M(t(1-t))^{-b+1/2-\epsilon/2} \dif{|\E[\delta_{i} \mid Q_i = F_S^{-1}(t)]|}
    \\
    &~\quad~\quad~= M(t(1-t))^{-b+1/2-\epsilon/2} |\E[\delta_{i} \mid Q_i = F_S^{-1}(t)| \bigg|_{0}^1 \\
    &~\quad~\quad~\quad~ - \int_{0}^1 M(t(1-t))^{-b-1/2-\epsilon/2} |\E[\delta_{i} \mid Q_i = F_S^{-1}(t)| \dif{t}
    \\
    &~\quad~\quad~= M(t(1-t))^{-b+1/2-\epsilon/2} |\E[\delta_{i} \mid Q_i = F_S^{-1}(t)| \bigg|_{0}^1 \\
    &~\quad~\quad~\quad~ - \int_{-\infty}^\infty M(F_S(q)(1-F_S(q)))^{-b-1/2-\epsilon/2} |\E[\delta_{i} \mid Q_i = q| \dif{F_S}(q),
\end{align*}
where the last line follows from the change of variables $t = F_S(q)$. Recalling that because of the constraints on $b$ in Assumption~\ref{assume:regularity}, $b<1/2$, and that Assumption~\ref{assume:nuisance} implies that \begin{equation*}
    \sup_{t} |\E[\delta_{i} \mid Q_i = F_S^{-1}(t)| < \infty,
\end{equation*}
the first term will be $0$. The second term converges to zero because
\begin{align*}
    &\int_{-\infty}^\infty M(F_S(q)(1-F_S(q)))^{-b-1/2-\epsilon/2} |\E[\delta_{i} \mid Q_i = q]| \dif{F_S}(q)\\
    &~\quad~\quad~\quad~= \E[M(F_S(Q_i)(1-F_S(Q_i)))^{-b-1/2-\epsilon/2} |\E[\delta_{i} \mid Q_i]|]
    \\&~\quad~\quad~\quad~\le \sqrt{\E[B^2(1-F_S(Q_i))(F_S(Q_i)(1-F_S(Q_i)))^{-2\epsilon} |\E[\delta_{i} \mid Q_i]|^2]\E[((1-F_S(Q_i))(F_S(Q_i)))^{-1 + \epsilon}]}
    \\&~\quad~\quad~\quad~\le \sqrt{C \E[B^2(1-F_S(Q_i))(F_S(Q_i)(1-F_S(Q_i)))^{-2\epsilon} \delta_i^2]} \to 0
\end{align*}
for a constant $\E[((1-F_S(Q_i))(F_S(Q_i)))^{-1 + \epsilon}] \le C < \infty$ and by Assumption~\ref{assume:nuisance} to conclude that the last line goes to $0$. Note that if $\epsilon' < \epsilon$ from Assumption~\ref{assume:nuisance}, the first term will involve an even larger exponent, and will still be bounded.
\end{proof}

\section{Additional Simulation Experiments}
We use several simulation experiments to probe aspects of RATE the metric's behavior, namely
the power of the RATE in testing against a null hypothesis of no heterogeneous treatment effects. We consider this both as a function of sample size and as a function of the signal-to-noise ratio of the heterogeneous treatment effect function.

We generate simulated data representing scenarios with both uncensored continuous and right-censored time-to-event/survival outcomes; employ several different kinds of prioritization rules, including CATE estimators like Causal Forests \citep{athey2019generalized}, S-Learners, and X-Learners \citep{kunzel2019metalearners}, as well as risk estimators like traditional random forests \citep{breiman2001random, hastie2009elements}; and consider both completely randomized trials as well as observational studies in which treatment assignment depends only on observable covariates (``unconfoundedness''). 
For the sake of illustration and brevity, we include in the main manuscript only one type continuous-outcome simulation and one type of survival-outcome simulation, each of which represents an observational study in the unconfoundedness setting. 


\subsection{Estimating RATEs in Observational Studies with Uncensored Outcomes and Unconfoundedness}
\label{sec:obs-study}

One common context for estimating the RATEs is in studies with binary treatments, $W_i \in \{0, 1\}$, continuous and uncensored treatment outcomes, $Y_i \in \mathbb{R}$, and where we allow for treatment selection based on observable covariates by assuming that potential outcomes are conditionally independent of the treatment conditioned on observed subject data, $Y_i(1), Y_i(0) \indep W_i | X_i$ (unconfoundedness).

\subsubsection{Simulator Design}
\label{sec:obs-study-sim}

We start from the simulation ``Setup A'' in \citet{nie2017quasi}, in which data is generated as follows:
\begin{align*}
    X_i &\sim \textnormal{Unif}(0,1)^d \\
    e(X_i) &= \textnormal{trim}_{0.1}\{\sin \left(\pi X_{i1} X_{i2}\right)\} \\
     W_i | X_i &\sim \textnormal{Bernoulli}(e(X_i)) \\
     b(X_i) &= \sin(\pi X_{i1} X_{i2}) + 2(X_{i3} - 0.5)^2 + X_{i4} + 0.5 X_{i5}
\end{align*}
where $\textnormal{trim}_{\eta}(x) = \max \{\eta, \min(x, 1 - \eta)\}$, $X_i$ represents the model input features for the $i^{th}$ subject, $e(\cdot)$ is the propensity function, $W_i$ is a binary treatment indicator ($1$ if the $i^{th}$ subject received treatment, $0$ otherwise), and $b(\cdot)$ is the baseline main effect modeled after the scaled \cite{friedman1991multivariate} function. The treatment effect function, $\tau(\cdot)$, and observed outcome for the $i^{th}$ subject, $Y_i$ are:
\begin{align*}
    \tau(X_i) &= \left(X_{i1} + X_{i2}\right)/2\\
    \varepsilon_i | X_i &\sim \mathcal{N}(0, 1) \\
    Y_i &= b(X_i) + (W_i - e(X_i)) \tau(X_i) + \sigma_{\varepsilon} \varepsilon_i
\end{align*} 
where $\sigma_{\varepsilon}$ represents the strength of random noise in the observation model. 

We modify this simulation slightly
by (1) introducing a parameter $\sigma_{\tau}$ that controls the signal strength/variance of the heterogeneous treatment effects, and (2) replacing $\tau(X_i)$ with 
\begin{equation*}
    \tilde{\tau}(X_i, \sigma_{\tau}) = \frac{\tau(X_i)}{SD(\tau)} \sigma_{\tau},
\end{equation*}
where $SD(\tau) = \sqrt{\frac{1}{n - 1}\sum_{i=1}^n \left(\tau(X_i) - \frac{1}{n}\sum_{i=1}^n\tau(X_i)\right)^2}.$

We note that in this simulation, the nuisance parameters $e(\cdot)$ and $b(\cdot)$ are difficult to learn, but the treatment effect function is more straightforward.

\subsubsection{Prioritization Rules}
\label{sec:obs-study-prioritization}
In the simulations, we study the RATE associated with prioritization rules generated as follows:

\textbf{Oracle} The Oracle prioritization rule assigns each subject a priority according to the true expected treatment effect conditioned on that subject's covariates. The Oracle thus has access to the true data generating process, though not the specific noise applied to each subject.

\textbf{Random} The Random prioritization rule assigns each subject a priority drawn uniformly at random from the interval [0, 1]. Note that a random prioritization rule will have a RATE of 0 for any weight function $\alpha(\cdot)$.

\textbf{Random Forest (RF) Risk} The RF Risk prioritization rule assigns each subject a priority in accordance with their estimated risk, where risk in this case is defined to be the subject's outcome in the absence of treatment: $\mu_0(x) = \EE{Y_i(0) | X_i  = x}$. We estimate the baseline risk using random forests, as implemented in \texttt{grf} \citep{athey2019generalized}. A key assumption in order for the RF Risk to achieve a high RATE is that risk is strongly correlated with the treatment effect. In our simulator setup, this is not the case.

\textbf{Causal Forest} The Causal Forest prioritization rule similarly uses ensembles of trees built with recursive partitioning to assign subject priorities but, unlike random forests, the causal forest chooses partitions 
to minimize the CATE R-loss criterion \citep{nie2017quasi}. We use the Causal Forest implementation provided by \texttt{grf}.

\textbf{X-learner} Like the Causal Forest prioritization rule, the X-learner \citep{kunzel2019metalearners} attempts to directly model the CATE as a function of patient covariates. However, the X-learner uses a different objective for targeting the CATE using the observed data.
The algorithm is described with more detail by \citet{kunzel2019metalearners}. We implement the X-learner with random forests (\texttt{grf}) as base model learners.

\subsubsection{Estimating Doubly Robust Scores of the Treatment Effect}

In this context, we estimate doubly robust scores for each participant using augmented inverse-propensity weighted (AIPW) scores, as described in equation \eqref{eqn:dr-score-unconfoundedness}. In these experiments, we use Random Forests \citep{breiman2001random, athey2019generalized, athey2019machine} as implemented in \texttt{grf} to fit the propensity score, $\hat{e}(x)$, and marginal response curve, $\hat{m}(x, w)$, on folds in a cross-fit manner \citep{ChernozhukovChDeDuHaNeRo18}.


\subsection{Estimating RATEs with Continuous, Right-Censored Outcomes and Unconfoundedness}
\label{sec:surv}


As discussed previously, a common context for validating prioritization rules with RATE metrics, especially in clinical studies with longitudinal outcomes, will be with time-to-event data with censoring. 


\subsubsection{Simulator Design}
\label{sec:surv-sim}
Adapting the ``Second Scenario'' in \citet{cui2023estimating}, we generate covariates independently from a uniform distribution on $[0, 1]^5$. The propensity function $e(X_i)$ is generated from a $\textnormal{Beta}(2, 4)$ distribution, $e(X_i) = (1 + \beta(X_{i2}; 2, 4)) / 4$. The failure time, $T_i$, is generated from a proportional hazard model and the censoring time, $C_i$, from an accelerated failure time model as follows:
\begin{align*}
    U_i &\sim \textnormal{Uniform(0, 1)} \\
    Z_i &\sim \mathcal{N}(0, 1) \\
    T_i &= \left(\frac{-\log(U_i)}{\exp\left(X_{i1} + (-0.4 + X_{i2}) \cdot W_i\right)}\right)^2 \\
    C_i &= \exp \left(X_{i1} - X_{i3} \cdot W_i + Z_i\right) \\
    Y_i &= \min(T_i, C_i) \\
    \Delta_i &= \mathbf{1}\{T_i \leq C_i\}
\end{align*}

Here $Y_i$ represents the observed study time and $\Delta_i$ represents an indicator for right-censoring, as described previously.

\subsubsection{Prioritization Rules}
\label{sec:surv-prioritization-rules}
While our \textbf{Oracle} and \textbf{Random} prioritization rules remain the same as discussed in Section~\ref{sec:obs-study-prioritization} for observational studies with unconfoundedness, we additionally analyzed three different prioritization rules adapted to the survival analysis setting.

\textbf{Random Survival Forest (RSF) Risk} Similar to the Random Forest (RF) Risk prioritization rule described previously, the Random Survival Forest (RSF) Risk prioritization rule estimates subjects' outcomes in the absence of treatment using a collection of single tree models. As in traditional random forest models, these trees are constructed via bootstrap aggregating (``bagging'') and random variable selection to decide which variables to split and how, given the data \citep{breiman2001random}. Here, however, the outcome considered is the conditional survival function, $S(t, x) = P(Y_i(0) > t | X_i = x)$. Additionally, nodes within trees are split by choosing the variable and split point which maximize survival difference between child nodes  \citep{ishwaran2008random}. We used survival forests as implemented in \texttt{grf} \citep{athey2019generalized} to learn RSF prioritization rules.

\textbf{Causal Survival Forest (CSF)} The Causal Survival Forest (CSF) prioritization rule adapts the Causal Forest algorithm described in Section~\ref{sec:obs-study-prioritization} to the setting of right-censored time-to-event outcomes. Concretely, given potential outcomes $T_i(1)$ and $T_i(0)$ representing survival times for the $i^{th}$ subject under treatment and control arms, respectively, the CSF prioritization rule learns to estimate the restricted mean survival time, $\tau(x) = \EE{ \min\{T_i(1), t_0\} - \min\{T_i(0), t_0\} \cond X=x}$, from the data. This is accomplished by constructing an ensemble of single tree models, each of which directly targets heterogeneity in the CATE when determining how and where to split nodes. We use an implementation of Causal Survival Forests in \texttt{grf} \citep{athey2019generalized} in our experiments.

\textbf{Cox Proportional Hazards S-learner} The Cox Proportional Hazards S-learner prioritization rule (CoxPH S-learner) adapts the canonical Cox Proportional Hazards survival model \citep{andersen1982cox} to the causal inference setting by (1) training a model to estimate the hazard ratio when treatment assignment is considered as simply an extra covariate in the model, $\hat{\lambda}(X_i = x, W_i = w)$;
then (2) estimating the CATE on the test set as the difference for each individual between the model's predictions if the subject were assigned to treatment and if they were assigned to control, i.e., $\hat{\tau}(x) = \hat{\lambda}(X_i = x, W_i = 1) - \hat{\lambda}(X_i = x, W_i = 0)$. In modeling $\hat{\lambda}$, we used an implementation of Cox Proportional Hazards provided within the \texttt{survival} package in R \citep{therneau2000cox}. We included all covariates, treatment assignment, and first-order interaction terms between treatment assignment and the other covariates as predictors. This is a common practice for estimating heterogeneous treatment effects in clinical studies \citep{rekkas2020predictive}.

\subsubsection{Estimating Doubly Robust Scores of the Treatment Effect}

For the time-to-event setting, we use a doubly robust score defined in equation~\eqref{eq:dr-score-censoring}. For the sake of brevity we omit the exact formula and refer readers to \citet{cui2023estimating} for details.  Mote that in our experiments we estimate nuisance parameters, including the expected remaining survival time for each individual and the conditional survival function, using survival forests \citep{ishwaran2008random} with cross-fitting, as implemented in the \texttt{grf} R package. Propensity score estimators are learned using random forests with cross-fitting, as also provided in the \texttt{grf} R package.

\subsection{Analysis of Statistical Power}
\label{sec:statistical-power}
The statistical power of the hypothesis test described in Section~\ref{sec:hypo-test} depends on the strength of the heterogeneous treatment effects, the sample size, the choice of weighting function, and the choice of score. In the following sub-sections, we illustrate the relationship between statistical power and strength of heterogeneous treatment effects and the relationship between power and sample size, respectively. Then, we consider the effect that the weighting function of the RATE has on the power, for different distributions of heterogeneous treatment effects.

\subsubsection{Power as a function of heterogeneous treatment effect strength}
Using synthetic data from the simulator described in Section~\ref{sec:obs-study-sim}, we learned nuisance parameter estimators and the prioritization rules described in Section~\ref{sec:obs-study-prioritization} on 1000 training samples and estimated the RATE on another 1000 i.i.d. test samples. For each prioritization rule, we calculated a $P$-value for $H_0: \text{RATE} = 0$ using the procedure described in Section~\ref{sec:hypo-test}. We repeated this procedure 1000 times, and counted the fraction of simulations in which $H_0$ was rejected at significance level $\alpha = 0.05$ as a measure statistical power. This process of estimating power was repeated for a variety of heterogeneous treatment effect strengths $\sigma_\tau$. When there was no heterogeneity in the treatment effects ($\sigma_\tau = 0$), the proportion of null hypotheses rejected was interpreted to be an estimate of the Type I error of our hypothesis test.

Results of this experiment can be seen in Figure \ref{fig:power_vs_snr}. For prioritization rules that directly target the CATE (e.g., Causal Forests, X-learners), statistical power increased monotonically with heterogeneous treatment effect strength, as desired. We additionally observe that these two CATE-targeting estimators perform substantially better than a risk-based prioritization rule using random forests. This is to be expected given the simulator design.

\begin{figure}[ht]
    \begin{subfigure}{\textwidth}
        \includegraphics[scale=0.8]{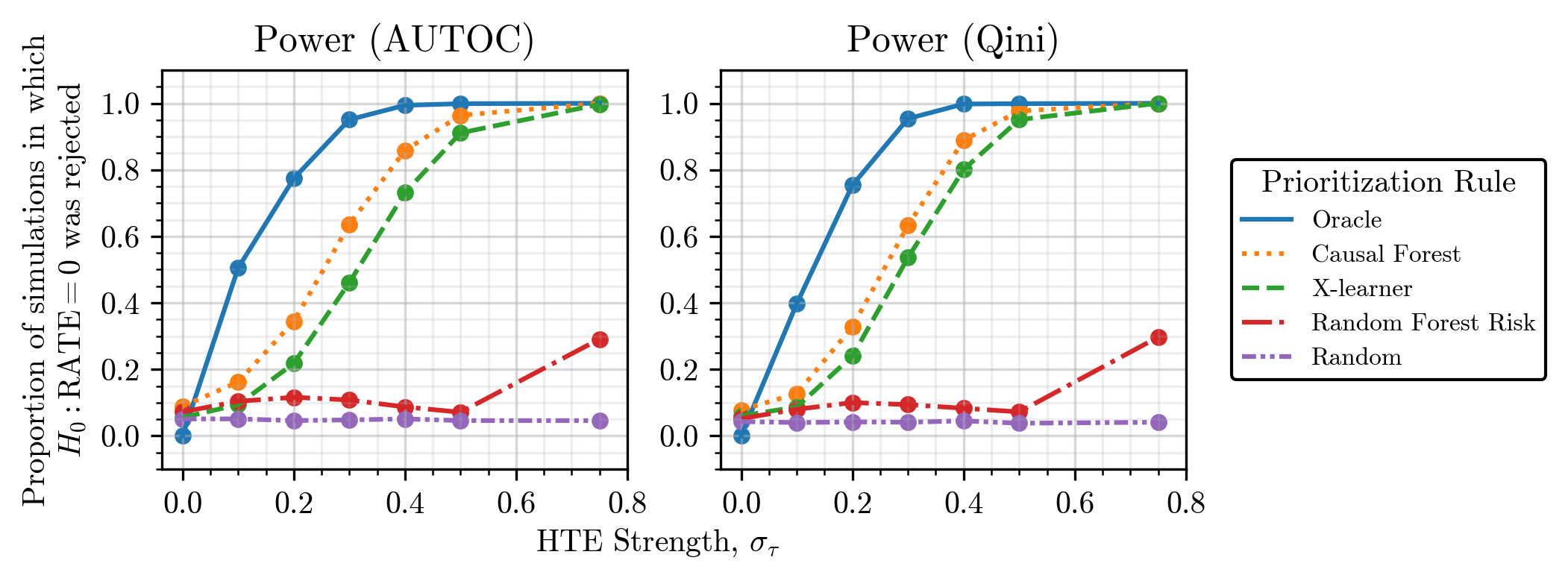} 
    \end{subfigure}
    \caption{Power as a function of heterogeneous treatment effect strength, $\sigma_{\tau}$, in Setup 1 for the RATE with both logarithmic (AUTOC) and linear (Qini) weighting. Of note, even when $\sigma_{\tau} = 0.3$ i.e. the variance/signal strength of the heterogeneous treatment effects was $\sigma_{\tau}^2/\sigma_{\varepsilon}^2 = 0.3^2 = 0.09 \approx 10\%$ that of the noise in the observed outcome model, we were able to use the procedure outlined in Section~\ref{sec:hypo-test} to test against the presence of heterogeneous treatment effects with a power of $> 0.6$ and a Type I error of approximately $0.05$. Once the signal-to-noise ratio approached closer to $\sigma_{\tau}^2/\sigma_{\varepsilon}^2 = 0.4^2 = 0.16$, we were able to test against the presence of heterogeneous treatment effects with a power of $>0.8$ and a Type I error of $0.05$ using the Causal Forest implemented in $\texttt{grf}$ with default hyperparameters.}
    \label{fig:power_vs_snr}
\end{figure}

\subsubsection{Power as a function of sample size}

Using the simulator described in Section~\ref{sec:surv-sim}, we estimated statistical power via the procedure outlined in the previous section. Here, however, instead of measuring power as a function of varying heterogeneous treatment effect strength, we analyzed how power changed with increasing sample size.

See Figure \ref{fig:power_vs_sample_size} for a summary of the results. Of note, the RATE increased monotonically with the sample size for prioritization rules that directly target the CATE (e.g., the Causal Survival Forest and Cox Proportional Hazards S-learner models). The same is not true of prioritization rules that only targeted risk/the baseline outcome (e.g., the Random Survival Forest Risk model).

\begin{figure}[ht]
    \begin{subfigure}{\textwidth}
        \includegraphics[scale=0.8]{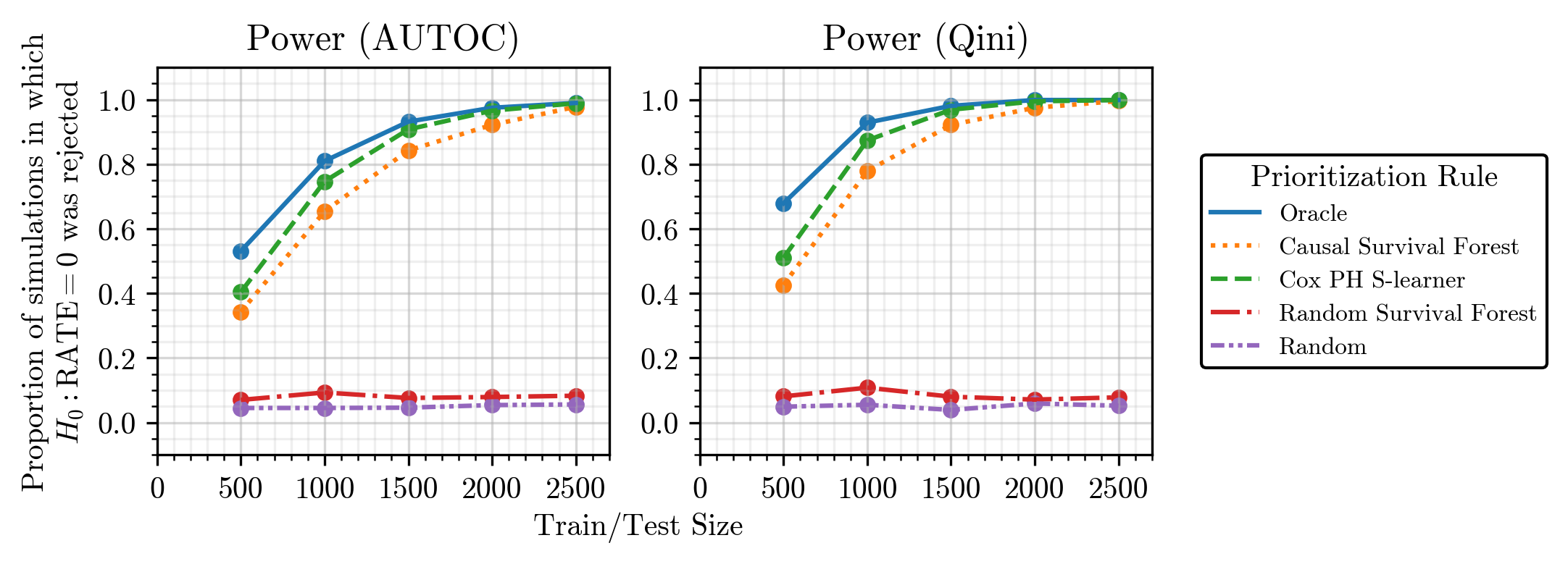} 
    \end{subfigure}
    \caption{Statistical power in testing against the absence of heterogeneous treatment effects, as a function of sample size (train and test set size are identical within each simulation), in the context of right-censored, time-to-event outcomes under unconfoundedness.
    }
    \label{fig:power_vs_sample_size}
\end{figure}

\subsubsection{Power vs. Weighting Function}

 In our simulation section, ``Section 4: Choosing a RATE Metric and Score'', we considered two aspects of the RATE estimator that impact its statistical power for detecting treatment effect heterogeneity: (1) the weighting function $\alpha(\cdot)$ (e.g., Qini vs. AUTOC), and (2) the form for the score, $\widehat{\Gamma}_i$ (e.g., AIPW vs. IPW scores). 

Regarding (1), the weighting function, we showed in simulation that the AUTOC yields greater statistical power than the Qini coefficient when only a small subset of the population experiences nontrivial heterogeneous treatment effects whereas the Qini coefficient yields greater power when heterogeneous treatment effects are diffuse and substantial across the population (see Figure 2 in the main manuscript). This simulation was performed using AIPW oracle scores, $\Gamma^*_i$, as given in Equation 2.13. For these Oracle scores, we assume nuisance parameters are known a priori, including $m(x,w)$ which represents the expected outcome given a subject's covariates, $x$, and treatment assignment, $w \in \{0, 1\}$. In practice, however, one never has access to $m(x, w)$ and this must be estimated instead. We assess the robustness/sensitivity of our findings by repeating the same experiment using Oracle scores, AIPW scores with \emph{estimated} $\hat{m}(x, w)$, and IPW scores. The results are shown in Figure \ref{fig:qini_vs_autoc} below. Notably, the trends highlighted in the manuscript regarding the impact of weighting function on statistical power are robust to choice of score.

\begin{figure}[ht]
    \centering
    \begin{subfigure}{0.32\textwidth}
        \includegraphics[scale=0.6]{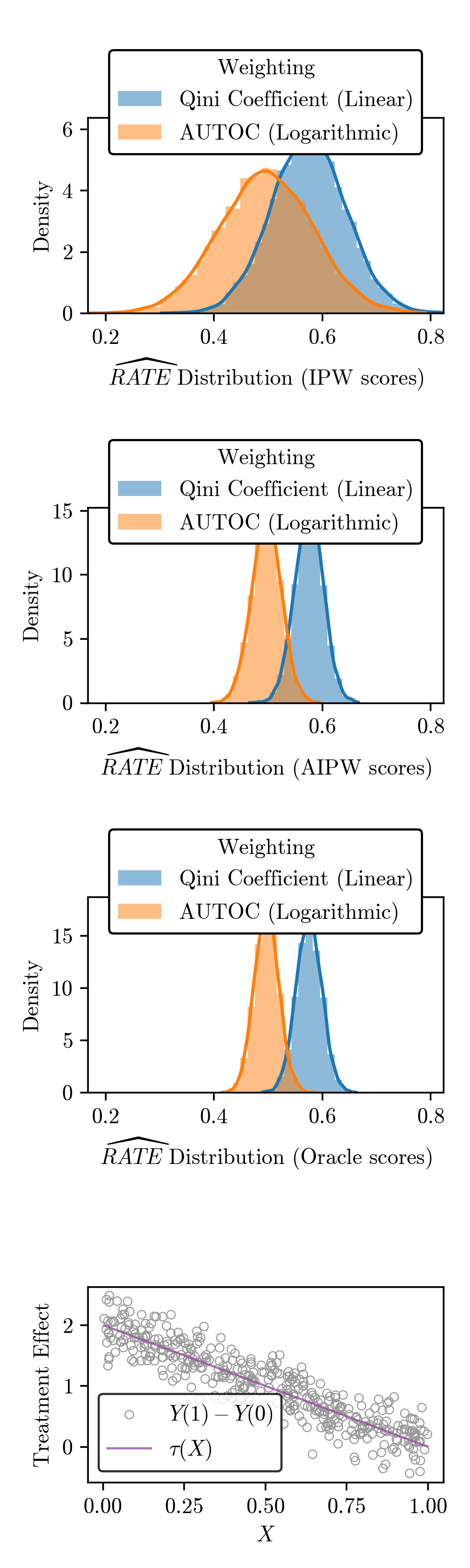} 
        \caption{100\% have $|\tau(X_i)| > 0$}
    \end{subfigure}
    \begin{subfigure}{0.32\textwidth}
        \includegraphics[scale=0.6]{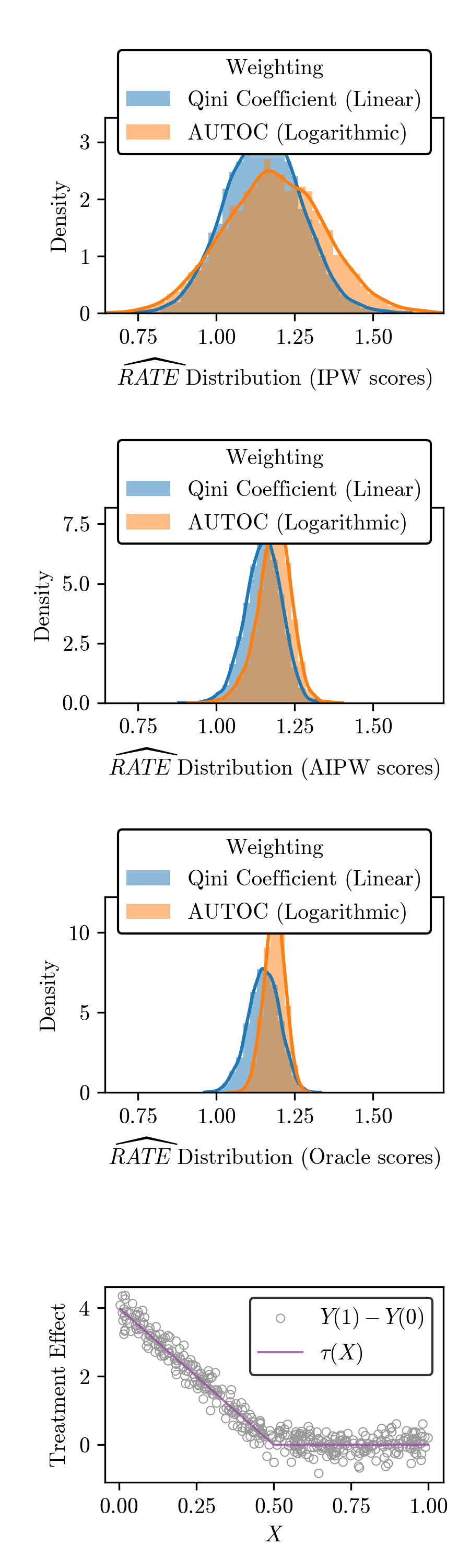}
        \caption{50\% have $|\tau(X_i)| > 0$}
    \end{subfigure}
    \begin{subfigure}{0.32\textwidth}
        \includegraphics[scale=0.6]{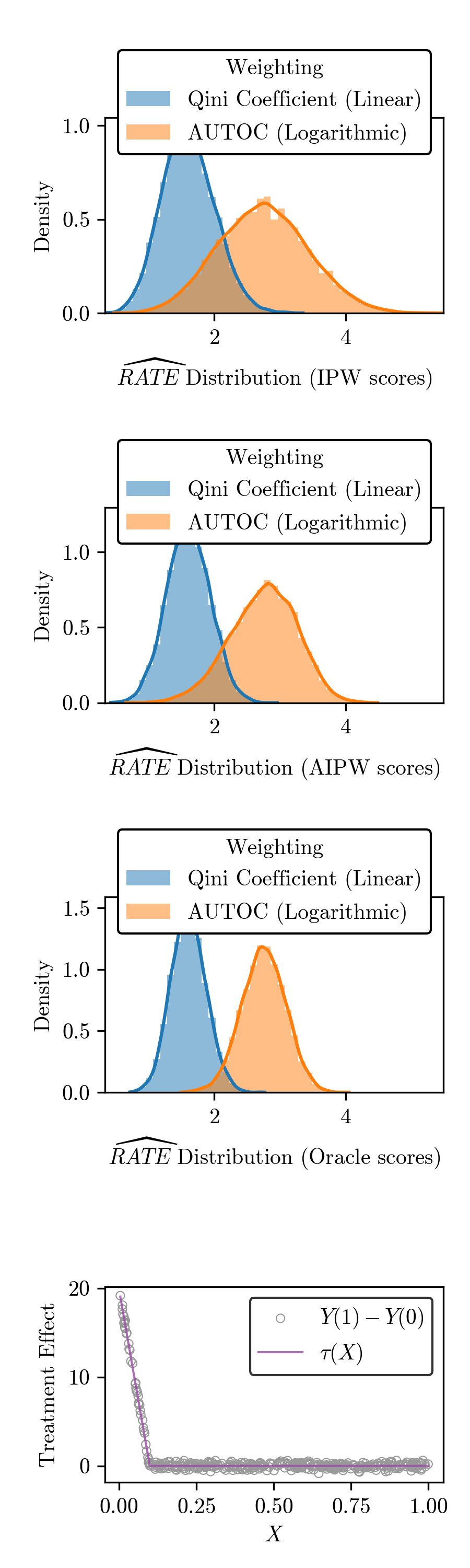}
        \caption{10\% have $|\tau(X_i)| > 0$}
    \end{subfigure}
    \caption{Comparison of linear (Qini) vs. logarithmic (AUTOC) weighting functions for different scores (IPW, AIPW, and Oracle scores). If treatment effects are nonzero for a significant proportion of the population (left column) the power of the estimated RATE when using linear weighting (Qini) tends to be greater than when using logarithmic (AUTOC) weighting. Conversely, if nonzero treatment effects are concentrated among a small proportion of the population (right column) we see that using logarithmic (AUTOC) weighting leads to a greater power for the estimated RATE relative to using linear (Qini) weighting.}
    \label{fig:qini_vs_autoc}
\end{figure}

\begin{table}[p]
\centering
\begin{tabular}{lcccc}
\toprule
 & \multicolumn{2}{c}{ACCORD-BP} & \multicolumn{2}{c}{SPRINT} \\
\cmidrule(lr){2-3} \cmidrule(lr){4-5}
 & n & \% & n & \% \\
\midrule
Total & 4684 &  & 9160 &  \\
\midrule
Non-Intensive & 2342 & 50.0 & 4567 & 49.9 \\
Intensive & 2342 & 50.0 & 4593 & 50.1 \\
\midrule
No Diabetes & 0 & 0.0 & 9160 & 100.0 \\
Diabetes & 4684 & 100.0 & 0 & 0.0 \\
\midrule
Male & 2455 & 52.4 & 5908 & 64.5 \\
Female & 2229 & 47.6 & 3252 & 35.5 \\
\midrule
Non-Black & 3578 & 76.4 & 6414 & 70.0 \\
Black & 1106 & 23.6 & 2746 & 30.0 \\
\midrule
Non-Hispanic & 4359 & 93.1 & 8190 & 89.4 \\
Hispanic & 325 & 6.9 & 970 & 10.6 \\
\midrule
Not current smoker & 4061 & 86.7 & 7932 & 86.6 \\
Current smoker & 623 & 13.3 & 1228 & 13.4 \\
\midrule
Not former smoked & 2719 & 58.0 & 5265 & 57.5 \\
Former smoker & 1965 & 42.0 & 3895 & 42.5 \\
\midrule
Not on Aspirin & 2225 & 47.5 & 4493 & 49.1 \\
On Aspirin & 2459 & 52.5 & 4667 & 50.9 \\
\midrule
Not on Statin & 1632 & 34.8 & 5165 & 56.4 \\
On Statin & 3052 & 65.2 & 3995 & 43.6 \\
\midrule
No Angina & 4153 & 88.7 & 7971 & 87.0 \\
Angina & 531 & 11.3 & 1189 & 13.0 \\
\midrule
 & Mean & SD & Mean & SD \\
\midrule
SBP & 139.15 & 15.78 & 139.66 & 15.59 \\
DBP & 75.95 & 10.39 & 78.15 & 11.94 \\
eGFR & 91.57 & 28.78 & 71.80 & 20.60 \\
\midrule
 & Median & [IQR] & Median & [IQR] \\
\midrule
Age & 62.05 & [57.60, 67.10] & 67.00 & [61.00, 76.00] \\
BP medications & 2.00 & [1.00, 2.00] & 2.00 & [1.00, 3.00] \\
Creatinine & 0.90 & [0.70, 1.00] & 1.01 & [0.86, 1.21] \\
Cholesterol & 188.00 & [162.00, 216.25] & 187.00 & [161.00, 215.00] \\
HDL cholesterol & 45.00 & [37.00, 54.00] & 50.00 & [43.00, 60.00] \\
Triglycerides & 147.00 & [98.00, 226.25] & 107.00 & [77.00, 150.00] \\
BMI & 31.63 & [28.05, 35.90] & 29.02 & [25.88, 32.90] \\
\bottomrule
\end{tabular}
\caption{Baseline characteristics for the SPRINT and ACCORD-BP trials}
\label{tab:baseline-characteristics}
\end{table}

\section{Combined analysis of SPRINT/ACCORD-BP}
\label{sec:combined-sprint-accord-RMST}

While in the main manuscript we train on SPRINT and evaluate on ACCORD-BP (or vice versa) to highlight how the RATE would likely be used in practice (train on one trial/population, evaluate how well prioritization rules generalize to new populations), we also recognize the value of in-distribution evaluation. In this experiment we follow the exact same protocol as described in Section \ref{sec:CVD}, except that we use as our ``train'' set for the Causal Survival Forest, Random Survival Forest, and Cox PH S-Learner models a combined dataset generated by combining subjects from the SPRINT and ACCORD-BP trials and selecting from that set 6922 subjects (uniformly at random). The ``test'' is constructed in similar fashion.

\begin{table}[htbp]
\centering
\begin{tabular}{|r|cc|}
	\hline
	Prioritization Rule & AUTOC (95\% CI) & $p$-value \\
	\hline
	Causal Survival Forest (grf) & -5.83 (-14.09, 2.43) & 0.17 \\
        Cox PH S-learner & -0.99 (-9.49, 7.52) & 0.82 \\
	\hline
	Random Survival Forest Risk (grf) & 3.03 (-6.41, 12.46) & 0.53 \\
	Framingham Risk Score & -7.01 (-18.89, 4.88) & 0.25 \\
	ACC/AHA Pooled Cohort Equations & -1.43 (-13.22, 10.36) & 0.81 \\
	\hline
\end{tabular}
\caption{AUTOC estimates (RMST) obtained using data from a test split of a combination of both SPRINT and ACCORD-BP subjects ($n = 6922$), with prioritization rules trained on SPRINT ($n = 6922$), if necessary. We also show 95\% confidence intervals obtained using the half-sample bootstrap,  along with associated $p$-values.}
\label{tab:train_sprint_test_accord_rmst}
\end{table}

\setlength{\bibsep}{0.2pt plus 0.3ex}

\def\spacingset#1{\renewcommand{\baselinestretch}%
{#1}\small\normalsize} \spacingset{1}
\bibliographystyle{plainnat-abbrev}
{\footnotesize
\bibliography{bib}

\begin{thebibliography}{75}
\providecommand{\natexlab}[1]{#1}
\providecommand{\url}[1]{\texttt{#1}}
\expandafter\ifx\csname urlstyle\endcsname\relax
  \providecommand{\doi}[1]{doi: #1}\else
  \providecommand{\doi}{doi: \begingroup \urlstyle{rm}\Url}\fi

\bibitem[Andersen and Gill(1982)]{andersen1982cox}
P.~K. Andersen and R.~D. Gill.
\newblock Cox's regression model for counting processes: a large sample study.
\newblock \emph{The Annals of Statistics}, pages 1100--1120, 1982.

\bibitem[Artigao-Rodenas et~al.(2013)Artigao-Rodenas, Carbayo-Herencia,
  Divison-Garrote, Gil-Guillen, Mass{\'o}-Orozco, Simarro-Rueda,
  Molina-Escribano, Sanchis, Carri{\'o}n-Valero, Lopez~de Coca,
  et~al.]{artigao2013framingham}
L.~M. Artigao-Rodenas, J.~A. Carbayo-Herencia, J.~A. Divison-Garrote, V.~F.
  Gil-Guillen, J.~Mass{\'o}-Orozco, M.~Simarro-Rueda, F.~Molina-Escribano,
  C.~Sanchis, L.~Carri{\'o}n-Valero, E.~Lopez~de Coca, et~al.
\newblock Framingham risk score for prediction of cardiovascular diseases: a
  population-based study from southern Europe.
\newblock \emph{PloS One}, 8\penalty0 (9):\penalty0 e73529, 2013.

\bibitem[Ascarza(2018)]{Ascarza18}
E.~Ascarza.
\newblock Retention futility: Targeting high-risk customers might be
  ineffective.
\newblock \emph{Journal of Marketing Research}, 55\penalty0 (1):\penalty0
  80--98, 2018.

\bibitem[Athey and Imbens(2019)]{athey2019machine}
S.~Athey and G.~W. Imbens.
\newblock Machine learning methods that economists should know about.
\newblock \emph{Annual Review of Economics}, 11:\penalty0 685--725, 2019.

\bibitem[Athey and Wager(2021)]{athey2021policy}
S.~Athey and S.~Wager.
\newblock Policy learning with observational data.
\newblock \emph{Econometrica}, 89\penalty0 (1):\penalty0 133--161, 2021.

\bibitem[Athey et~al.(2019)Athey, Tibshirani, and Wager]{athey2019generalized}
S.~Athey, J.~Tibshirani, and S.~Wager.
\newblock Generalized Random Forests.
\newblock \emph{Annals of Statistics}, 47\penalty0 (2):\penalty0 1148--1178,
  2019.

\bibitem[Basu et~al.(2017)Basu, Sussman, and Hayward]{basu2017detecting}
S.~Basu, J.~B. Sussman, and R.~A. Hayward.
\newblock Detecting heterogeneous treatment effects to guide personalized blood
  pressure treatment: a modeling study of randomized clinical trials.
\newblock \emph{Annals of Internal Medicine}, 166\penalty0 (5):\penalty0
  354--360, 2017.

\bibitem[Beddhu et~al.(2018)Beddhu, Chertow, Greene, Whelton, Ambrosius,
  Cheung, Cutler, Fine, Boucher, Wei, et~al.]{beddhu2018effects}
S.~Beddhu, G.~M. Chertow, T.~Greene, P.~K. Whelton, W.~T. Ambrosius, A.~K.
  Cheung, J.~Cutler, L.~Fine, R.~Boucher, G.~Wei, et~al.
\newblock Effects of intensive systolic blood pressure lowering on
  cardiovascular events and mortality in patients with type 2 diabetes mellitus
  on standard glycemic control and in those without diabetes mellitus:
  reconciling results from ACCORD BP and SPRINT.
\newblock \emph{Journal of the American Heart Association}, 7\penalty0
  (18):\penalty0 e009326, 2018.

\bibitem[Breiman(2001)]{breiman2001random}
L.~Breiman.
\newblock Random forests.
\newblock \emph{Machine learning}, 45\penalty0 (1):\penalty0 5--32, 2001.

\bibitem[Chen et~al.(1997)]{chen1997cast}
Z.-M. Chen et~al.
\newblock CAST: randomised placebo-controlled trial of early aspirin use in 20
  000 patients with acute ischaemic stroke.
\newblock \emph{The Lancet}, 349\penalty0 (9066):\penalty0 1641--1649, 1997.

\bibitem[Chen et~al.(2000)Chen, Sandercock, Pan, Counsell, Collins, Liu, Xie,
  Warlow, and Peto]{chen2000indications}
Z.~Chen, P.~Sandercock, H.~Pan, C.~Counsell, R.~Collins, L.~Liu, J.~Xie,
  C.~Warlow, and R.~Peto.
\newblock Indications for early aspirin use in acute ischemic stroke: a
  combined analysis of 40 000 randomized patients from the Chinese Acute Stroke
  Trial and the International Stroke Trial.
\newblock \emph{Stroke}, 31\penalty0 (6):\penalty0 1240--1249, 2000.

\bibitem[Chernozhukov et~al.(2018{\natexlab{a}})Chernozhukov, Chetverikov,
  Demirer, Duflo, Hansen, Newey, and Robins]{ChernozhukovChDeDuHaNeRo18}
V.~Chernozhukov, D.~Chetverikov, M.~Demirer, E.~Duflo, C.~Hansen, W.~Newey, and
  J.~Robins.
\newblock Double/debiased machine learning for treatment and structural
  parameters.
\newblock \emph{The Econometrics Journal}, 21\penalty0 (1):\penalty0 C1--C68,
  2018{\natexlab{a}}.

\bibitem[Chernozhukov et~al.(2018{\natexlab{b}})Chernozhukov, Demirer, Duflo,
  and Fernandez-Val]{chernozhukov2018generic}
V.~Chernozhukov, M.~Demirer, E.~Duflo, and I.~Fernandez-Val.
\newblock Generic machine learning inference on heterogenous treatment effects
  in randomized experiments.
\newblock Technical report, National Bureau of Economic Research,
  2018{\natexlab{b}}.

\bibitem[Chernozhukov et~al.(2022)Chernozhukov, Escanciano, Ichimura, Newey,
  and Robins]{chernozhukov2016locally}
V.~Chernozhukov, J.~C. Escanciano, H.~Ichimura, W.~K. Newey, and J.~M. Robins.
\newblock Locally robust semiparametric estimation.
\newblock \emph{Econometrica}, 90\penalty0 (4):\penalty0 1501--1535, 2022.

\bibitem[Chung and Romano(2013)]{chung2013exact}
E.~Chung and J.~P. Romano.
\newblock Exact and asymptotically robust permutation tests.
\newblock \emph{The Annals of Statistics}, 41\penalty0 (2):\penalty0 484--507,
  2013.

\bibitem[Cui et~al.(2023)Cui, Kosorok, Sverdrup, Wager, and
  Zhu]{cui2023estimating}
Y.~Cui, M.~R. Kosorok, E.~Sverdrup, S.~Wager, and R.~Zhu.
\newblock Estimating heterogeneous treatment effects with right-censored data
  via causal survival forests.
\newblock \emph{Journal of the Royal Statistical Society Series B: Statistical
  Methodology}, 85\penalty0 (2):\penalty0 179--211, 2023.

\bibitem[D'Agostino et~al.(2008)D'Agostino, Vasan, Pencina, Wolf, Cobain,
  Massaro, and Kannel]{DAgostinoVaPeWoCoMaKa08}
R.~B. D'Agostino, R.~S. Vasan, M.~J. Pencina, P.~A. Wolf, M.~Cobain, J.~M.
  Massaro, and W.~B. Kannel.
\newblock General Cardiovascular Risk Profile for Use in Primary Care.
\newblock \emph{Circulation}, 117\penalty0 (6):\penalty0 743--753, 2008.

\bibitem[DeFilippis et~al.(2015)DeFilippis, Young, Carrubba, McEvoy, Budoff,
  Blumenthal, Kronmal, McClelland, Nasir, and Blaha]{defilippis2015analysis}
A.~P. DeFilippis, R.~Young, C.~J. Carrubba, J.~W. McEvoy, M.~J. Budoff, R.~S.
  Blumenthal, R.~A. Kronmal, R.~L. McClelland, K.~Nasir, and M.~J. Blaha.
\newblock An analysis of calibration and discrimination among multiple
  cardiovascular risk scores in a modern multiethnic cohort.
\newblock \emph{Annals of Internal Medicine}, 162\penalty0 (4):\penalty0
  266--275, 2015.

\bibitem[Diemert et~al.(2018)Diemert, Betlei, Renaudin, and
  Amini]{DiemertBeReAm18}
E.~Diemert, A.~Betlei, C.~Renaudin, and M.-R. Amini.
\newblock A Large Scale Benchmark for Uplift Modeling.
\newblock In \emph{KDD}, 2018.

\bibitem[Ding et~al.(2016)Ding, Feller, and Miratrix]{DingFeMi16}
P.~Ding, A.~Feller, and L.~Miratrix.
\newblock Randomization inference for treatment effect variation.
\newblock \emph{Journal of the Royal Statistical Society: Series B: Statistical
  Methodology}, pages 655--671, 2016.

\bibitem[Dvoretzky et~al.(1956)Dvoretzky, Kiefer, and
  Wolfowitz]{DvoretzkyKiWo56}
A.~Dvoretzky, J.~Kiefer, and J.~Wolfowitz.
\newblock {Asymptotic Minimax Character of the Sample Distribution Function and
  of the Classical Multinomial Estimator}.
\newblock \emph{The Annals of Mathematical Statistics}, 27\penalty0
  (3):\penalty0 642 -- 669, 1956.
\newblock \doi{10.1214/aoms/1177728174}.
\newblock URL \url{https://doi.org/10.1214/aoms/1177728174}.

\bibitem[Efron(1982)]{Efron82}
B.~Efron.
\newblock \emph{The jackknife, the bootstrap and other resampling plans}.
\newblock SIAM, 1982.

\bibitem[Feigin et~al.(2021)Feigin, Stark, Johnson, Roth, Bisignano, Abady,
  Abbasifard, Abbasi-Kangevari, Abd-Allah, Abedi, et~al.]{feigin2021global}
V.~L. Feigin, B.~A. Stark, C.~O. Johnson, G.~A. Roth, C.~Bisignano, G.~G.
  Abady, M.~Abbasifard, M.~Abbasi-Kangevari, F.~Abd-Allah, V.~Abedi, et~al.
\newblock Global, regional, and national burden of stroke and its risk factors,
  1990--2019: A systematic analysis for the Global Burden of Disease Study
  2019.
\newblock \emph{The Lancet Neurology}, 20\penalty0 (10):\penalty0 795--820,
  2021.

\bibitem[Fisher and Curfman(2018)]{fisher2018hypertension}
N.~D. Fisher and G.~Curfman.
\newblock Hypertension—a public health challenge of global proportions.
\newblock \emph{Journal of the American Medical Association}, 320\penalty0
  (17):\penalty0 1757--1759, 2018.

\bibitem[Friedman(1991)]{friedman1991multivariate}
J.~H. Friedman.
\newblock Multivariate Adaptive Regression Splines.
\newblock \emph{The Annals of Statistics}, pages 1--67, 1991.

\bibitem[Goff et~al.(2014)Goff, Lloyd-Jones, Bennett, Coady, D’Agostino,
  Gibbons, Greenland, Lackland, Levy, O’Donnell, Robinson, Schwartz, Shero,
  Smith, Sorlie, Stone, and Wilson]{goff2014accaha}
D.~C. Goff, D.~M. Lloyd-Jones, G.~Bennett, S.~Coady, R.~B. D’Agostino,
  R.~Gibbons, P.~Greenland, D.~T. Lackland, D.~Levy, C.~J. O’Donnell, J.~G.
  Robinson, J.~S. Schwartz, S.~T. Shero, S.~C. Smith, P.~Sorlie, N.~J. Stone,
  and P.~W.~F. Wilson.
\newblock 2013 ACC/AHA Guideline on the Assessment of Cardiovascular Risk.
\newblock \emph{Circulation}, 129\penalty0 (25\_suppl\_2):\penalty0 S49--S73,
  2014.
\newblock \doi{10.1161/01.cir.0000437741.48606.98}.
\newblock URL
  \url{https://www.ahajournals.org/doi/abs/10.1161/01.cir.0000437741.48606.98}.

\bibitem[Gordon et~al.(2019)Gordon, Zettelmeyer, Bhargava, and
  Chapsky]{GordonZeBhCh19}
B.~R. Gordon, F.~Zettelmeyer, N.~Bhargava, and D.~Chapsky.
\newblock A comparison of approaches to advertising measurement: Evidence from
  big field experiments at Facebook.
\newblock \emph{Marketing Science}, 38\penalty0 (2):\penalty0 193--225, 2019.

\bibitem[Green et~al.(1966)Green, Swets, et~al.]{green1966signal}
D.~M. Green, J.~A. Swets, et~al.
\newblock \emph{Signal Detection Theory and Psychophysics}, volume~1.
\newblock Wiley New York, 1966.

\bibitem[Group(2010)]{accord2010effects}
A.~S. Group.
\newblock Effects of intensive blood-pressure control in type 2 diabetes
  mellitus.
\newblock \emph{New England Journal of Medicine}, 362\penalty0 (17):\penalty0
  1575--1585, 2010.

\bibitem[Group(1997)]{group1997international}
I.~S. T.~C. Group.
\newblock The International Stroke Trial (IST): a randomised trial of aspirin,
  subcutaneous heparin, both, or neither among 19 435 patients with acute
  ischaemic stroke.
\newblock \emph{The Lancet}, 349\penalty0 (9065):\penalty0 1569--1581, 1997.

\bibitem[Group(2015)]{sprint2015randomized}
S.~R. Group.
\newblock A randomized trial of intensive versus standard blood-pressure
  control.
\newblock \emph{New England Journal of Medicine}, 373\penalty0 (22):\penalty0
  2103--2116, 2015.

\bibitem[Hastie et~al.(2009)Hastie, Tibshirani, and
  Friedman]{hastie2009elements}
T.~Hastie, R.~Tibshirani, and J.~Friedman.
\newblock \emph{The Elements of Statistical Learning: Data Mining, Inference,
  and Prediction}.
\newblock Springer Science \& Business Media, 2009.

\bibitem[Hill(2011)]{hill2011bayesian}
J.~L. Hill.
\newblock Bayesian nonparametric modeling for causal inference.
\newblock \emph{Journal of Computational and Graphical Statistics}, 20\penalty0
  (1):\penalty0 217--240, 2011.

\bibitem[Hohnhold et~al.(2015)Hohnhold, O'Brien, and Tang]{HenningObTa15}
H.~Hohnhold, D.~O'Brien, and D.~Tang.
\newblock Focus on the Long-Term: It's better for Users and Business.
\newblock In \emph{Proceedings 21st Conference on Knowledge Discovery and Data
  Mining}, Sydney, Australia, 2015.
\newblock URL \url{http://dl.acm.org/citation.cfm?doid=2783258.2788583}.

\bibitem[Imai and Li(2023)]{ImaiL19}
K.~Imai and M.~L. Li.
\newblock Experimental evaluation of individualized treatment rules.
\newblock \emph{Journal of the American Statistical Association}, 118\penalty0
  (541):\penalty0 242--256, 2023.

\bibitem[Imbens and Rubin(2015)]{imbens2015causal}
G.~W. Imbens and D.~B. Rubin.
\newblock \emph{Causal inference in statistics, social, and biomedical
  sciences}.
\newblock Cambridge University Press, 2015.

\bibitem[Ishwaran et~al.(2008)Ishwaran, Kogalur, Blackstone, Lauer,
  et~al.]{ishwaran2008random}
H.~Ishwaran, U.~B. Kogalur, E.~H. Blackstone, M.~S. Lauer, et~al.
\newblock Random survival forests.
\newblock \emph{Annals of Applied Statistics}, 2\penalty0 (3):\penalty0
  841--860, 2008.

\bibitem[Johnson et~al.(2017)Johnson, Lewis, and Nubbemeyer]{JohnsonLeNu17}
G.~A. Johnson, R.~A. Lewis, and E.~I. Nubbemeyer.
\newblock Ghost Ads: Improving the Economics of Measuring Online Ad
  Effectiveness.
\newblock \emph{Journal of Marketing Research}, 54\penalty0 (6):\penalty0
  867--884, 2017.

\bibitem[Kaul(2017)]{kaul2017tale}
S.~Kaul.
\newblock A tale of two trials: reconciling differences in results by exploring
  heterogeneous treatment effects.
\newblock \emph{Annals of Internal Medicine}, 166\penalty0 (5):\penalty0
  370--372, 2017.

\bibitem[Kennedy(2023)]{kennedy2023towards}
E.~H. Kennedy.
\newblock Towards optimal doubly robust estimation of heterogeneous causal
  effects.
\newblock \emph{Electronic Journal of Statistics}, 17\penalty0 (2):\penalty0
  3008--3049, 2023.

\bibitem[Kent et~al.(2016)Kent, Nelson, Dahabreh, Rothwell, Altman, and
  Hayward]{Kent2016risk}
D.~M. Kent, J.~Nelson, I.~J. Dahabreh, P.~M. Rothwell, D.~G. Altman, and R.~A.
  Hayward.
\newblock Risk and treatment effect heterogeneity: re-analysis of individual
  participant data from 32 large clinical trials.
\newblock \emph{International journal of epidemiology}, 45\penalty0
  (6):\penalty0 2075--2088, 2016.

\bibitem[Kent et~al.(2020)Kent, van Klaveren, Paulus, D'Agostino, Goodman,
  Hayward, Ioannidis, Patrick-Lake, Morton, Pencina, Raman, Ross, Selker,
  Varadhan, Vickers, Wong, and Steyerberg]{KentEtAl20}
D.~M. Kent, D.~van Klaveren, J.~K. Paulus, R.~D'Agostino, S.~Goodman,
  R.~Hayward, J.~P. Ioannidis, B.~Patrick-Lake, S.~Morton, M.~Pencina,
  G.~Raman, J.~S. Ross, H.~P. Selker, R.~Varadhan, A.~Vickers, J.~B. Wong, and
  E.~W. Steyerberg.
\newblock The Predictive Approaches to Treatment effect Heterogeneity (PATH)
  Statement.
\newblock \emph{Annals of Internal Medicine}, 172\penalty0 (1):\penalty0
  35--45, 2020.

\bibitem[K{\"u}nzel et~al.(2019)K{\"u}nzel, Sekhon, Bickel, and
  Yu]{kunzel2019metalearners}
S.~R. K{\"u}nzel, J.~S. Sekhon, P.~J. Bickel, and B.~Yu.
\newblock Metalearners for estimating heterogeneous treatment effects using
  machine learning.
\newblock \emph{Proceedings of the National Academy of Sciences}, 116\penalty0
  (10):\penalty0 4156--4165, 2019.

\bibitem[Mammen(2012)]{Mammen12}
E.~Mammen.
\newblock \emph{When does bootstrap work?: Asymptotic results and simulations},
  volume~77.
\newblock Springer Science \& Business Media, 2012.

\bibitem[Manski(2004)]{manski2004statistical}
C.~F. Manski.
\newblock Statistical treatment rules for heterogeneous populations.
\newblock \emph{Econometrica}, 72\penalty0 (4):\penalty0 1221--1246, 2004.

\bibitem[Muntner et~al.(2018)Muntner, Carey, Gidding, Jones, Taler, Wright~Jr,
  and Whelton]{muntner2018potential}
P.~Muntner, R.~M. Carey, S.~Gidding, D.~W. Jones, S.~J. Taler, J.~T. Wright~Jr,
  and P.~K. Whelton.
\newblock Potential US population impact of the 2017 ACC/AHA high blood
  pressure guideline.
\newblock \emph{Circulation}, 137\penalty0 (2):\penalty0 109--118, 2018.

\bibitem[Nie and Wager(2021)]{nie2017quasi}
X.~Nie and S.~Wager.
\newblock Quasi-oracle estimation of heterogeneous treatment effects.
\newblock \emph{Biometrika}, 108\penalty0 (2):\penalty0 299--319, 2021.

\bibitem[Oikonomou et~al.(2022)Oikonomou, Spatz, Suchard, and
  Khera]{oikonomou2022individualising}
E.~K. Oikonomou, E.~S. Spatz, M.~A. Suchard, and R.~Khera.
\newblock Individualising intensive systolic blood pressure reduction in
  hypertension using computational trial phenomaps and machine learning: a
  post-hoc analysis of randomised clinical trials.
\newblock \emph{The Lancet Digital Health}, 4\penalty0 (11):\penalty0
  e796--e805, 2022.

\bibitem[Owolabi et~al.(2021)Owolabi, Thrift, Mahal, Ishida, Martins, Johnson,
  Pandian, Abd-Allah, Yaria, Phan, et~al.]{owolabi2021primary}
M.~O. Owolabi, A.~G. Thrift, A.~Mahal, M.~Ishida, S.~Martins, W.~D. Johnson,
  J.~Pandian, F.~Abd-Allah, J.~Yaria, H.~T. Phan, et~al.
\newblock Primary stroke prevention worldwide: translating evidence into
  action.
\newblock \emph{The Lancet Public Health}, 2021.

\bibitem[Powers et~al.(2019)Powers, Rabinstein, Ackerson, Adeoye, Bambakidis,
  Becker, Biller, Brown, Demaerschalk, Hoh, et~al.]{powers2019guidelines}
W.~J. Powers, A.~A. Rabinstein, T.~Ackerson, O.~M. Adeoye, N.~C. Bambakidis,
  K.~Becker, J.~Biller, M.~Brown, B.~M. Demaerschalk, B.~Hoh, et~al.
\newblock Guidelines for the early management of patients with acute ischemic
  stroke: 2019 update to the 2018 guidelines for the early management of acute
  ischemic stroke: a guideline for healthcare professionals from the American
  Heart Association/American Stroke Association.
\newblock \emph{Stroke}, 50\penalty0 (12):\penalty0 e344--e418, 2019.

\bibitem[Pr{\ae}stgaard and Wellner(1993)]{praestgaard1993exchangeably}
J.~Pr{\ae}stgaard and J.~A. Wellner.
\newblock Exchangeably weighted bootstraps of the general empirical process.
\newblock \emph{The Annals of Probability}, pages 2053--2086, 1993.

\bibitem[Psaty et~al.(2003)Psaty, Lumley, Furberg, Schellenbaum, Pahor,
  Alderman, and Weiss]{psaty2003health}
B.~M. Psaty, T.~Lumley, C.~D. Furberg, G.~Schellenbaum, M.~Pahor, M.~H.
  Alderman, and N.~S. Weiss.
\newblock Health outcomes associated with various antihypertensive therapies
  used as first-line agents: a network meta-analysis.
\newblock \emph{Journal of the American Medical Association}, 289\penalty0
  (19):\penalty0 2534--2544, 2003.

\bibitem[Radcliffe and Surry(1999)]{RadcliffeSu99}
N.~Radcliffe and P.~Surry.
\newblock Differential response analysis: Modeling true responses by isolating
  the effect of a single action.
\newblock \emph{Credit Scoring and Credit Control IV}, 1999.

\bibitem[Radcliffe(2007)]{Radcliffe07}
N.~J. Radcliffe.
\newblock Using control groups to target on predicted lift: Building and
  assessing uplift models.
\newblock \emph{Direct Marketing Analytics Journal}, 1\penalty0 (3):\penalty0
  14--21, 2007.

\bibitem[Rekkas et~al.(2020)Rekkas, Paulus, Raman, Wong, Steyerberg, Rijnbeek,
  Kent, and van Klaveren]{rekkas2020predictive}
A.~Rekkas, J.~K. Paulus, G.~Raman, J.~B. Wong, E.~W. Steyerberg, P.~R.
  Rijnbeek, D.~M. Kent, and D.~van Klaveren.
\newblock Predictive approaches to heterogeneous treatment effects: a scoping
  review.
\newblock \emph{BMC Medical Research Methodology}, 20\penalty0 (1):\penalty0
  1--12, 2020.

\bibitem[Robins et~al.(1994)Robins, Rotnitzky, and Zhao]{robins1994estimation}
J.~M. Robins, A.~Rotnitzky, and L.~P. Zhao.
\newblock Estimation of regression coefficients when some regressors are not
  always observed.
\newblock \emph{Journal of the American statistical Association}, 89\penalty0
  (427):\penalty0 846--866, 1994.

\bibitem[Rosenbaum and Rubin(1983)]{rosenbaum1983central}
P.~R. Rosenbaum and D.~B. Rubin.
\newblock The central role of the propensity score in observational studies for
  causal effects.
\newblock \emph{Biometrika}, 70\penalty0 (1):\penalty0 41--55, 1983.

\bibitem[Sandercock et~al.(2011)Sandercock, Niewada, and
  Cz{\l}onkowska]{sandercock2011international}
P.~A. Sandercock, M.~Niewada, and A.~Cz{\l}onkowska.
\newblock The international stroke trial database.
\newblock \emph{Trials}, 12\penalty0 (1):\penalty0 1--7, 2011.

\bibitem[Saposnik et~al.(2008)Saposnik, Cote, Phillips, Gubitz, Bayer, Minuk,
  and Black]{saposnik2008stroke}
G.~Saposnik, R.~Cote, S.~Phillips, G.~Gubitz, N.~Bayer, J.~Minuk, and S.~Black.
\newblock Stroke outcome in those over 80: a multicenter cohort study across
  Canada.
\newblock \emph{Stroke}, 39\penalty0 (8):\penalty0 2310--2317, 2008.

\bibitem[Semenova and Chernozhukov(2021)]{semenova2017debiased}
V.~Semenova and V.~Chernozhukov.
\newblock Debiased machine learning of conditional average treatment effects
  and other causal functions.
\newblock \emph{The Econometrics Journal}, 24\penalty0 (2):\penalty0 264--289,
  2021.

\bibitem[Shorack(1972)]{Shorack72}
G.~R. Shorack.
\newblock Functions of order statistics.
\newblock \emph{The Annals of Mathematical Statistics}, 43\penalty0
  (2):\penalty0 412--427, 1972.

\bibitem[Shorack and Wellner(2009)]{ShorackWe09}
G.~R. Shorack and J.~A. Wellner.
\newblock \emph{Empirical Processes with Applications to Statistics}.
\newblock SIAM, 2009.

\bibitem[Sun et~al.(2021)Sun, Du, and Wager]{sun2021treatment}
H.~Sun, S.~Du, and S.~Wager.
\newblock Treatment Allocation under Uncertain Costs.
\newblock \emph{arXiv preprint arXiv:2103.11066}, 2021.

\bibitem[Therneau and Grambsch(2000)]{therneau2000cox}
T.~M. Therneau and P.~M. Grambsch.
\newblock The Cox model.
\newblock In \emph{Modeling Survival Data: Extending the Cox Model}, pages
  39--77. Springer, 2000.

\bibitem[Tsiatis(2007)]{tsiatis2007semiparametric}
A.~Tsiatis.
\newblock \emph{Semiparametric theory and missing data}.
\newblock Springer Science \& Business Media, 2007.

\bibitem[Ullberg et~al.(2015)Ullberg, Zia, Petersson, and
  Norrving]{ullberg2015changes}
T.~Ullberg, E.~Zia, J.~Petersson, and B.~Norrving.
\newblock Changes in functional outcome over the first year after stroke: an
  observational study from the Swedish stroke register.
\newblock \emph{Stroke}, 46\penalty0 (2):\penalty0 389--394, 2015.

\bibitem[Van Der~Vaart and Wellner(1996)]{VanDerVaartWe96}
A.~W. Van Der~Vaart and J.~A. Wellner.
\newblock Weak Convergence.
\newblock In \emph{Weak Convergence and Empirical Processes}, pages 16--28.
  Springer, 1996.

\bibitem[Wager and Athey(2018)]{wager2018estimation}
S.~Wager and S.~Athey.
\newblock Estimation and inference of heterogeneous treatment effects using
  random forests.
\newblock \emph{Journal of the American Statistical Association}, 113\penalty0
  (523):\penalty0 1228--1242, 2018.

\bibitem[Whelton et~al.(2018)Whelton, Carey, Aronow, Casey, Collins,
  Dennison~Himmelfarb, DePalma, Gidding, Jamerson, Jones,
  et~al.]{whelton20182017}
P.~K. Whelton, R.~M. Carey, W.~S. Aronow, D.~E. Casey, K.~J. Collins,
  C.~Dennison~Himmelfarb, S.~M. DePalma, S.~Gidding, K.~A. Jamerson, D.~W.
  Jones, et~al.
\newblock 2017 ACC/AHA/AAPA/ABC/ACPM/AGS/APhA/ASH/ASPC/NMA/PCNA guideline for
  the prevention, detection, evaluation, and management of high blood pressure
  in adults: a report of the American College of Cardiology/American Heart
  Association Task Force on Clinical Practice Guidelines.
\newblock \emph{Journal of the American College of Cardiology}, 71\penalty0
  (19):\penalty0 e127--e248, 2018.

\bibitem[Xu et~al.(2023)Xu, Ignatiadis, Sverdrup, Fleming, Wager, and
  Shah]{xu2022treatment}
Y.~Xu, N.~Ignatiadis, E.~Sverdrup, S.~Fleming, S.~Wager, and N.~Shah.
\newblock Treatment heterogeneity with survival outcomes.
\newblock In \emph{Handbook of Matching and Weighting Adjustments for Causal
  Inference}, pages 445--482. Chapman and Hall/CRC, 2023.

\bibitem[Yadlowsky et~al.(2018)Yadlowsky, Hayward, Sussman, McClelland, Min,
  and Basu]{yadlowsky2018clinical}
S.~Yadlowsky, R.~A. Hayward, J.~B. Sussman, R.~L. McClelland, Y.-I. Min, and
  S.~Basu.
\newblock Clinical implications of revised pooled cohort equations for
  estimating atherosclerotic cardiovascular disease risk.
\newblock \emph{Annals of Internal Medicine}, 169\penalty0 (1):\penalty0
  20--29, 2018.

\bibitem[Zhao et~al.(2013)Zhao, Tian, Cai, Claggett, and Wei]{ZhaoTiCaClWe13}
L.~Zhao, L.~Tian, T.~Cai, B.~Claggett, and L.-J. Wei.
\newblock Effectively selecting a target population for a future comparative
  study.
\newblock \emph{Journal of the American Statistical Association}, 108\penalty0
  (502):\penalty0 527--539, 2013.

\bibitem[Zheng and van~der Laan(2011)]{zheng2011cross}
W.~Zheng and M.~J. van~der Laan.
\newblock Cross-validated targeted minimum-loss-based estimation.
\newblock In \emph{Targeted Learning}, pages 459--474. Springer, 2011.

\bibitem[Zweig and Campbell(1993)]{zweig1993receiver}
M.~H. Zweig and G.~Campbell.
\newblock Receiver-operating characteristic (ROC) plots: a fundamental
  evaluation tool in clinical medicine.
\newblock \emph{Clinical Chemistry}, 39\penalty0 (4):\penalty0 561--577, 1993.

\bibitem[van~der Vaart(1998)]{VanDerVaart98}
A.~W. van~der Vaart.
\newblock \emph{Asymptotic Statistics}.
\newblock Cambridge Series in Statistical and Probabilistic Mathematics.
  Cambridge University Press, 1998.

\end{thebibliography}
}

\end{document}